%% file: main_CSBM_matching_final.tex
\documentclass{article}
\usepackage{microtype}
\usepackage{graphicx}
\usepackage{booktabs} 
\usepackage{adjustbox}
\usepackage{comment}
\usepackage{hyperref}

\usepackage[accepted]{icml2023}
\usepackage{amsmath}
\usepackage{amssymb}
\usepackage{mathtools}
\usepackage{diagbox}
\usepackage{amsthm}
\usepackage{bbm}
\usepackage[capitalize,noabbrev]{cleveref}
\usepackage{multirow}
\usepackage{dsfont}
\usepackage{fancyhdr}
\usepackage{graphicx}
\usepackage{fancybox}
\usepackage{makecell}
\usepackage{float}
\usepackage{color}
\usepackage{subfig}
\usepackage{enumitem}
\theoremstyle{plain}
\newtheorem{theorem}{Theorem}[section]
\numberwithin{theorem}{section}
\newtheorem{proposition}[theorem]{Proposition}
\newtheorem{lemma}[theorem]{Lemma}
\newtheorem{corollary}[theorem]{Corollary}
\theoremstyle{definition}
\newtheorem{definition}[theorem]{Definition}

\newtheorem{remark}[theorem]{Remark}

\numberwithin{equation}{section}

\renewcommand{\P} {\mathbb{P}}

\DeclareMathOperator*{\argmax}{arg\,max}

\newcommand{\beq}{ \begin{equation} }
	\newcommand{\eeq}{ \end{equation} }

\newcommand{\caE}{{\mathcal E}}

\newcommand{\caG}{{\mathcal G}}

\newcommand{\caN}{{\mathcal N}}

\newcommand{\caS}{{\mathcal S}}

\newcommand{\bbR}{{\mathbb R}}

\newcommand{\bss}{{\boldsymbol s}}
\newcommand{\bst}{{\boldsymbol t}}

\newcommand{\sign}{{\mathbf{Sign}}}
\renewcommand{\algorithmicrequire}{\textbf{Input:}}
\renewcommand{\algorithmicensure}{\textbf{Output:}}
\usepackage[textsize=tiny]{todonotes}

\icmltitlerunning{Efficient Algorithms for Exact Graph Matching on
           Correlated Stochastic Block Models with Constant Correlation}

\begin{document}

\twocolumn[
\icmltitle{Efficient Algorithms for Exact Graph Matching on\\
           Correlated Stochastic Block Models with Constant Correlation}



\icmlsetsymbol{equal}{*}

\begin{icmlauthorlist}
\icmlauthor{Joonhyuk Yang}{equal,sch}
\icmlauthor{Dongpil Shin}{equal,sch}
\icmlauthor{Hye Won Chung}{sch}
\end{icmlauthorlist}

\icmlaffiliation{sch}{School of Electrical Engineering, KAIST, Daejeon, Korea}

\icmlcorrespondingauthor{Hye Won Chung}{hwchung@kaist.ac.kr}

\icmlkeywords{Machine Learning, ICML}

\vskip 0.3in
]



\printAffiliationsAndNotice{\icmlEqualContribution} 

\input{section/abstract}

\input{section/introduction_rev1}

\input{section/algorithm_result_rev1}

\input{section/proof_rev1}
\input{section/experiment_rev1}

\input{section/discussion_rev2}
\input{section/acknowledgements}

 	\newpage
\bibliography{main_CSBM_matching_final}
\bibliographystyle{icml2023}

    \newpage
\appendix
    \onecolumn
    \input{section/supp_rev5}



\end{document}

%% file: section/abstract.tex
\begin{abstract}
We consider the problem of graph matching, or learning vertex correspondence, between two correlated stochastic block models (SBMs).
The graph matching problem arises in various fields, including computer vision, natural language processing and bioinformatics, and in particular, matching graphs with inherent community structure has significance related to de-anonymization of correlated social networks. 
Compared to the correlated Erd\H{o}s-R\'enyi (ER) model, where various efficient algorithms have been developed, among which a few algorithms have been proven to achieve the exact matching with constant edge correlation, no low-order polynomial algorithm has been known to achieve  exact matching for the correlated SBMs with constant correlation. In this work, we propose an efficient algorithm for matching graphs with community structure, based on the comparison between partition trees rooted from each vertex, by extending the idea of \citet{MRT21a} to graphs with communities. The partition tree divides the large neighborhoods of each vertex into disjoint subsets using their edge statistics to different communities. Our algorithm is the first low-order polynomial-time  algorithm achieving  exact matching between two correlated SBMs with high probability in dense graphs. 
\end{abstract}

%% file: section/introduction_rev1.tex
\section{Introduction}\label{sec:introduction}

	Graph matching aims to align two (or more) graphs to reveal a bijection between the vertex sets such that the number of aligned edges is maximized. Given two graphs with $n$ vertices, graph matching finds a solution for a quadratic assignment problem (QAP), $\max _{\Pi \in S_{n}} \langle A, \Pi B \Pi^{\top} \rangle$ over the set of all $n \times n$ permutation matrices $S_{n}$, where $A$ and $B$ denote the adjacency matrices of the two graphs.
	This problem has been studied in various fields, including social network analysis \citep{NS09}, computer vision \citep{SS05}, pattern recognition \citep{CFSV04}, natural language processing \citep{HNM05}, machine learning \citep{LQ12}, and bioinformatics \citep{KHGM16,chen2006detecting}. 

Although the QAP is known to be NP-hard in the worst case  \citep{BCPP98}, the graph matching can be solved in polynomial time for the average case of random graph models. Thus, many previous works have studied the graph matching for random graph models, especially for Erd\H{o}s-R\'enyi (ER) graphs. 
In particular, the correlated ER graph model proposed by \citet{PG11} has been widely studied. In this model, there exists a parent ER graph $G_0\sim \mathcal{G}\left(n, p/(1-\alpha)\right)$, and two subgraphs $G$ and $G'$ are obtained by independent subsampling of $G_0$, where each edge of $G_0$ is removed independently with probability $\alpha$. The parameter $\alpha$ indicates the noise level, and $1-\alpha$ is the correlation between $G$ and $G'$. Assuming a permutation $\pi:[n]\to[n]$ and denoting by $G^\pi$ the graph obtained by permuting the vertices of $G$ with $\pi$, exact graph matching aims to recover $\pi$ from the two graphs $G^\pi,G'\sim  \mathcal{G}(n, p)$.

For the correlated ER model, many works have focused on two fundamental questions: 1) deriving the information-theoretic limit on $(n,p,\alpha)$ for exact matching, and 2) developing polynomial-time algorithms for recovering $\pi$.
Regarding the first question, it was shown in \citep{cullina2016improved,wu2022settling} that exact matching is achievable if $np(1-\alpha)\geq (1+\epsilon)\log n$ for any $\epsilon>0$ where $p/(1-\alpha)=o(1)$. This limit implies that graph matching is information-theoretically possible for constant $\alpha$ if $np=\Theta(\log n)$ and even for $\alpha$ close to 1 if $np\gg \log n$. However, achieving this limit with polynomial-time algorithms is still open, although various efficient algorithms have been proposed \cite{DWMX21,pmlr-v119-fan20a,MRT21b}. Most of these algorithms require the noise level $\alpha$ to be arbitrarily close to 0 in order to guarantee exact matching. Some recent work has developed polynomial-time algorithms for constant correlation in the spare regime \citep{MRT21a} or in the sparse and dense regimes but with $\alpha$ less than some specific constant \citep{MWXY22}. 

While this previous line of work has revealed several important aspects in matching correlated random graphs, it is still largely open how these results generalize to more practical random graph models beyond the Erd\H{o}s-R\'enyi (ER) model. 
In particular, matching graphs with inherent community structure has relevance to real-world applications such as de-anonymization of social networks, but the investigation of efficient graph-matching algorithms for the graphs with community structure has been largely unexplored. 

In this paper, we consider exact graph matching for random graphs with community structure, and develop a polynomial-time matching algorithm for constant edge-correlation. We focus on the correlated stochastic block models (SBMs), where the parent graph $G_0$ is assumed to be sampled from the SBM, which is known to be one of the most natural generative models for networks with community structure. 
To the best of our knowledge, our algorithm is the first low-order polynomial-time algorithm that guarantees exact matching of the correlated SBMs with constant correlation.

	\subsection{Correlated Stochastic Block Models}\label{sec:model}
	Consider an undirected graph of $n$ vertices with a planted partition. 
	Suppose the vertex set $[n]$ is partitioned into disjoint subsets of $k\geq 1$ communities, $C_1, C_2,\dots, C_k$, where the number of vertices in community $C_i$ is $|C_i|=n_i$ and $\sum_{i=1}^k n_i=n$. Without loss of generality, we assume that $n_1\geq n_2\geq \cdots \geq n_k:=n_{\min}$.
	Given the partition $\{C_i\}_{i=1}^k$, the correlated stochastic block models \citep{onaran2016optimal} are parameterized by $p,q \in [0,1]$, $p>q$, and $\alpha \in [0,1)$.
	We assume that there exists a parent graph $G_0$ with the given partition $\{C_i\}_{i=1}^k$, where the edges between each pair of vertices are drawn independently as follows: $u\in C_i$ and $v\in C_j$ are connected with probability $p/(1-\alpha)$ if $i=j\in[k]$ and with probability $q/(1-\alpha)$ otherwise. Then two subgraphs $G$ and $G'$ are obtained by independent subsampling of $G_0$ as follows: $G$ is obtained by removing each edge of $G_0$ independently with probability $\alpha$, and $G'$ is obtained independently in the same way as $G$. Note that
$
		\P\{(i,j)\in \mathcal{E}(G) |  
			(i,j)\in \mathcal{E}(G') \} 
			= 1-\alpha
$ where $\mathcal{E}(G)$ is the edge set of the graph $G$.
	 Given a permutation $\pi:[n]\to[n]$, $G^\pi$ is the graph obtained by permuting the vertices of $G$ by $\pi$.
Our goal is to design a polynomial-time algorithm that can exactly recover the permutation $\pi$ at the constant noise level $\alpha$. We will prove the performance guarantees of the proposed algorithm in terms of the parameters $(p,q, \alpha, k, n_{\min})$, both for the cases with and without the knowledge of the exact community structure $\{C_i\}_{i=1}^k$.

	\subsection{Prior Work and Main Question}
	    \begin{table*}[t]
    \centering
    \caption{Comparison with literature for the exact matching of the correlated SBMs with $k$-balanced communities of size $m$.}
    \label{tbl:previous algorithm}
   \small{    \begin{tabular}{c|c|c|c}
    \toprule
      Algorithm & Density & Noise level & Time complexity   \\ \midrule
      Black swan 
      \citep{BCL18+} & $ m p \geq m^{o(1)}$ & $\alpha \leq 1-(\log m)^{-o(1)}$ & $k\times m^{O(\log m)}$\\
  Degree profile 
  \citep{DWMX21} & $mp \geq (\log m)^{C}$ & $\alpha \leq (\log m)^{-C}$ & $k\times {O}(m^{3}p^{2}+m^{2.5})$ \\
  GRAMPA 
  \citep{FMWX19a} &$mp\geq (\log m)^C$ &$\alpha \leq (\log m)^{-C} $ & $k\times O(m^3)$ \\
     Binary tree 
     \citep{MRT21a} & $ (1+\epsilon) \log m \leq mp\leq m^{{1}/{(C \log \log m)}}$ & $\alpha \leq \min\{\text{const.}, \epsilon /4\} $ & $k\times m^{2+o(1)}$   \\ 
     
     Chandelier 
      \citep{MWXY22}& $(1+\epsilon) \log m \leq mp(1-\alpha)$ & $\alpha^{2} \leq  1-\sqrt{0.338} $ & $\geq k\times m^{25}$ \\ 
     \midrule
    \multirow{2}{*}{Our result ($2^{k'}$-ary partition tree)} & $(\log m)^C \leq mp \leq m^{1/20}$  & \multirow{2}{*}{$\alpha \leq \text{const.}$ } & \multirow{2}{*}{$k\times O(m^{4}p) $}  \\
    & $mq=\Omega\left((\log \log m)^{2}\right)$ & \\  \bottomrule
    \end{tabular}}
    \end{table*}

	Graph matching on correlated stochastic block models has been studied from two different aspects: first, when the community structure is revealed in both $G^\pi$ and $G'$, what are the information-theoretically feasible regimes where exact matching is possible so that the vertex identities can be de-anonymized; second, when the goal is to recover the community structure of the original graph $G_0$ from the subsampled graph $G'$, what is the regime where having $G^\pi$ as side information can be beneficial? The first aspect was studied in \citet{onaran2016optimal,CSKM16} for the two-community case ($k=2$), and the sufficient condition for the optimal maximum likelihood estimator to recover $\pi$ with high probability is shown to be $(1-\alpha) (p+q)n/2>2(\log n)$ when $n_1=n_2=n/2$.
The second aspect was studied in \cite{RS21} for the case of two equal-sized communities, and it was shown that even if the communities of $G^\pi$ and $G'$ are not revealed, the maximum likelihood estimator can recover $\pi$ when $(1-\alpha) (p+q)n/2>(\log n)$, which matches the impossibility results for the exact matching proved in \citep{CSKM16}. Note that $(1-\alpha) (p+q)n/2>(\log n)$ is the necessary condition to guarantee that there are no isolated vertices in $G\cap G'$ with high probability. The result of \citet{RS21} implies that there exist regimes of $(p,q,\alpha)$ where having the correlated graph $G^\pi$ can bring an information advantage in recovering the community of $G'$ (or $G_0$) through the graph matching between $G^\pi$ and $G'$.

These previous works establish the information-theoretic thresholds for exact graph matching of the correlated SBMs with two communities, and in particular, it is shown that when $p=\Theta(q)=\Omega(\log n/n)$, graph matching is information-theoretically possible with constant correlation $(1-\alpha)$, even if the community structure is not revealed in both graphs. However, no polynomial-time algorithm has been proven to guarantee exact graph matching for the correlated SBMs with constant correlation. 

When the partition of vertices based on their community labels is revealed in both $G^\pi$ and $G'$, it is possible to apply the graph matching algorithm for the correlated ER model to each of the communities, since the subgraphs induced by $C_i$ in $G^\pi$ and $G'$ are the correlated ER graphs with a parent graph $G_0(C_i)\sim \mathcal{G}(n_i, p/(1-\alpha))$.
In Table \ref{tbl:previous algorithm}, we summarize the results of the previous algorithms, originally developed for the correlated ER model, when applied to each recovered community $(C_1,\cdots,C_k)$ of the correlated SBMs. Here we simply assume that all the communities have the same size of $m:=n/k$.

The `Black swan' algorithm \citep{BCL18+} allows the largest noise level of $\alpha=1-(\log m)^{-o(1)}$, but it requires quasi-polynomial time complexity $m^{\Theta(\log m)}$. This algorithm defines a rich family of rare subgraphs called `black swans', each consisting of $\Theta(\log m)$ vertices, and finds the rare subgraphs that appear in both graphs with $m^{\Theta(\log m)}$ time complexity. The `Chandelier' \citep{MWXY22} algorithm also uses the idea of `subgraph counting', but with a family of unbalanced rooted trees of size $\Theta(\log m)$, called `chandeliers', which can be counted approximately in polynomial time. This algorithm succeeds in graph matching for both sparse and dense regimes with constant correlation,
but the time complexity scales in high-order polynomials, $\Omega(m^{25})$.

Among computationally more efficient methods with time complexity $O(m^3)$, the only algorithm that has been proven to guarantee exact matching with constant correlation is the `Binary (partition) tree' algorithm by \citet{MRT21a}.
The main idea is to define a signature vector of dimension $\Theta(\log m)$ associated with each vertex $i\in[m]$ by using the degree information of a large neighborhood of each $i$. By measuring the similarity between the signature vectors, the algorithm recovers the exact match for $m$ vertices.
The signature vector is defined by constructing a depth-$\ell$ binary partition tree rooted at each vertex $i\in[m]$, where at each depth $r\in[\ell]$ the neighborhood around the vertex $i$ with radius-$r$ is partitioned into $2^{r}$ disjoint subsets. To define a signature vector of dimension $2^\ell=\Theta(\log m)$, the depth of the binary partition tree needs to be $\ell=\Theta(\log\log m)$. Thus, the graph should be sparse as $mp\leq m^{1/(C\log\log m)}$ to avoid generating a loop in the partition tree, and this limits the application of this algorithm in the denser regime. 

The main open question is how to construct an efficient matching algorithm for the correlated SBMs in the dense regime. 
We focus on the fact that the edges between communities have not been exploited in the previous algorithms. By using the correlation of both intra- and inter-community edges, as well as the known community structure, we develop a matching algorithm using the similarity between the signature vectors of the correct pairs, as in \citet{MRT21a}, but by constructing $2^{k'}$-ary  (rather than binary) partition trees, which effectively reduces the required depth of the tree and thus successfully operates in the dense regime. 

	

\subsection{Main Results}
	
	Our main contribution is the development of a low-order polynomial-time algorithm for the exact matching of the correlated SBMs 
 with constant correlation in the dense regime.
	\begin{theorem}[Exact matching for the correlated SBMs with known community structure]\label{thm:main_dense}
There exist absolute constants $\alpha_1, M, M'>0$ with the properties below. 
Consider the two graphs $G^\pi$ and $G'$, which are generated from the correlated SBMs defined in Sec. \ref{sec:model}, 
with the underlying permutation $\pi$ and correlation $1-\alpha$. Suppose that the community labels $(C_1,C_2,\ldots,C_k)$ are given in both graphs, and assume that $C_k$ is the smallest community of size $n_{\min}$. Assume that $\alpha\in (0,\alpha_1)$, $n_{\min} =\Omega(n^{10/19})$ and 
 \begin{equation}\label{eqn:main_thm_cond1}
			(\log n_{\min})^{1.1} \leq n_{\min}p \leq n_{\min}^{1/20},
\end{equation} 
\begin{equation}\label{eqn:main_thm_cond2}
    n_{\min}q\geq M' (\log \log n_{\min})^{2} ,
\end{equation}
\begin{equation}\label{eqn:main_thm_cond3}
    k \geq \left(\frac{M(\log \log n_{\min}) \cdot (\log n_{\min}p)}{\log n_{\min}} \vee 3\right).
\end{equation}    
Then, there exists a polynomial-time algorithm that recovers $\pi$ exactly with high probability as $n\to\infty$.  The complexity of the algorithm scales as $O(km^4 p)$ for the balanced communities of size $m:=n/k$.
	\end{theorem}	
Even if the community structure is not revealed in both graphs, we can apply our matching algorithm by imposing an additional assumption on $n_{\min}$, which is sufficient for a polynomial-time algorithm to recover the community structure in both graphs. We combine Theorem \ref{thm:main_dense} with the result of \citet{yan2018provable}, which proposed a semidefinite programming (SDP) for community detection with provable guarantees.
	\begin{corollary}[Exact matching for the correlated SBMs with unknown community structure]\label{cor:main_all}
There exist absolute constants $\alpha_1,M,M’,M_1>0$ with with the properties below. 
Consider the two graphs $G^\pi$ and $G'$, which are generated from the correlated SBMs defined in Sec. \ref{sec:model}, 
with the underlying permutation $\pi$ and correlation $1-\alpha$. Assume that the sizes of all communities are different and the smallest community size is $n_{\min}$. Assume that $\alpha\in(0,\alpha_1/2)$ and \eqref{eqn:main_thm_cond1}, \eqref{eqn:main_thm_cond2} and \eqref{eqn:main_thm_cond3} hold. Also assume that
\begin{equation}\label{eq:main cor dense nmin}
    n_{\min} \geq M_{1}\left( \frac{\sqrt{np}}{p-q} \vee \frac{p\log n}{(p-q)^{2}}\right).
\end{equation}
    Then, there exists a polynomial-time algorithm that recovers $\pi$ exactly with high probability as $n\to\infty$. 
	\end{corollary}
	
	
\begin{remark}\label{rmk:balance}
   The assumption in Corollary \ref{cor:main_all}  that community sizes are unique can be satisfied with high probability even if the community label of each vertex is sampled from a distribution, as long as each  label is sampled with probability $\Theta(1/k)$. We prove this in Appendix  \ref{sec:balanced_community_size}. 

\end{remark}

	\subsection{Notation}
Asymptotic dependencies are denoted by the standard notations $O(\cdot), o(\cdot), \Omega(\cdot), \omega(\cdot), \Theta(\cdot)$ with $n\to\infty$. 
	For $n\in \mathbb{N^+}$, we denote the set $[n]:=\{1, 2, \ldots, n\}$. 
	Let $\wedge$ denote the minimum operator and $\vee$ denote the maximum operator.
	For a real number $x$, $\sign(x)=1$ if $x\geq 0$ and $\sign(x)=-1$ if $x<0$. 	
	For a graph $G$, we define $\caN_G(i)$ as the set of neighbors of vertex $i$ in $G$, and $\caN_G(S)$ as the union of $\caN_G(i)$ over all $i$ in set $S \subset [n]$.
For $d\in \mathbb{N^+}$,  $\mathcal{B}_{G}(i, d)$ ($\mathcal{S}_{G}(i, d)$) denotes the ball (sphere) of radius $d$ centered at vertex $i$ in graph $G$. 
	For a subset $R\subset [n] $ and a permutation $\pi:[n]\to[n]$, $G(R)$ denotes the subgraph of $G$ induced by $R$, and $\pi |_{R}$ denotes the permutation over $R$ defined by $\pi$. 
	Let $\deg_{G}(i)$ denote the degree of $i$ in a graph $G$.  
	For a graph $G$ with communities $(C_1,C_2,\dots, C_k)$, $\operatorname{deg}^{a}_{G}(i)$ denotes the size of neighbors of $i$ in the community $C_{a}$. 

%% file: section/algorithm_result_rev1.tex
\section{Algorithm and Results}\label{sec:alg_result}

\subsection{Overview of Algorithm}
Consider the correlated SBMs $(G^\pi,G')$ with correlation $1-\alpha$. Suppose that the community structure $(C_1,C_2,\dots, C_k)$ is revealed in both graphs, and $|C_1|\geq |C_2|\geq \dots  \geq |C_k|=n_{\min}$. 
Our algorithm consists of three stages.
In the first stage, we focus on the minimum-size community $C_k$ and generate a signature vector for each vertex in $G^\pi(C_k)$ and $G'(C_k)$. In this stage, we use both the edges within the community $C_k$ and the edges across $C_k$ and other (randomly sampled) $k'<k-1$ communities, say $C_1,\dots, C_{k'}$. 
Each vertex $i\in G^\pi(C_k)$ and $j\in G'(C_k)$ is associated with a signature vector $f(i)\in \bbR^{2^{k'\ell}}$ and $f'(j)\in \bbR^{2^{k'\ell}}$, respectively, where the $(2^{k'\cdot \ell})$-dimensional signature is obtained from a depth-$\ell$  $(2^{k'})$-ary partition tree rooted from $i$ ($j$), which partitions the vertices in the sphere $\caS_{G^\pi(G_k)}(i,r)$ ($\caS_{G'(G_k)}(j,r)$) at each depth $r\in[\ell]$ into $2^{k'r}$ disjoint subsets. The detailed steps of generating the partition trees and the signature vectors are explained in Sec. \ref{sec:partition_tree}. We will choose $(k',\ell)$ such that $k'\cdot \ell=\Theta(\log\log n_{\min} )$ so that the signature contains enough information to distinguish and match $n_{\min}$ vertices in $G^\pi (G_k)$ and $G'(G_k)$. The maximum depth $\ell$ of the partition tree needs to satisfy $(n_{\min} p)^\ell=O( n_{\min}) $ to guarantee that the partition tree does not contain a loop with high probability, and this gives a condition on $k'$ such that $k'\geq\frac{C(\log\log n_{\min})\cdot(\log n_{\min} p)}{\log n_{\min}}$ for our algorithm to operate successfully. 

In the second stage, we compare the signature vectors $f(i)$ and $f'(j)$ for each $(i,j)$ pair in $G^\pi(C_k)$ and $G'(C_k)$, and find a pair whose ``normalized distance" is less than some threshold. This gives an almost exact matching for the vertices in $C_k$, and a subsequent refinement step gives the exact matching for $n_{\min}$ vertices.

In the last stage, we use the matched vertex pairs in $G^\pi(C_k)$ and $G'(C_k)$ as initial seeds, and apply a seeded graph matching algorithm multiple times to match the rest of the vertices. 
In the following, we explain each of the three stages in more detail.


\subsection{Partition Trees using Community Structure}\label{sec:partition_tree}
  \begin{algorithm}[t]
    \caption{Vertex Signature Using Community Structure}
    \label{alg:algorithm1}
\begin{algorithmic}[1]
  \REQUIRE{a graph $\Gamma$ with community structure $(C_1, \dots, C_k)$ where $|C_k|=n_k$, parameters $k', \ell\in \mathbb{N}$, and $i\in C_{k}$}
    \ENSURE{$f(i) \in \mathbb{R}^{2^{k'\ell}}$ and $\mathrm{v}(i) \in \mathbb{R}^{2^{k'\ell}}$}
    \STATE \text{$T_{\varnothing}^{0} \leftarrow\{i\}$}
    \FOR{$r=0, \ldots, \ell-1$}
    \FOR{$\bss^{r} \in\{-1,1\}^{k'r}$ and $\bss_{r+1} \in \{-1,1\}^{k'}$}
    \STATE  $T_{(\bss^{r},\bss_{r+1})}^{r+1}(i) \leftarrow $
    
    $\left\{ j \in \mathcal{N}_{\Gamma(C_{k})}\left(T_{\bss^{r}}^{r}(i)\right) \cap \mathcal{S}_{\Gamma(C_{k})}(i, r+1): \right.$

 $\left. \sign \left( \deg^{a}_{\Gamma}(j) - n_{a} q\right) = \bss_{r+1}(a) \text{ for } a\in[k'] \right\}$
    \ENDFOR
    \ENDFOR
    \STATE Define $f(i),\mathrm{v}(i) \in \mathbb{R}^{2^{k' \ell}}$ by 

 $ f_{\bss^\ell}(i):=  \sum_{j \in T_{\bss^\ell}^{\ell}(i)} \left(\operatorname{deg}^{k}_{\Gamma}(j)-1- n_{k}p\right)$, 
 
      $\mathrm{v}_{\bss^\ell}(i):=n_{k} p(1-p)\left|T_{\bss^\ell}^{\ell}(i) \right| $ for $\bss^\ell \in\{-1,1\}^{k'\ell}$. 
     \textbf{return} $f(i) \text{ and } \mathrm{v}(i)$  
\end{algorithmic}
\end{algorithm}
   
In this section, we explain the process of constructing a partition tree for each vertex $i\in C_k$ in a graph $\Gamma\sim\text{SBM}(n,p,q)$ using a known community structure $(C_1,\dots, C_k)$ for $[n]$. 
This partition tree is used to define a signature vector $f(i)$ for each $i\in C_k$.
We randomly select $k'<k-1$ communities and denote them by $C_1,\dots, C_{k'}$ for notational simplicity. 
For a given vertex $i\in C_k$, the partition tree starts from a root $T^0_{\varnothing}:=\{i\}$, and at each level $r=1,\dots,\ell$ the tree partitions the vertices in $\caS_{\Gamma(C_k)}(i,r)$ into $2^{k'r}$ disjoint subsets, denoted by nodes $T^r_{(\bss_1,\dots,\bss_r)}$ for $(\bss_1,\dots,\bss_r)\in\{-1,1\}^{k'r}$. For notational simplicity, let $\bss^r:=(\bss_1,\dots,\bss_r)$. For every $r<\ell$, the node $T^r_{\bss^r}$ has $2^{k'}$ children  $T^{r+1}_{(\bss^r,\bss_{r+1})}$, $\bss_{r+1}\in\{-1,1\}^{k'}$, where $T^{r+1}_{(\bss^r,\bss_{r+1})}$ contains all vertices $j\in\caN_{\Gamma(C_k)}(T^r_{\bss^r})\cap \caS_{\Gamma(C_k)}(i,r+1) $ satisfying $\sign(\deg^a_{\Gamma}(j)-n_aq)=\bss_{r+1}(a)$ for each $a\in[k']$, where $\bss_{r+1}(a)$ is the $a$-th entry of $\bss_{r+1}\in\{-1,1\}^{k'}$. In other words, we check whether the degree of vertex $j$ to each community $C_a$, $a\in[k']$, is greater than or equal to the average degree $n_a q$, encode this information by $\bss_{r+1}\in\{-1,1\}^{k'}$, and assign the vertex $j$ to the corresponding node set $T^{r+1}_{(\bss^r,\bss_{r+1})}$ in the partition tree.
After constructing the partition tree of depth-$\ell$ rooted at $i$, the signature vector $f(i)\in \mathbb{R}^{2^{k'\ell}}$ is generated based on the degrees of the vertices in the leaves $T^\ell_{\bss^\ell}$, $\bss^\ell\in\{-1,1\}^{k'\ell}$.
We also define a variance vector $\mathrm{v}(i)\in \mathbb{R}^{2^{k'\ell}}$ that stores the variances of each entry of the signature vector $f(i)$.

Algorithm \ref{alg:algorithm1} describes the process of inductively constructing the partition tree and generating the signature vector.

\subsection{Exact Matching for the Smallest Community}

The second stage aims to match all the correct vertex pairs in $G^\pi(C_k)$ and $G'(C_k)$ by computing and comparing $(f(i), f'(j))$ for all possible pairs  $i\in G^\pi(C_k)$ and $j\in G'(C_k)$. 
For this stage, we follow the two-step procedures, designed by \citet{MRT21a}, where the first-step finds a rough estimate of the permutation over $[n_k]$ vertices by comparing the similarity between the signature vectors, and the second step refines the output of the first step to recover the exact permutation. 
In the first step, every pair of $i\in G^\pi(C_k)$ and $j\in G'(C_k)$ is compared in terms of the normalized distance $\sum_{\bss^\ell \in J} \frac{\left(f_{\bss^\ell}(i)-f_{\bss^\ell}^{\prime}\left(j\right)\right)^2}{\mathrm{v}_{\bss^\ell}(i)+\mathrm{v}_{\bss^\ell}^{\prime}\left(j\right)}$ where $J$ is a random subset uniformly drawn from $\{-1,1\}^{k'\ell}$ with poly-logarithmic cardinality in $n_{k}$. If there exists $(i,j)$ whose distance is less than the threshold $|J|(1-\frac{1}{\sqrt{\log n_k}})$ then the pair is matched. To resolve the case where more than one vertices $j\in G'(C_k)$ are matched to the same vertex $i\in G^\pi(C_k)$ or no vertex is matched to $i\in G^\pi(C_k)$, a clean-up step is applied to generate a permutation $\tilde{\pi}_k: [n_k]\to [n_k]$ over $C_k$. 
Algorithm \ref{alg:algorithm2} describes the process of obtaining an initial estimate $\tilde{\pi}_k$ of the permutation. 
Given the estimate $\tilde{\pi}_k$, one can apply the refinement step to obtain the permutation $\hat{\pi}_k:[n_k]\to [n_k]$ equal to the true one $\pi|_{C_k}$.
Due to space limitations, we present the refinement algorithm (Algorithm \ref{alg:algorithm4} \citep{MRT21a}) in Appendix \S\ref{app:sec:exact}.

\begin{algorithm}[t]
    \caption{Almost Exact Matching \citep{MRT21a}}\label{alg:algorithm2}
    \begin{algorithmic}[1]
      \REQUIRE{two graphs $\Gamma$ and $\Gamma'$ with community structure $(C_1, \dots, C_k)$ where $|C_k|=n_k$}
       \ENSURE{a permutation $\tilde{\pi}_k:[n_k]\to [n_k]$}
        \STATE $\ell \leftarrow \left\lceil \frac{\log n_k}{40 \log n_kp}\right\rceil \wedge \lceil 42\log \log n_k \rceil$
        \STATE $w \leftarrow $ $\lfloor (\log n_{k})^{5} \rfloor, \; k' \leftarrow \lceil1680 (\log \log n_{k}) \frac{\log n_{k}p}{\log n_{k}}\rceil$
        \FOR{$i\in C_{k}$}
        \STATE $ \left(f(i),\mathrm{v}(i)\right) \leftarrow \text{Algorithm \ref{alg:algorithm1}} (\Gamma,k',\ell, i)$
        \STATE $ \left(f'(i),\mathrm{v}'(i)\right) \leftarrow \text{Algorithm \ref{alg:algorithm1}}(\Gamma',k',\ell, i)$
        \ENDFOR
        \STATE $J \leftarrow$ a random subset uniformly drawn  from $\{-1,1\}^{k'\ell}$ with cardinality $2w$,
        \FOR{$i\in \Gamma(C_{k})$ and $j\in \Gamma'(C_{k})$}
        \IF{ $\sum\limits_{\bss^\ell\in J} \frac{(f_{\bss^\ell}(i)-f'_{\bss^\ell}(j))^{2}}{\mathrm{v}_{\bss^\ell}(i)+\mathrm{v}'_{\bss^\ell}(j)}<2w\left(1-\frac{1}{\sqrt{\log n_{k}}}\right)$}
        \STATE $ B_{i,j}\leftarrow 1$
        \ELSE
        \STATE $ B_{i,j}\leftarrow 0$
        \ENDIF
        \ENDFOR
        \STATE $H \leftarrow$ the bipartite graph with adjacency matrix $B$
        \STATE let $V=V'=[n_k]$ 
        \WHILE{$H$ has at least one edge}
        \STATE choose a random edge $i \sim j$ in $H$ for $(i,j) \in V\times V'$ 
        \STATE define $\tilde{\pi}\left(j\right):=i$ and remove the edge $i \sim j$ from $H$
        \STATE $V \leftarrow V \backslash\{i\}$, $ V^{\prime} \leftarrow V^{\prime} \backslash\left\{j\right\}$
        \ENDWHILE
        \IF{ \text{$V \neq \varnothing$} }
        \STATE define $\left.\tilde{\pi}\right|_V$ as an arbitrary bijection from $V$ to $V^{\prime}$
        \ENDIF
        \STATE \textbf{return} $\tilde{\pi}_k$
    \end{algorithmic}
\end{algorithm}

\begin{theorem}[Almost Exact Matching on $C_k$]\label{thm:almost_smallest}
  For any constant $D>0$ there exist constants $n_{0}, c$ depending on $D$ and absolute constants $\alpha_1,M,M'$ with the properties below. 
  Consider the two graphs $G^\pi$ and $G'$, generated from the correlated SBMs defined in Sec. \ref{sec:model}, 
  with the underlying permutation $\pi:[n]\to[n]$. Suppose that community labels are given in both graphs, and  $C_k$ is the smallest community of size $n_{\min}$. Assume that $\alpha \in (0,\alpha_1)$, $n_{\min} \geq n_{0}$, \eqref{eqn:main_thm_cond1}, \eqref{eqn:main_thm_cond2} and \eqref{eqn:main_thm_cond3} hold.
		Then, Algorithm \ref{alg:algorithm2} applied to the inputs $(G^{\pi},G')$ with the known community structure in both graphs yields the permutation $\tilde{\pi}_k$ over $C_k$ such that
			\begin{equation}
	    |i\in C_k : \tilde{\pi}_k(i) \neq \pi(i)| \leq4 {n^{1-c}_{\min}}.
	\end{equation}
	 with probability at least $1-n_{\min}^{-D}$. 
\end{theorem}

By combining Theorem \ref{thm:almost_smallest} with the result of \citet{MRT21a}, where it was shown that the exact matching of the correlated ER model can be obtained from the almost exact matching under some mild conditions on the graphs, which can already be satisfied by the conditions in Theorem \ref{thm:almost_smallest}, we can obtain the result for the exact matching. 
\begin{corollary}[Exact matching on $C_k$]\label{cor:exact_smallest} 
Assume the conditions of Theorem \ref{thm:almost_smallest}. Let $\tilde{\pi}_k:[n_k]\to[n_k]$ be the output of Algorithm \ref{alg:algorithm2} when the inputs are the correlated SBMs $G^{\pi}$ and $G'$. Then, when the refinement matching algorithm (Algorithm \ref{alg:algorithm4} \citep{MRT21a}) is applied to  $\tilde{\pi}_k$, the output $\hat{\pi}_k:[n_k]\to[n_k]$ recovers the original permutation over $C_k$, i.e., $\hat{\pi}_k=\pi|_{C_k}$, with probability at least $1-2n_{\min}^{-D}$. 
\end{corollary}

\begin{remark}[Conditions on $k$ and $q$]
In Theorem \ref{thm:almost_smallest}, in addition to the density condition \eqref{eqn:main_thm_cond1},
we need two additional conditions on the number $k$ of communities \eqref{eqn:main_thm_cond3} and on the edge density $q$ between communities \eqref{eqn:main_thm_cond2}.
The condition on $k$ is required to have the signature vector of dimension $2^{k'\ell}=\Theta(\log n_{\min})$ at the maximum tree depth $\ell$, the diameter of $C_k$. If $n_{\min}p=\Theta(n_{\min}^{1/R\log \log n_{\min}})$, we only need a constant number of communities, but in the denser regime $n_{\min}p=n_{\min}^c$, $c>0$, we need $k=\Theta(\log\log n_{\min})$, since the diameter is a constant.
The condition on $q$ is needed to guarantee that there are enough edges between $C_k$ and other communities so that the $2^{k'}$-ary partition tree is well constructed by using the edges between communities.
\end{remark}

\subsection{Seeded Graph Matching}
\begin{algorithm}[t]
    \caption{Seeded Matching}\label{alg:algorithm5}
    \begin{algorithmic}[1]
     \REQUIRE{two graphs $\Gamma$ and $\Gamma'$  with community $(C_1,\dots, C_k)$}, a permutation $\pi_s:[n_s]\to[n_s]$ over $C_s$
     \ENSURE{a permutation $\hat{\pi}_t:[n_t]\to [n_t]$ over $C_t$}
    \FOR{$i \in C_{t}$ and $i' \in C_{t}$}
    \STATE let the weight $w(i,i')=0$
    \FOR{each $j \in C_{s}$}
    \IF{$\exists z \in \caN_{\Gamma(C_{t})}(i)$ and $\exists  z'\in \caN_{\Gamma'(C_{t})}(i')$ such that $(z,{\pi}_{s}(j))\in \caE(\Gamma)$ and $(z',j)\in \caE(\Gamma')$}
    \STATE $w(i,i')=w(i,i')+1$
    \ENDIF
    \ENDFOR
    \IF{$w(i,i') \geq ({n_{t}n_{s}pq})/{8}$}
    \STATE $\hat{\pi}_{t}(i')=i$
    \ENDIF
    \ENDFOR
    \STATE \textbf{return} $\hat{\pi}_{t}$
    \end{algorithmic}
\end{algorithm}
The last stage aims to recover the complete permutation $\pi:[n]\to[n]$ using the matched pairs in $C_k$ as initial {\it seeds}. The main idea is to apply a seeded graph matching algorithm (similar to that in \citep{KL13}), where vertex pairs $i\in G^\pi(C_t)$ and $i' \in G'(C_t)$ are compared by counting the number of common seed pairs in the 2-hop neighborhoods of $i$ and $i'$. Algorithm \ref{alg:algorithm5} describes the exact seeded matching algorithm, and Theorem \ref{thm:seeded graph matching} describes the condition for recovering a permutation $\pi_t$ over $C_t$ by using a known permutation $\pi_s$ over $C_s$.

\begin{theorem}[Seeded Matching for a Community]\label{thm:seeded graph matching}
There exists an absolute constant $K>0$ with the properties below. 
Consider the two graphs $G^\pi$ and $G'$, generated from the correlated SBMs defined in Sec. \ref{sec:model}, 
with the underlying permutation $\pi:[n]\to[n]$. Suppose that the community labels $(C_1,C_2,\dots, C_k)$ are given in both graphs. Additionally, the exact permutation $\pi_s:[n_s]\to[n_s]$ for the vertices in community $C_s$ is given. Assume that the parameters $(p, q, n_t, n_s, \alpha)$ of the model satisfy $\alpha<\frac{1}{10}$, $p\leq \frac{1}{256}$, and
        \begin{equation}\label{eq:seeded matching condition_main}
            n_{t}p \geq K\log n_{t},\; n_{t}n_{s}pq \geq K\log n_{t},\; n_{t}pq\leq 1/256.
        \end{equation}
		Then Algorithm \ref{alg:algorithm5} applied to $G^{\pi}$ and $G'$ yields $\hat{\pi}_t=\pi|_{C_t}$ with probability at least $1-{4}/{n_{t}}$.
\end{theorem}
Using the recovered permutation $\pi_k:[n_k]\to [n_k]$ over $C_k$,  we first apply Algorithm \ref{alg:algorithm5} to recover the permutation over $C_r$ for some $r\notin\{1,\dots, k', k\}$, which is a community not used in the previous stage of signature vector generation. 
We can check that all the conditions in \eqref{eq:seeded matching condition_main} except the last one are satisfied by the conditions on $|C_k|:=n_{\min}$ and $(p,q)$ in \eqref{eqn:main_thm_cond1} and \eqref{eqn:main_thm_cond2}, since $n_t\geq n_s=n_{\min}$. To satisfy $n_t pq \leq 1/256$ for any $t\in\{1,\dots, k-1\}$, we need an additional condition on $n_{\min}$ such that $n_{\min} \geq M_1 n^{10/19}$. 
With these conditions, the exact matching for the vertices in $C_r$ can be recovered by Algorithm  \ref{alg:algorithm5}.
Then, given $\pi_r:[n_r]\to[n_r]$, we use the matched pairs over $C_r$ as seeds and recover the matching over the rest of the communities $(C_1,\dots, C_{r-1}, C_{r+1}, \cdots, C_{k-1})$. The reason we do not use $C_k$ but $C_r$ as seeds in the permutation recovery is to avoid the dependency issue caused by reusing the edges over $C_{k}$ and $(C_1,\dots, C_{k-1})$ that were used in generating the signature vectors for the vertices of $C_k$. 

By combining Corollary \ref{cor:exact_smallest}  and Theorem \ref{thm:seeded graph matching}, we obtain our complete result to exactly recover $\pi:[n]\to[n]$.
  
\begin{corollary}[Exact Matching]\label{cor:combined}
For any constant $D>0$ there exists a constant $n_0$ depending on $D$ and absolute constants $\alpha_{1}, M,M',M_1>0$ with the properties below. 
Consider the two graphs $G^\pi$ and $G'$, generated from the correlated SBMs defined in Sec. \ref{sec:model}, 
with the underlying permutation $\pi:[n]\to[n]$ and correlation $1-\alpha$. 
Suppose that community labels $(C_1,C_2,\dots, C_k)$ are given in both graphs, and assume that $C_k$ is the smallest community with size $n_{\min}$. Assume that $\alpha\in(0,\alpha_1)$, \eqref{eqn:main_thm_cond1}, \eqref{eqn:main_thm_cond2} and \eqref{eqn:main_thm_cond3} hold. Also assume that $n_{\min} \geq M_1 n^{10/19}$.
Run Algorithm \ref{alg:algorithm2} with input graphs $(G^{\pi},G')$ to get an initial estimate of the permutation $\tilde{\pi}_k:[n_k]\to[n_k]$ over $C_k$. Using $\tilde{\pi}_k$ as input, run the refinement matching algorithm (Algorithm \ref{alg:algorithm4}) to obtain $\hat{\pi}_k:[n_k]\to[n_k]$ over $C_k$. Using $\hat{\pi}_k$ as input, run Algorithm \ref{alg:algorithm5} to find a matching $\hat{\pi}_r:[n_r]\to[n_r]$ over $C_r$ that was not used in Algorithm \ref{alg:algorithm2}. Then, by using $\hat{\pi}_r$ as input, run Algorithm \ref{alg:algorithm5} repeatedly to recover $\hat{\pi}_i$, for all $i\in[k]\backslash\{r,k\}$. By combining the recovered permutations over each community, we get $\hat{\pi}:=(\hat{\pi}_1,\dots, \hat{\pi}_k):[n]\to[n]$ such that $\hat{\pi}=\pi$ with probability at least $1-2n_{\min}^{-D}-n^{-1/19}$.
\end{corollary}
Note that Corollary \ref{cor:combined} implies Theorem \ref{thm:main_dense}.


\begin{remark}[Time complexity]
We summarize the total time complexity. For simplicity, we assume that the size of each community is equal to $m:=n/k$.
First, the time complexity to run Algorithm \ref{alg:algorithm1} is $O((m p)^\ell \cdot (k' m q) )$, since the depth-$\ell$ partition tree contains $(m p)^\ell $ vertices, and to assign each vertex to one of the partitioning nodes at each level,  $(k' m q)$ neighboring vertices need to be checked on average. Since we will choose $(mp)^\ell=O(m)$, this time-complexity is bounded by $O( k'm^2 q )$.
In Algorithm \ref{alg:algorithm2}, we apply Algorithm \ref{alg:algorithm1} to $m$ vertices in $C_k$ and compare $m^2$ pairs of signature vectors of dimension $(\log m)^5$. Thus, the complexity is $O( k'm^3 q +m^2(\log m)^5)$. In addition, the refinement matching by Algorithm \ref{alg:algorithm4}  requires time complexity of $O(m^3)$. The seeded graph matching by Algorithm \ref{alg:algorithm5} requires $O(m^2 (m p)(m))$-time complexity, and we run this algorithm a total of $(k-1)$ times. Thus, the complexity for seeded matching is $O(km^4p)$. The total time complexity is dominated by that of the seeded graph matching, ${O}(km^4p)$.
\end{remark}


\begin{remark}[Density regime]\label{rmk:general density}
We discuss the most general density regime, where there exist low-order polynomial-time algorithms for exact matching of the correlated SBMs with constant correlation. Assume approximately balanced community sizes, i.e., $n_{\min}=\Theta(n/k)$.
Our main result (Theorem \ref{thm:main_dense}) shows that exact matching is achievable in polynomial time for the correlated SBMs  whose average degrees within the smallest community satisfy $(\log n_{\min})^{1.1} \leq n_{\min}p \leq n_{\min}^{{1}/{20}}$ with two additional conditions on the number $k$ of communities \eqref{eqn:main_thm_cond3} and the edge density $q$ between communities \eqref{eqn:main_thm_cond2}. 
The result of \citet{MRT21a} (summarized in Table \ref{tbl:previous algorithm}), on the other hand, shows the existence of an efficient algorithm in the sparse regime $(1+\epsilon)\log n_{\min}\leq n_{\min} p \leq n_{\min}^{1/(C\log\log n_{\min})}$ for some constant $C>0$.
For both the algorithms, we need an additional assumption on $n_{\min}$ in \eqref{eq:main cor dense nmin}, which is equivalent to $ k\leq C' (\sqrt{\frac{n}{p}}(p-q) \wedge \frac{n(p-q)^{2}}{p \log n})$ for some constant $C'>0$, in order to recover the community labels in both graphs before performing graph matching. In summary, for the correlated SBMs with the density
$
   (1+\epsilon) \log n \leq np \leq n^{1/20},
$
there exists a low-order polynomial-time algorithm that can achieve exact matching of vertices if $\alpha\in(0,\text{const.})$ and
         $
            \frac{C\log \log n \cdot \log np}{\log n}\leq k\leq C' (\frac{\sqrt{n}(p-q)}{\sqrt{p}}\wedge \frac{n(p-q)^{2}}{p \log n}).
      $
  For the dense regime, we also need     
       $
             nq/k = \Omega \left( (\log \log n)^{2} \right).
        $
\end{remark}

%% file: section/proof_rev1.tex
\section{Outline of Proof}\label{sec:proof}

In this section, we outline the proof of Theorem \ref{thm:almost_smallest}.
Recall that Algorithm \ref{alg:algorithm1} constructs a $2^{k'}$-ary partition tree of depth $\ell$, rooted from each vertex $i\in[C_k]$, of size $|C_k|=n_{\min}$, where the tree rooted from $i$ is constructed within $C_k$ by using the intra-community edges, while the partitioning of the vertices at each depth of the tree is done using the inter-community edges between $C_k$ and $(C_1,\dots,C_{k'})$, the randomly selected $k'$ communities.
Our proof begins by showing that for $\ell= \lceil\frac{\log n_{\min
}}{40\log n_{\min}p}\rceil \wedge \lceil 42\log \log n_{\min} \rceil$ and $k'\ell=\Theta(\log\log n_{\min})$, the length-$\ell$ neighborhoods of the majority of vertices in $C_k$ form trees, and the vertices in the length-$\ell$ sphere $\mathcal{S}_{G(C_{k})}(i,\ell)$ of $i$ are uniformly partitioned into the leaf nodes $T^\ell_{\bss^{\ell}}(i)$, $\bss^\ell\in\{-1,1\}^{k'\ell}$, each with a rough size of $({n_{\min}p}/{2^{k'}})^\ell$.

We then compare the overlap between $T^{\ell}_{\bss^\ell}(i,G)$ and $T^{\ell}_{\bss^\ell}(i,G')$, the leaf nodes of the partition trees of $i\in C_k$ in graph $G$ and $G'$, and show that there is significant overlap, $\left| T^{\ell}_{\bss^\ell}(i,G) \cap T^{\ell}_{\bss^\ell}(i,G') \right| \geq ({n_{\min}p}/{2^{k'}})^\ell (1-\epsilon)^{k'\ell}$, where $\epsilon$ is an arbitrary small constant. This overlap leads to a correlation between the signature vectors $f(i)$ and $f'(i)$. 
When comparing the $\ell_2$-normalized distance between $f(i)$ and $f'(i)$, we use the sparsification procedure suggested by \citet{MRT21a}, to mitigate the dependency in the entries of the signature vector of $i$ caused by the overlap between the sets $T^\ell_{\bss^{\ell}}(i,G)$ and $T^\ell_{\bst^{\ell}}(i,G)$ for $\bss^{\ell}\neq \bst^{\ell}$. After sparsification, the normalized distance between the sparsified signatures of correct pairs is shown to be smaller than a given threshold with high probability. 


For incorrect pairs of vertices $i\neq i'$, the sizes of $T^{\ell}_{\bss^\ell}(i,G) \cap \mathcal{B}_{G_{0}(C_{k})}\left(i', \ell\right)$ and $T^{\ell}_{\bss^\ell}\left(i',G'\right) \cap \mathcal{B}_{G_{0}(C_{k})}(i, \ell)$ summed over the sparsified set are well controlled. As a result, the normalized distance between the sparsified signatures of distinct vertices is larger than the threshold with high probability.  
Detailed proofs for theorems are available in Appendix.

%% file: section/experiment_rev1.tex
\section{Experiment}\label{sec:experiment}
In this section, we evaluate our algorithm on both synthetic and real networks with inherent community structure. 
\begin{figure*}[t]
\centering
	\subfloat[SBM varying $p$ and $1-\alpha$ \label{fig:SBM_p}]{\includegraphics[width=0.25\textwidth, height=3.5cm]{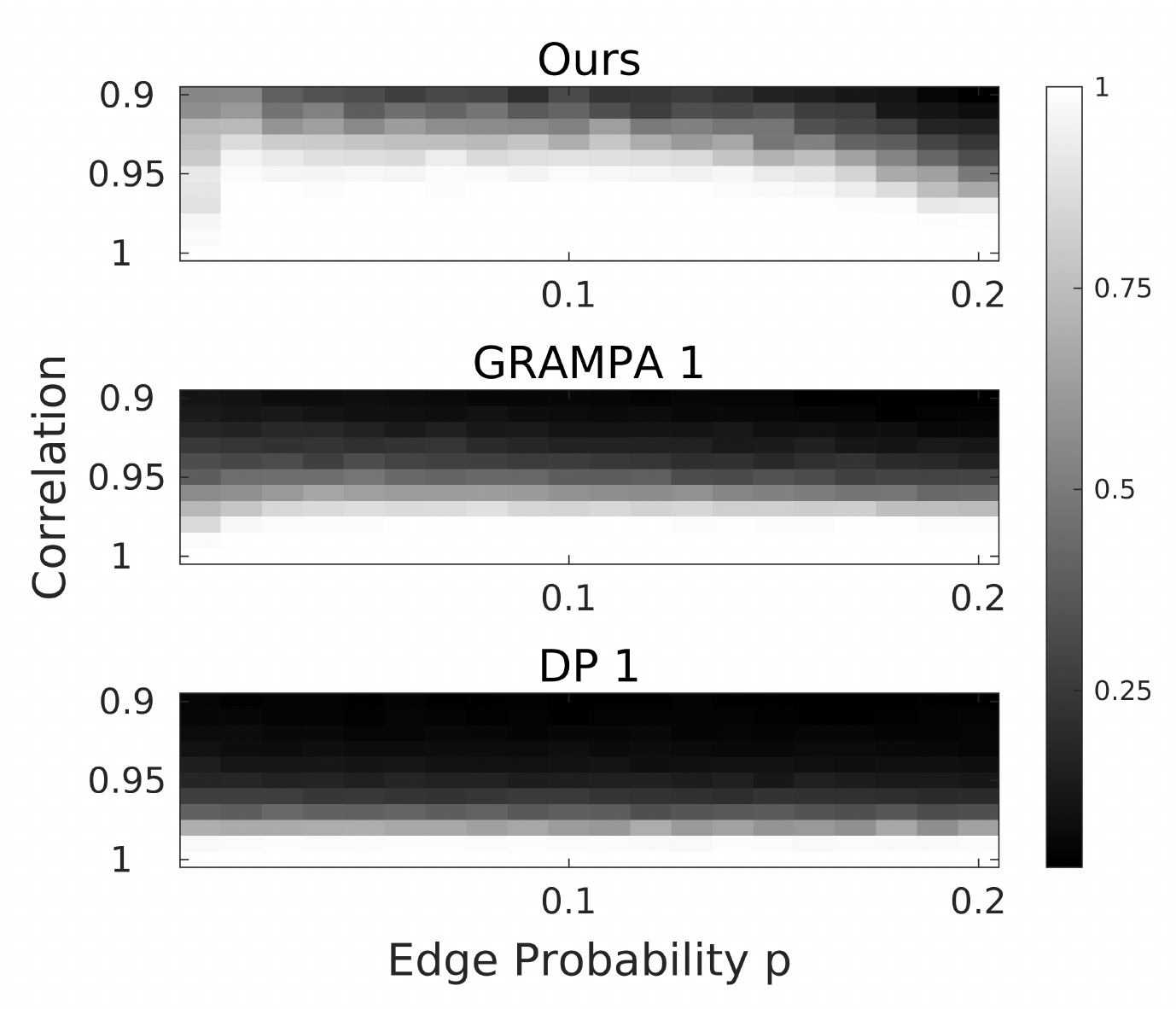}}
    \subfloat[SBM varying $1-\alpha$ \label{fig:SBM_overall}]{\includegraphics[width=0.3\textwidth, height=3.5cm]{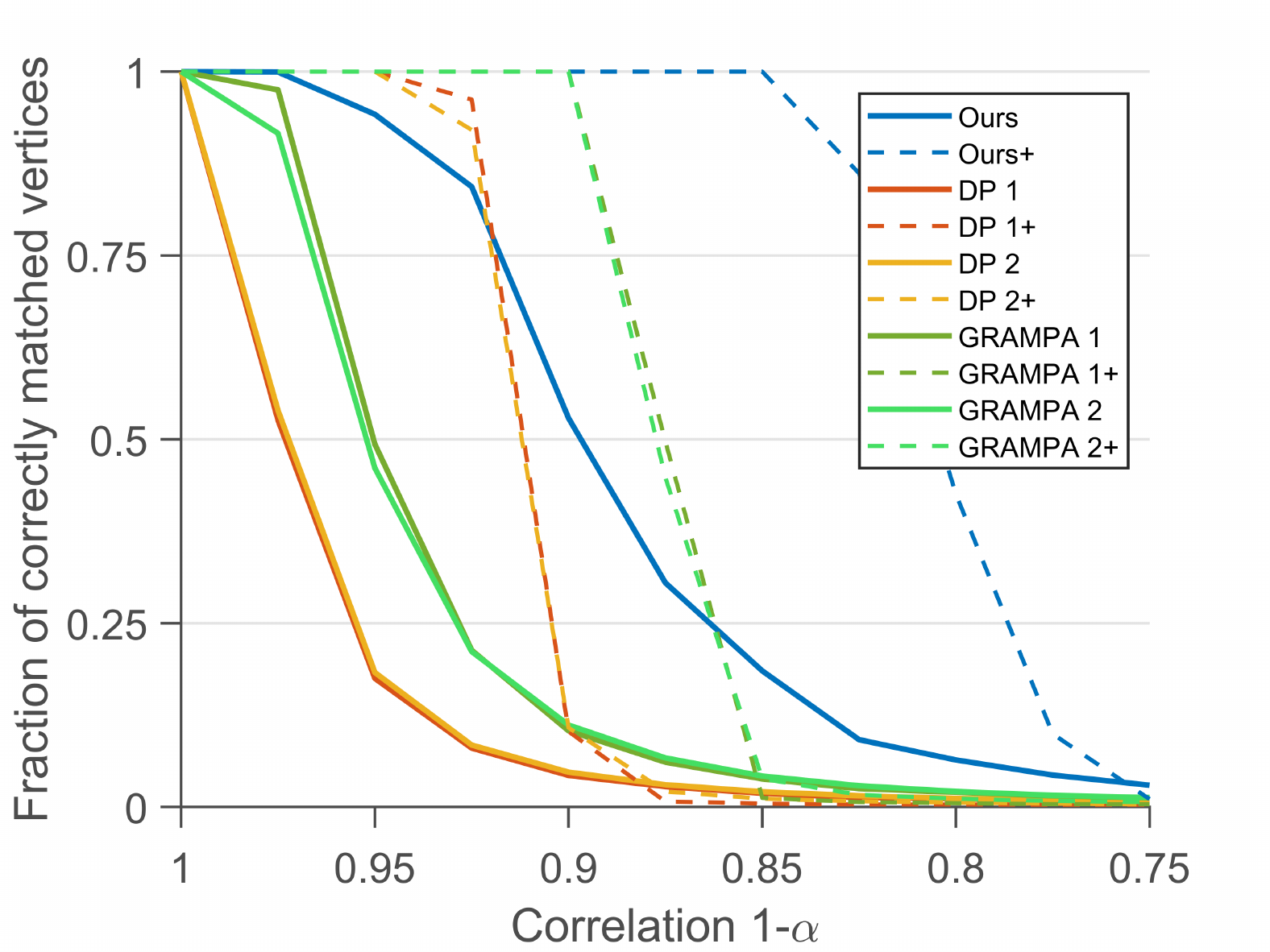}}
    \subfloat[BlogCatalog varying $1-\alpha$ \label{fig:Blog}]{\includegraphics[width=0.25\linewidth, height=3.5cm]{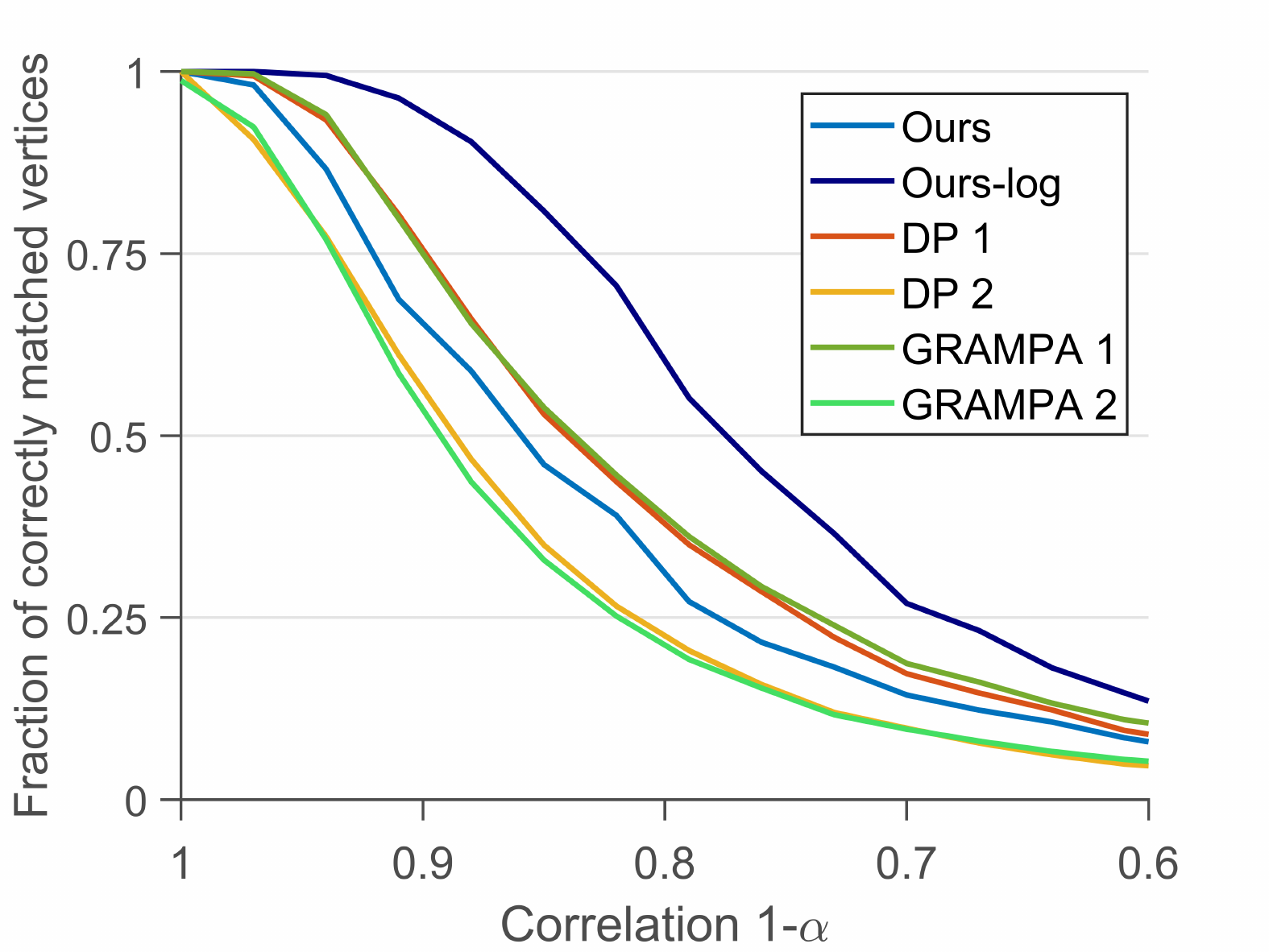}} 
    \subfloat[Movie \label{fig:Movie}]{\includegraphics[width=0.2\linewidth, height=3.5cm]{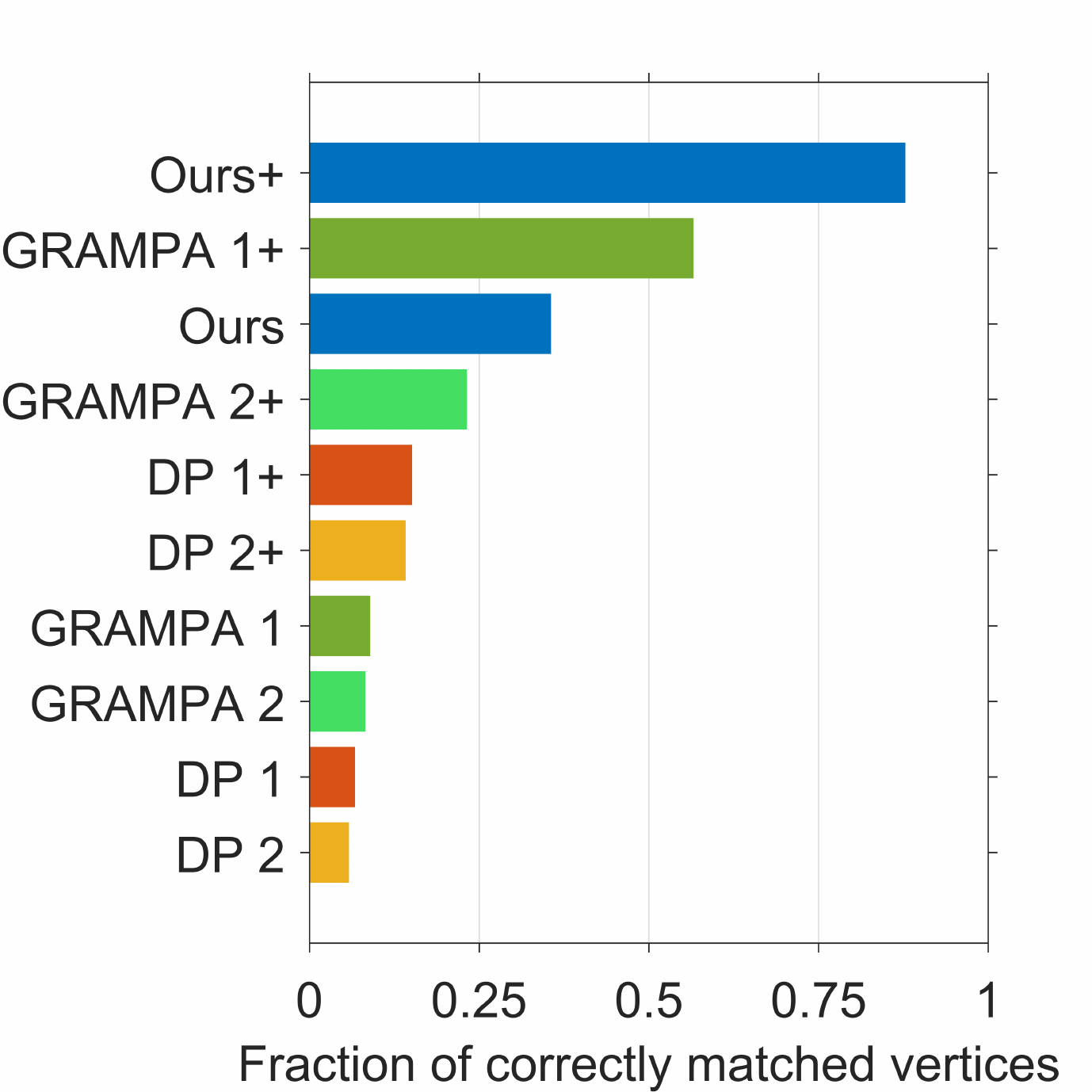}} 
    \caption{Comparison of our partition tree algorithm with other baselines, GRAMPA and Degree Profile (DP) for matching networks of (a)-(b) correlated SBMs, (c) BlogCatalog network and (d) two correlated movie networks (`RT' and 'IMDb').  }
    \label{fig:SBM}
\end{figure*}
\begin{table}[b]
\centering
    \caption{Details of three networks.}
     \label{tbl:dataset}  
\small{
\begin{tabular}{  c|  c  c c c }
 \toprule
 Type & Model/Source  & $n$ & $\#$ of edges & $k$\\
 \midrule
Synthetic   & SBM & 4,998 & 139k & 6 \\
\midrule
Sampled & BlogCatalog & 5,196 & 171k & 6\\
\midrule
 \multirow{2}{*}{Correlated } 
 & Rotten-Tomatoes & 4,780 & 110k & 6 \\ 
 & IMDb & 4,780 & 90k & 6 \\
 \bottomrule
\end{tabular}}
\end{table}
\paragraph{Datasets}
We use three types of datasets described in Table \ref{tbl:dataset}. 
For the synthetic dataset, we sample a parent graph $G_0$ from SBM with $p=0.025$ and $q=p/3$, where the nodes are partitioned into six communities of equal size. Then, two graphs $G$ and $G'$ are subsampled independently from $G_0$ by removing the edges of $G_0$ with probability $\alpha\in[0,1)$. For real datasets, we obtain correlated graphs in two different ways: 1) For the `BlogCatalog' dataset \citep{10.1145/2806416.2806501}, the parent graph is given by the network of the blogger community, where each user is treated as a vertex and the edges are connected between the users if they follow each other. We subsample the parent graph twice independently to generate the correlated graphs. The community label is given based on the predefined categories of the blogs. 2) For the `Movie' dataset, there is no parent graph, but two correlated networks are constructed from two different movie datasets from `Rotten Tomatoes' and `IMDb' that share common vertices (movie items). Edges are connected between vertices if there is at least one common actor in the top cast list, where the `Rotten Tomatoes' dataset contains up to six actors per movie and the `IMDb' dataset contains up to five actors per movie. The community label is assigned based on six groups of release years. 
We also perform experiments on two more real datasets, Flickr data \citep{10.1145/2806416.2806501} and ACM-DBLP data \citep{8443159}, and present the results in Appendix \ref{app:exp:add_exp}.


\paragraph{Baselines} 
We compare the performance of our algorithm to two baselines: `GRAMPA' \citep{pmlr-v119-fan20a} and `Degree Profile (DP)' \citep{DWMX21}.
As reviewed in Table \ref{tbl:previous algorithm}, these algorithms are computationally efficient graph-matching algorithms with provable performance guarantees. 
GRAMPA is an algorithm based on a spectral method where a similarity matrix is constructed using all eigenvector pairs of the two graphs. Then, a linear assignment problem applied to the similarity matrix yields a match (permutation over the vertices).
Degree Profile computes and compares the degree profile of the neighbors of each vertex, which is used to generate a similarity matrix. 
Since these baselines are developed in a way that does not use the community structure, for a fair comparison, we consider two variants of these algorithms (Methods 1 and 2): Method 1 applies the original algorithm (either GRAMPA or Degree Profile) to the whole graph of $[n]$ vertices to generate a $n\times n$ similarity matrix, but then selects a subpart of the matrix corresponding to the vertices in $C_k$ to define the similarity matrix over $[n_k]$. After solving the assignment problem to generate a permutation $\tilde{\pi}_k:[n_k]\to[n_k]$, we apply the rest of the steps  same as ours, including the refinement matching for $C_k$ by Algorithm \ref{alg:algorithm4} \citep{MRT21a} and seeded graph matching (Algorithm \ref{alg:algorithm5}) using the matched pairs in $C_k$ as initial seeds. Method 2 applies the original algorithm to each of the recovered communities to obtain the similarity matrices, and then solves the assignment problem and refinement matching for each of the communities.
To compare the performance of our algorithm with these baselines, we also define a $(n_k\times n_k)$-dimensional similarity matrix $S$ with its $(i,j)$-th entry $S_{ij}: =\sum_{\bss^\ell} \frac{(f_{\bss^\ell}(i)-f'_{\bss^\ell}(j))^{2}}{\mathrm{v}_{\bss^\ell}(i)+\mathrm{v}'_{\bss^\ell}(j)}$. Instead of thresholding the entries of $S$ as in our Algorithm \ref{alg:algorithm2}, we apply the linear assignment problem to $S$, as in GRAMPA \citep{pmlr-v119-fan20a}, to compare the accuracy of the recovered permutation over $C_k$. We also compare the final accuracy of these three algorithms.

Our code is publicly available at \url{https://github.com/cabaksa/cSBM_Matching}.
\paragraph{Results}

We present the experimental results for the SBM networks in Fig.~\ref{fig:SBM_p}--\ref{fig:SBM_overall}, the BlogCatalog networks in Fig. \ref{fig:Blog} and the Movie networks in Fig. \ref{fig:Movie}, respectively. Each baseline (GRAMPA and DP) has two versions (1 and 2), as explained before. The lines without $+$ indicate the fraction of correctly matched vertices within the comparison set $C_k$ after solving the linear assignment problem on the similarity matrix of each algorithm, and those with $+$ indicate the final matching accuracy over the $[n]$ vertices.

In Fig. \ref{fig:SBM_p}, we plot the empirical performances of our algorithm, GRAMPA1, and Degree Profile1, averaged over 10 runs, as we change the intra-community edge density $p$ ($q=p/3$) and the correlation parameter $1-\alpha$. A lighter color indicates a lower fractional error. We can observe that our algorithm is more robust against the decrease of the correlation $1-\alpha$ for all $p\in[0.01,0.2]$ range. In Fig. \ref{fig:SBM_overall} we can observe that our algorithm outperforms other baselines in recovering the permutation $\pi_k: [n_k]\to[n_k]$ over $C_k$ (solid lines) and maintains high accuracy longer as the correlation decreases. The final matching accuracy (dotted lines) is also better for our algorithm. 




The degree distribution of BlogCatalog graphs is known to follow a power law \citep{10.1145/2806416.2806501}.
Since Algorithm \ref{alg:algorithm2} computes the normalized distance between the signature vectors using the normal approximation of the binomial distribution, we adjust the definition of the signature vectors by using the log of the degree instead of the degree itself. This variation is denoted `Ours-log' in Fig. \ref{fig:Blog}. We can observe that `Our-log' outperforms other baselines, although the accuracy of our original algorithm without the degree fixing is lower than that of GRAMPA1 and DP1.


For `Movie' networks, unlike other data, there is no parent graph, and the two correlated graphs have different edge densities. The common edges between the two graphs are $\sim$80k, which means that the correlation $1-\alpha$ is 0.73 and 0.89, respectively, from the perspective of each graph. 
Figure \ref{fig:Movie} shows the fraction of correctly matched vertices on this dataset, where our algorithm achieves the best performance.

%% file: section/discussion_rev2.tex
\section{Discussion and Open Problems}\label{sec:discussion}


We presented a polynomial-time algorithm with low-order complexity for achieving exact matching on correlated SBMs with constant correlation. Our result is the first of its kind in the dense regime.
In the sparse regime of the graphs where  $(1+\epsilon)\log n_{\min}\leq n_{\min} p \leq n_{\min}^{1/(C\log\log n_{\min})}$, by directly applying the binary partition tree algorithm from \citet{MRT21a} to each recovered community, one can achieve the exact graph matching with constant correlation in polynomial-time complexity. By generalizing the idea from  \citet{MRT21a}, we developed a $2^{k'}$-ary partition tree algorithm for exact graph matching for the correlated SBMs, using the known community structure and the edges between communities. Our algorithm achieves the exact matching in $(\log n_{\min})^{1.1} \leq n_{\min}p \leq n_{\min}^{1/20}$ with constant correlation, with two additional conditions on the minimum number of communities $k$ and the edge density across communities. Our work leaves several important open questions, as discussed below. 

\paragraph{Exact graph matching without exact community recovery} In the correlated SBMs, where $G^{\pi}$ and $G'$ are sampled from a parent graph $G_0$, one can try to match the vertices of the correlated graphs $G^{\pi}$ and $G'$ using a known community structure, or one can try to recover the community labels by first matching the two correlated graphs. The second problem was considered in \cite{RS21} for the case of two equal-sized communities, and it was shown that there exist regimes of $(p,q,\alpha)$ where having the correlated graph $G^\pi$ as side information can provide an information advantage in recovering the community of $G'$ (or $G_0$) by graph matching between $G^\pi$ and $G'$.
Our work addressed the first problem and showed that there exists a low-order polynomial-time algorithm that can achieve the exact graph matching with constant correlation by using the known community structure in both graphs. However, it requires an additional assumption on $n_{\min}$, the minimum size of the community, as in \eqref{eq:main cor dense nmin}, to recovery the community structure exactly in polynomial time.  Then the natural question is whether exact community recovery is necessary to achieve exact graph matching with constant correlation in polynomial time. The answer may be no, and this is an interesting future work. 

\paragraph{Exact graph matching via exact community recovery} Even if we assume that the community labels in both graphs are given as side information, there is a gap between the information-theoretic  limits for exact graph matching and the regime where the current polynomial algorithms can exactly recover the matching between the vertices of the correlated SBMs.
First, our algorithm achieves the exact graph matching in the dense regime, but up to $n_{\min} \leq n_{\min}^{1/20}$. The upper bound on $n_{\min}$ is needed in the proof of Lemma \ref{prop:type}  to guarantee that the majority of vertices in $C_k$ form trees up to depth $\ell$.
Thus, we may need another algorithmic approach beyond the partition tree to achieve the exact matching in the denser regime, e.g., $n_{\min}p=n_{\min}^{1-o(1)}$. 
For the sparse regime, to apply the algorithm from \citet{MRT21a} to each community, we need the condition $(1+\epsilon)\log n_{\min}\leq n_{\min} p \leq n_{\min}^{1/(C\log\log n_{\min})}$. Consider the correlated SBMs with two balanced (size $n/2$ each) communities. The information-theoretic limit from \citet{RS21} shows that exact matching is possible if $\frac{a+b}{2}(1-\alpha)>1$ on the correlated SBMs with two balanced communities,  where $p=\frac{a \log n}{(1-\alpha) n}$, and $q=\frac{b \log n}{(1-\alpha)n}$ for positive constant $a,b$. Thus, if $a=1.5$, $b=1$, and $\alpha$ is a sufficiently small constant, the exact matching is information-theoretically possible. However, since the density of each community in this regime is  $n_{\min}p\approx 0.75\log n_{\min}<\log n_{\min}$, we cannot apply the matching algorithm from \citet{MRT21a} to each community. Thus, there is  a gap between the information-theoretic limit and the computational limit when we apply the matching algorithm to each recovered community. This implies that we need another algorithmic approach that uses both inter-/intra-community edges even in the sparse regime to bridge the gap between the information-theoretic limit and the computational limit. 
\paragraph{Using seeded graph matching} Our algorithm first uses Algorithm \ref{alg:algorithm2} to achieve the almost exact matching  on community $C_k$, and then uses the refinement matching (Algorithm \ref{alg:algorithm4}) to achieve exact matching on community  $C_k$. Finally, it uses vertex pairs from the community $C_k$ as seeds to complete exact matching for the rest of the vertices in the correlated SBMs. In the correlated ER model, the seeded graph matching algorithm has been extensively studied. There are seeded graph matching algorithms  based on percolation \citep{YG13,Percolation}, and algorithms using large neighbor statistics \citep{MX20}. Furthermore, \citet{KHG15,LS18,YXL21} have proposed algorithms for seeded graph matching, even when the initial seeds are noisy. It is also worth investigating how we can improve the conditions for our exact graph matching algorithm by adjusting the first step and allowing some noisy matching from the initial seed set $C_k$. 

%% file: section/acknowledgements.tex
\section{Acknowledgements}
This research was supported by the National Research Foundation of Korea under grant 2021R1C1C11008539, and  by the MSIT(Ministry of Science and ICT), Korea, under the ITRC(Information Technology Research Center) support program(IITP-2023-2018-0-01402) supervised by the IITP(Institute for Information \& Communications Technology Planning \& Evaluation).

%% file: section/supp_rev5.tex
The appendix of this paper is organized as follows.
Additional experimental details and results are presented in Sec. \ref{app:sec:experiment}, and the proofs of the theoretical results follow.
The proof of Corollary \ref{cor:main_all}, where we give the sufficient conditions for exact matching of the correlated SBMs with unknown community structure, is presented in Sec. \ref{app:sec:cor_main_all}.
Theorem \ref{thm:seeded graph matching}, which presents the performance of seeded graph matching (Algorithm \ref{alg:algorithm5}), is proved in Sec. \ref{app:sec:thm_seeded}.
The proof of Corollary \ref{cor:combined} is given in Sec. \ref{app:sec:cor_combined}.
The proof of our main theorem (Theorem \ref{thm:almost_smallest}) is presented in Sec. \ref{sec:proof of almost}. Various lemmas to prove Theorem \ref{thm:almost_smallest} are presented in Sections \ref{sec:vertex classes}, \ref{sec:proof of correct final} and \ref{sec:proof of wrong final}. Some technical tools to prove the main results are summarized in Sec. \ref{sec:appx tool}, and the previous results of \citet{MRT21a} regarding the refinement algorithm for exact matching and its performance guarantees are summarized in Sec. \ref{app:sec:exact}.

\section{Experiment}\label{app:sec:experiment}
\subsection{Experimental Details}
\begin{itemize}
\item Movie data: To generate two correlated networks sharing a common set of vertices, we use the `Rotten Tomatoes'\footnote{\url{https://www.kaggle.com/datasets/ayushkalla1/rotten-tomatoes-movie-database?resource=download}} dataset and the `IMDb'\footnote{\url{https://www.kaggle.com/datasets/jyoti1706/imdbmoviesdataset?datasetId=9670}} dataset. We generate two networks that share a common set of vertices corresponding to 4780 movies. Edges are connected between vertices if there is at least one common actor in the top cast list, where the `Rotten Tomatoes' dataset contains up to six actors per movie and the `IMDb' datasets contains up to five actors per movie. The community label is assigned based on six groups of release years. 

\item Community labels: In all the experiments reported in Section \ref{sec:experiment}, to focus on the performance comparisons between different matching algorithms, we assume that the ground-truth community labels are given in the networks. 

\item  Parameters: In order to apply our algorithm (Algorithm \ref{alg:algorithm1}), the model parameters $(p,q)$ are needed to obtain the values of $n_a q, n_k p$ and $n_k p(1-p)$, which are used to compute the signatures. Since these parameters are unknown in real datasets, we estimate these values from the data. Specifically,  the median of node degrees within and across communities is used instead of $n_kp$ and $n_aq$, respectively, and the empirical variance of node degrees within $C_k$ replaces $n_k p(1-p)$.

\item Signature vector: The two parameters $(k',\ell)$, which determine the dimension of the signature vectors, can be considered as hyperparameters of our algorithm. In our experiments, we choose $k'=4$ and $\ell=2$. In Appendix \S\ref{sec:varyingkl}, we provide additional experiments with varying $(k',\ell)$.

\item Similarity matrix and assignment problem: The two main baselines (GRAMPA and Degree Profile) use two-step procedures, generating the similarity matrix over the vertices and solving the assignment problem to output a permutation from the similarity matrix. To compare the performance of our algorithm with these baselines, we also define a $(n_k\times n_k)$-dimensional similarity matrix $S$ by defining its $(i,j)$-th entry as $S_{ij}: =\sum_{\bss^\ell} \frac{(f_{\bss^\ell}(i)-f'_{\bss^\ell}(j))^{2}}{\mathrm{v}_{\bss^\ell}(i)+\mathrm{v}'_{\bss^\ell}(j)}$. 
If both the partitioning nodes $T^{\ell}_{\bss^\ell}(i)$ and $T'^{\ell}_{\bss^\ell}(j)$ are empty for some $\bss^\ell\in \{-1,1\}^{k'\ell}$, we set $\frac{(f_{\bss^\ell}(i)-f'_{\bss^\ell}(j))^{2}}{\mathrm{v}_{\bss^\ell}(i)+\mathrm{v}'_{\bss^\ell}(j)}:=0$. 
Instead of thresholding the entries of $S$ as in our Algorithm \ref{alg:algorithm2} to generate the permutation over $C_k$, we apply the linear assignment problem to $S$ suggested in \citep{pmlr-v119-fan20a}, to compare the accuracy of the permutation obtained from our similarity matrix with that of the two baselines.

\item Seeded matching : In Algorithm \ref{alg:algorithm5}, we count the number of common seeds on 2-hop neighborhood when comparing each pair of vertices, and match the vertices if the number of common seeds exceeds the given threshold. 
The reason we considered 2-hop neighborhood instead of 1-hop neighborhood was to make the algorithm work even from very sparse regime of inter-community edge density, $n_{\min}q= C (\log \log n_{\min})^2$ for $C>0$. 
In the simulation, since the inter-community edge density is not very sparse, it is sufficient to count the number of seeds in the 1-hop neighborhood as the seeded matching algorithm proposed by \citet{KL13}. Thus, we use the greedy seeded matching algorithm, where we match the pair of vertices with the maximum number of common seeds in the 1-hop neighborhood among the set of remaining vertices as we match the nodes one by one.

\item Refinement matching : After the initial matching, refinement matching is performed on each community for $T$ rounds. 
Theorem 2.4 in \citep{MRT21a} shows that exact matching is achievable from partial matching using refinement matching by Algorithm \ref{alg:algorithm4}. In Algorithm \ref{alg:algorithm4}, the refinement step runs over $\left\lceil \log_{2}n \right\rceil$ rounds by checking whether the number of common pairs in the 1-hop neighborhood of the matched pairs exceeds a certain threshold and whether it is below a certain threshold for all unmatched pairs. Since setting the threshold requires knowledge of the model parameters $(p,q)$, instead we solve the linear assignment problem $\pi_{t} = \argmax_{\pi_{*} \in \mathcal{S}_n} \sum_i |\pi^{-1}_{t-1}\left(\caN_{G^{\pi}}(\pi_{*}(i))\right) \cap \caN_{G'}(i)|$ where $\mathcal{S}_n$ is the permutation matrix. 
Appendix \S\ref{app:exp:iteration} shows the effect of the number of iteration rounds $T$ on the final accuracy of the matching.

\end{itemize}



\subsection{Degree Distribution of Datasets}

\begin{figure}[!htb]
\centering
	\subfloat[SBM \label{fig:deg_sbm}]{\includegraphics[width=0.3\textwidth, height=3.5cm]{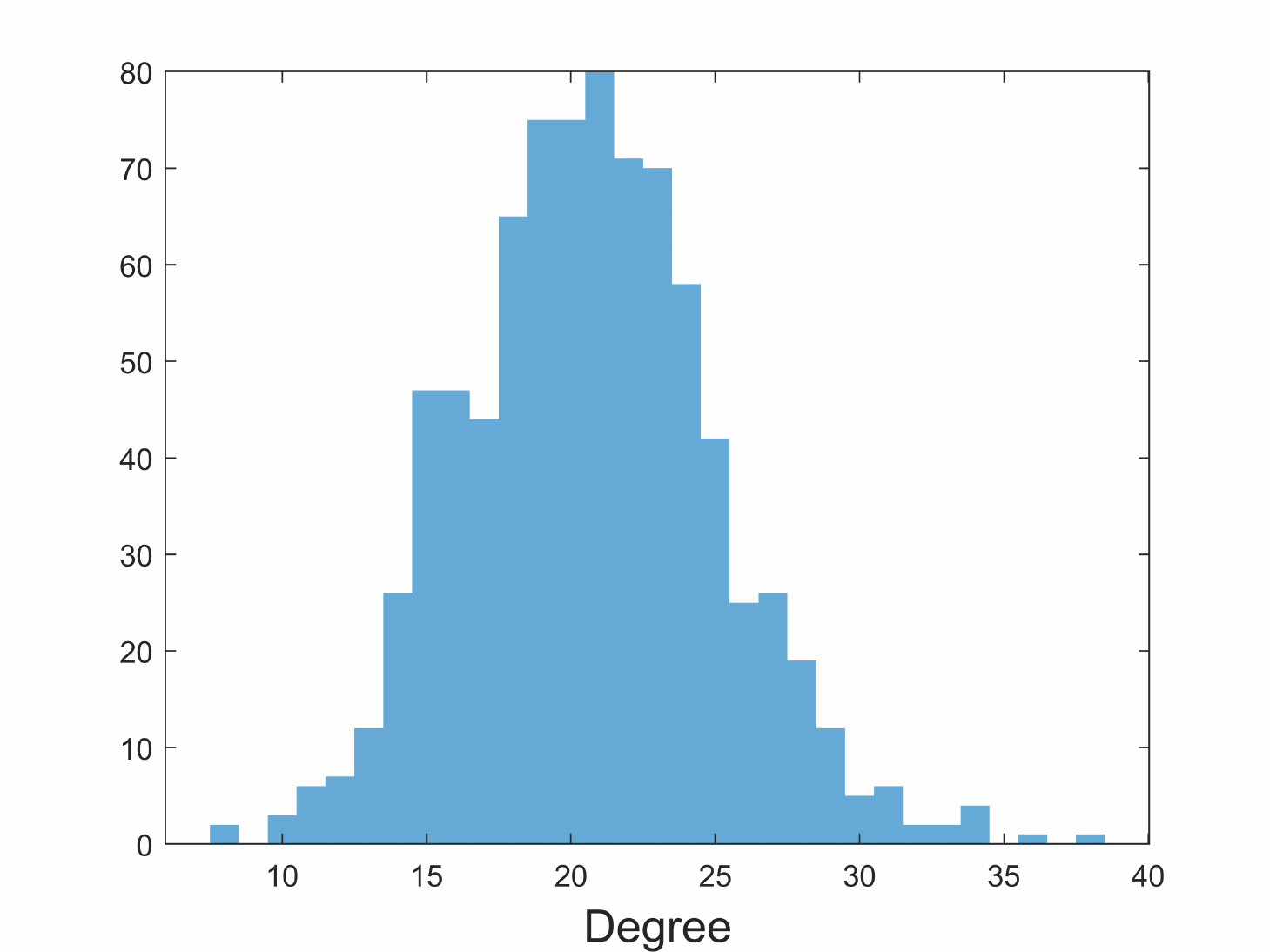}}
    \subfloat[BlogCatalog \label{fig:deg_blog}]{\includegraphics[width=0.3\textwidth, height=3.5cm]{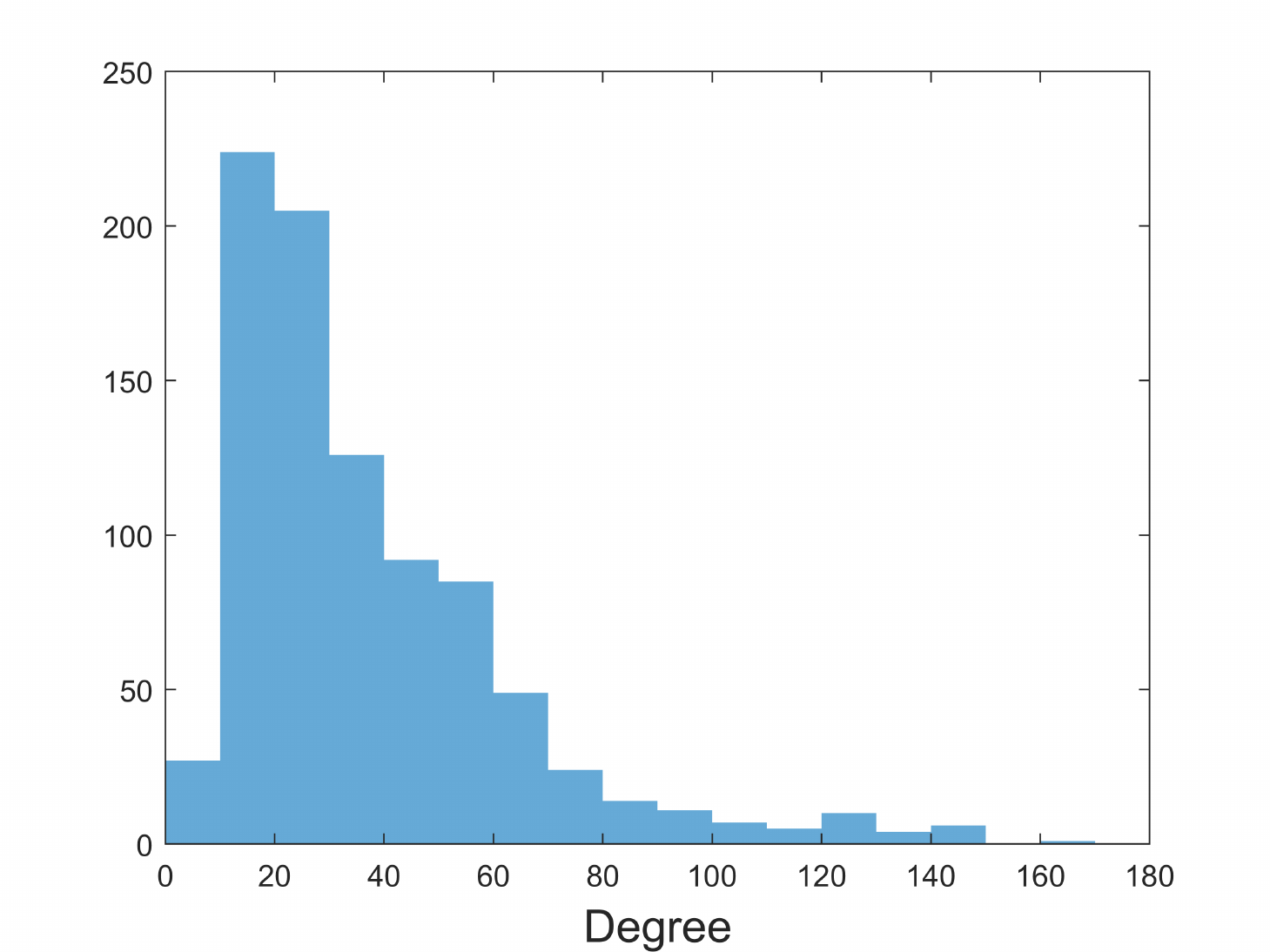}}
    \subfloat[Movie \label{fig:deg_mv}]{\includegraphics[width=0.3\linewidth, height=3.5cm]{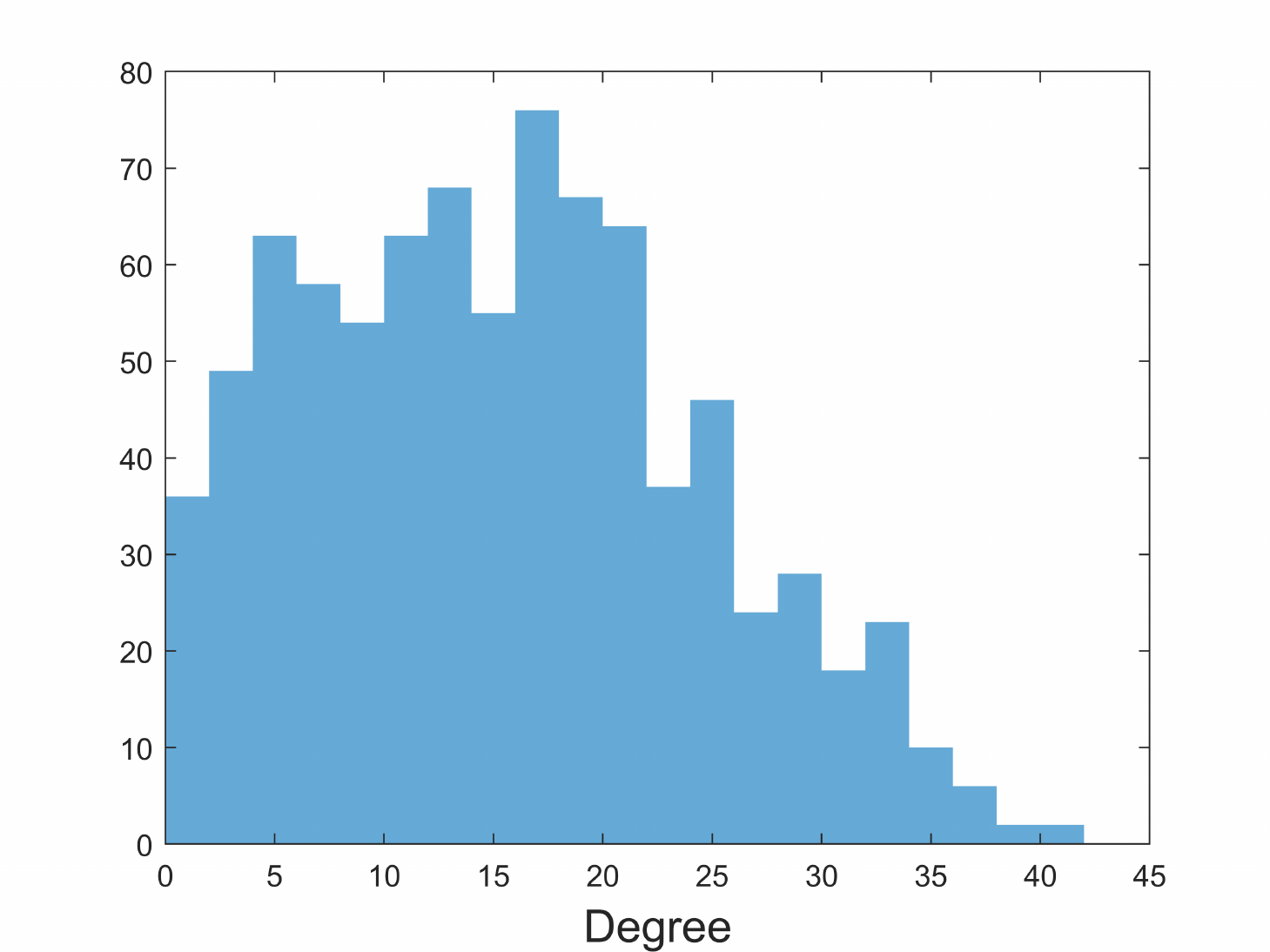}} 
    \caption{Degree histogram for three different datasets}
    \label{fig:deg}
\end{figure}

Figure \ref{fig:deg} shows the histograms of vertex degrees within the community for three different datasets, considered in Sec. \ref{sec:experiment}. 
Unlike the `SBM' dataset, where we sampled the network from the statistical model and thus the degree distribution within each community follows the binomial distribution, the degree distribution in real datasets deviates from the SBM. In particular, the degree distribution of the `BlogCatalog' dataset is very different from the SBM and follows a power law. Since Algorithm \ref{alg:algorithm2} calculates the normalized distance between the signature vectors using the normal approximation to the binomial distribution, we modify the definition of the signature vectors by using the logarithm of the degree instead of the degree itself for this dataset.

\subsection{Impact of Hyperparamters $k'$ and $\ell$ for Partition Tree Algorithm}\label{sec:varyingkl}

\begin{table*}[!hbt]
\centering
\caption{Fraction of correctly matched vertices per $k'$ and $\ell$}
\begin{tabular}{c|cc|cc|cc}
\toprule
                    & \multicolumn{2}{c|}{\begin{tabular}[c]{@{}c@{}}SBM\\ ($1-\alpha=0.9$)\end{tabular}} & \multicolumn{2}{c|}{\begin{tabular}[c]{@{}c@{}}BlogCatalog\\ ($1-\alpha=0.85$)\end{tabular}} & \multicolumn{2}{c}{Movie}                                         \\ \hline
\backslashbox{$k'$}{$\ell$}              & \multicolumn{1}{c|}{1}       & \multicolumn{1}{c|}{2}              & \multicolumn{1}{c|}{1}           & \multicolumn{1}{c|}{2}                 & \multicolumn{1}{c|}{1}    & \multicolumn{1}{c}{2}       \\ \hline
 1 & \multicolumn{1}{c|}{0.0099}    & \multicolumn{1}{c|}{0.0148}    & \multicolumn{1}{c|}{0.0605}    & \multicolumn{1}{c|}{0.1283}     & \multicolumn{1}{c|}{0.0342} & \multicolumn{1}{c}{0.0565} \\ \hline
 2 & \multicolumn{1}{c|}{0.0298}    & \multicolumn{1}{c|}{0.0683}       & \multicolumn{1}{c|}{0.1552}        & \multicolumn{1}{c|}{0.3582}            & \multicolumn{1}{c|}{0.0907}  & \multicolumn{1}{c}{0.2167} \\ \hline 
 3 & \multicolumn{1}{c|}{0.0740}    & \multicolumn{1}{c|}{0.2339}      & \multicolumn{1}{c|}{0.2622}        & \multicolumn{1}{c|}{0.6317}        & \multicolumn{1}{c|}{0.1567} & \multicolumn{1}{c}{0.2945}  \\ \hline
 4 & \multicolumn{1}{c|}{0.1505}    & \multicolumn{1}{c|}{0.5598}    & \multicolumn{1}{c|}{0.3747}        & \multicolumn{1}{c|}{0.8015}           & \multicolumn{1}{c|}{0.1908} & \multicolumn{1}{c}{0.3557} \\ \hline
5 & \multicolumn{1}{c|}{0.2241}    & \multicolumn{1}{c|}{0.8252}          & \multicolumn{1}{c|}{0.4452}        & \multicolumn{1}{c|}{0.8989}           & \multicolumn{1}{c|}{0.2073} & \multicolumn{1}{c}{0.3416}\\ \bottomrule
\end{tabular}
\label{tbl:impact of k' and ell}
\end{table*}

In our partition tree algorithm with $2^{k'}$-ary tree of depth-$\ell$, the two parameters $(k',\ell)$, which determine the dimension of the signature vectors, can be considered as hyperparameters.
In Table \ref{tbl:impact of k' and ell}, we show the impact of these parameters by reporting the performance of almost exaction matching over $C_k$ for different $(k',\ell)$ pairs, where the numbers indicate the fraction of correctly matched vertices over $C_k$ after solving the linear assignment problem for the similarity matrix obtained by the signature vectors from the partition tree. The numbers reported for the `SBM' and `BlogCatalog' datasets are the results averaged over 20 independent runs by randomly sampling the networks, while the numbers reported for the `Movie' dataset are the result from a given fixed correlated networks (without any sampling or averaging).
For a fixed $\ell$, it can be seen that the accuracy increases with $k'$ in most cases.
For a fixed $k'$, the accuracy increases as $\ell$ increases in all the cases considered.
When we generate a $2^{k'}$-ary partition tree of depth $\ell$, if $k'\ell>8$, the number of leaves $2^{k'\ell}$ in a partition tree becomes larger than the number of vertices in $C_k$ (which is about $830$), and thus some of the leaf nodes become empty. So we choose $(k',\ell)=(4,2)$ in our experiment.

\subsection{Final Accuracy of Refinement Matching with Different Iterations on BlogCatalog Dataset}\label{app:exp:iteration}
\begin{figure}[!htb]
\centering
	\subfloat[Ours-log \label{fig:T_ours}]{\includegraphics[width=0.31\textwidth, height=3.5cm]{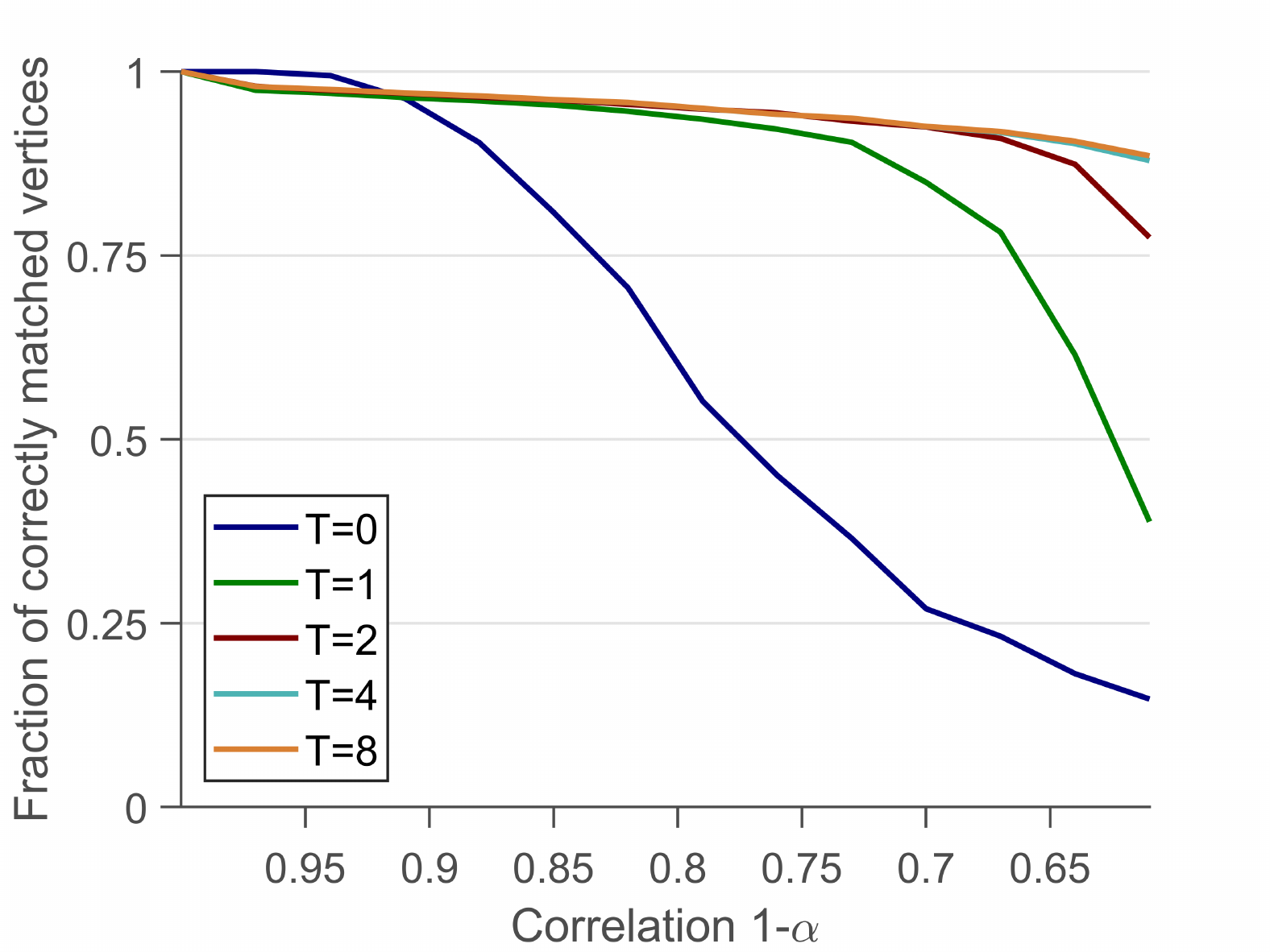}} \hfill
    \subfloat[Degree Profile 1 \label{fig:T_dp1}]{\includegraphics[width=0.31\textwidth, height=3.5cm]{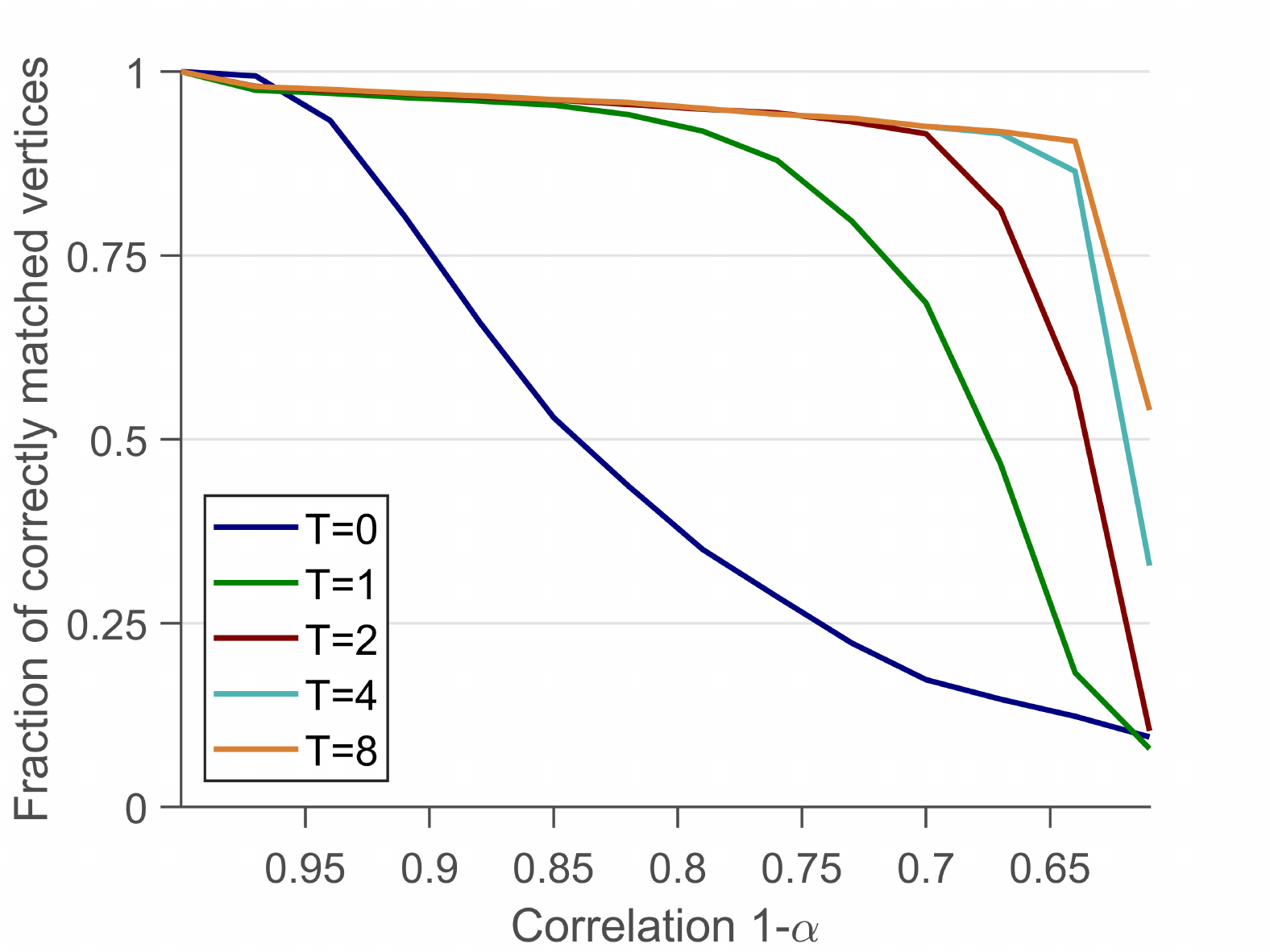}} \hfill
    \subfloat[Degree Profile 2 \label{fig:T_dp2}]{\includegraphics[width=0.31\linewidth, height=3.5cm]{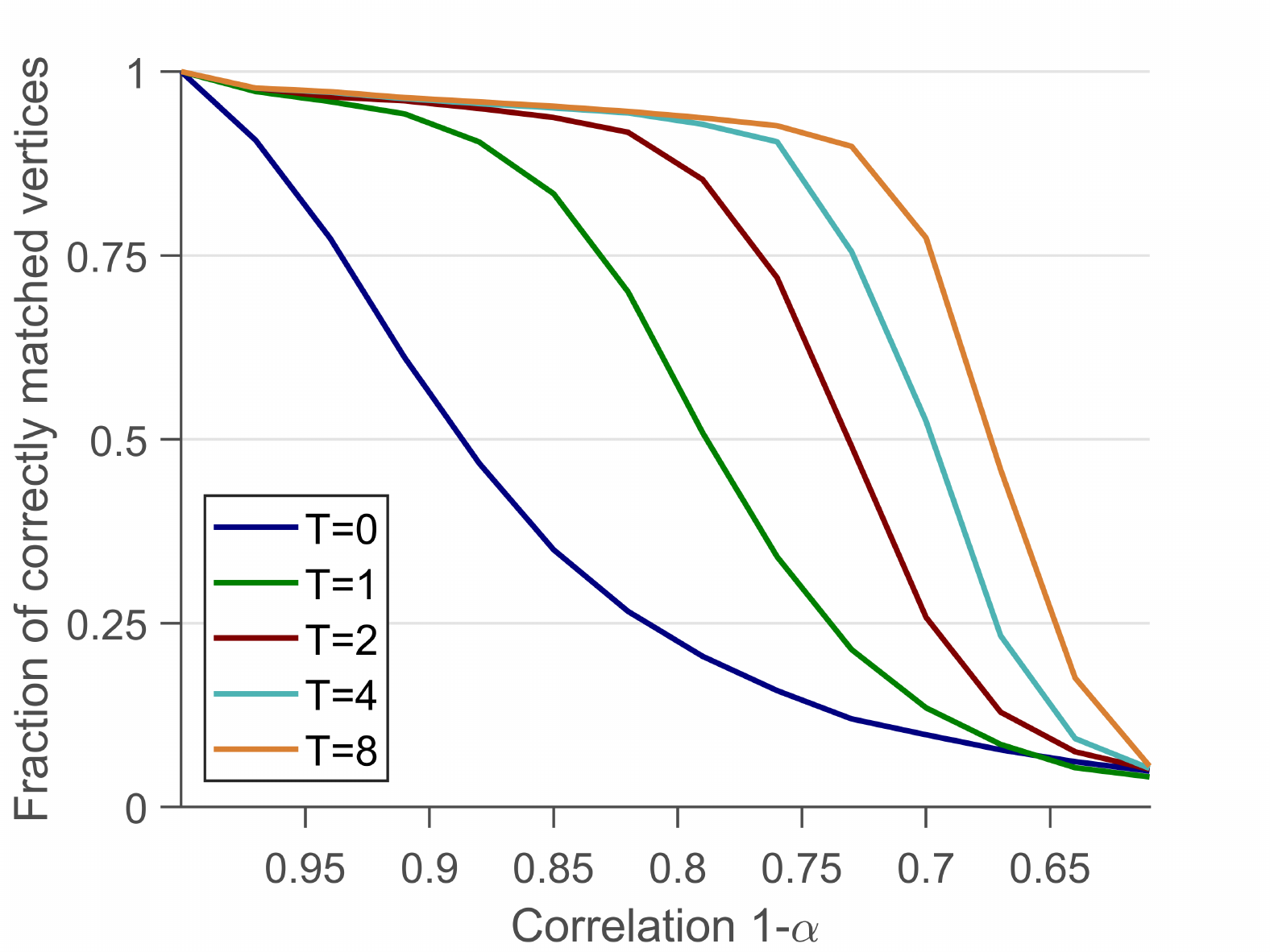}}
    
    \subfloat[GRAMPA 1 \label{fig:T_gp1}]{\includegraphics[width=0.31\textwidth, height=3.5cm]{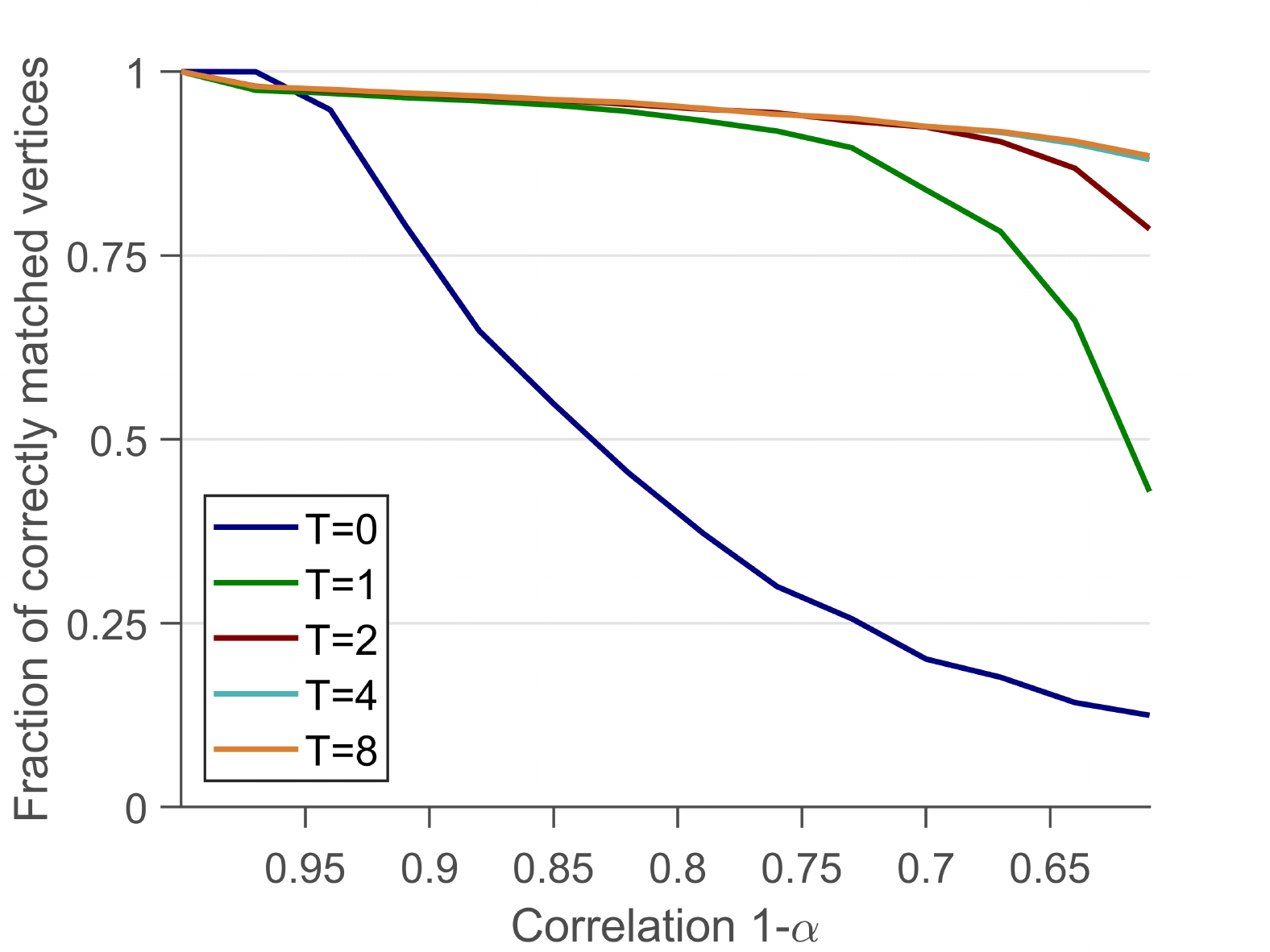}}
    \subfloat[GRAMPA 2 \label{fig:T_gp2}]{\includegraphics[width=0.31\linewidth, height=3.5cm]{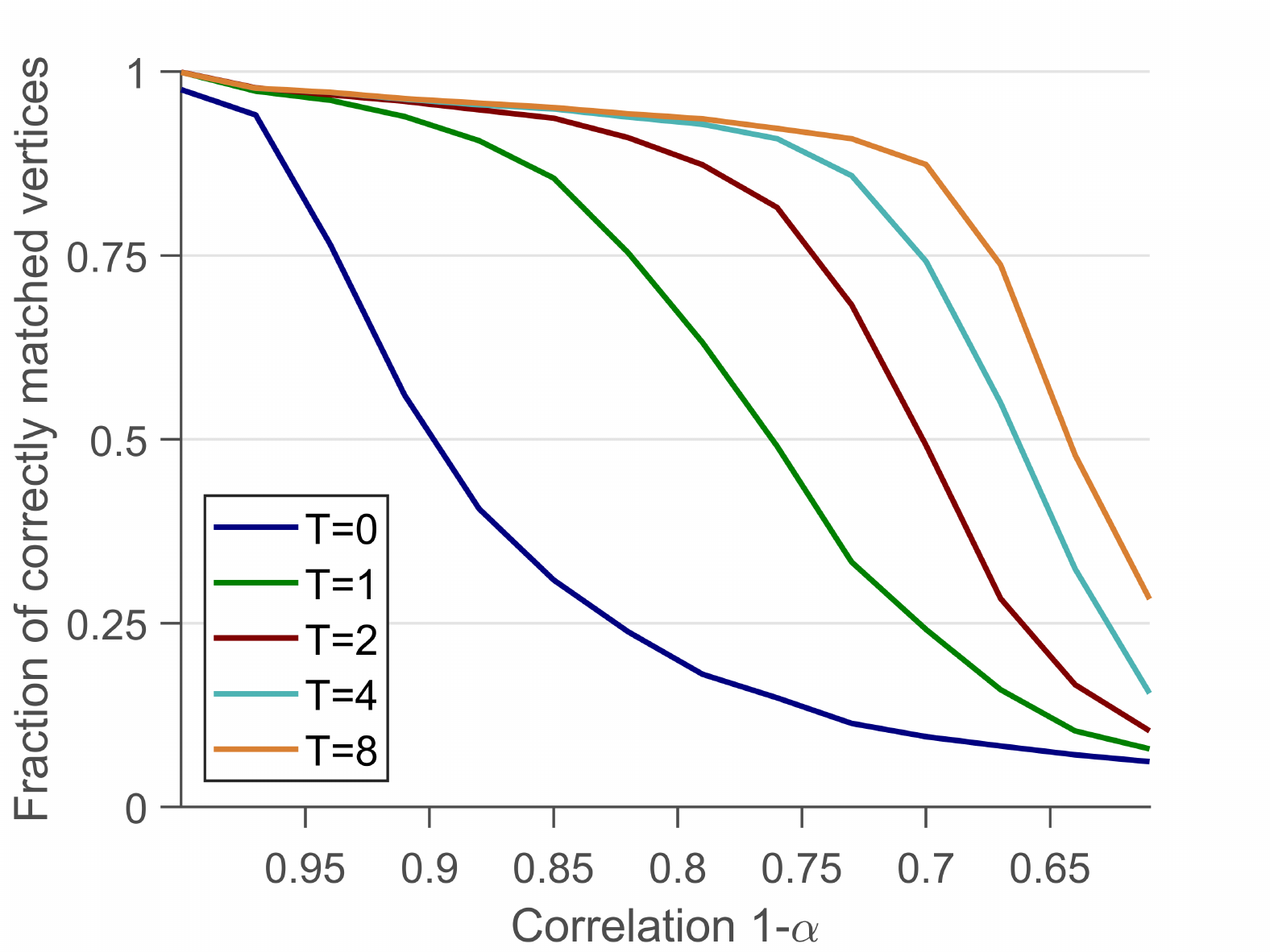}}   
    \caption{The fraction of correctly matched vertices in the final matching with respect to the correlation $1-\alpha$, for different number $T$ of iteration rounds.}
    \label{fig:test_T}
\end{figure}
Figure \ref{fig:test_T} shows the effect of iteration rounds on the refinement matching for the BlogCatalog Dataset. 
In most experiments, the accuracy improves as the number of iteration rounds $T$ increases, but each algorithm requires a different number of rounds to converge, depending on the initial accuracy of the algorithm. 
For example, for $T=2$ and $0\leq 1-\alpha \leq 0.65$, our algorithm and GRAMPA 1 converge, while others require more iterations to converge. This result implies that our algorithm requires less computation time for the final matching. 

\subsection{Additional Experiments}\label{app:exp:add_exp}
\begin{table}[b]
\centering
    \caption{Details of additional datasets.}
     \label{tbl:add_dataset}  
\small{
\begin{tabular}{  c|  c  c c c }
 \toprule
 Type & Model/Source  & $n$ & $\#$ of edges & $k$\\
\midrule
Sampled & Flickr & 7,575 & 240k & 9\\
\midrule
 \multirow{2}{*}{Correlated } 
 & ACM & 6,299 & 25k & 4 \\ 
 & DBLP & 6,299 & 28k & 4 \\
 \bottomrule
\end{tabular}}
\end{table}
\begin{figure}[!htb]
\centering
	\subfloat[Flickr \label{fig:deg_fli}]{\includegraphics[width=0.3\textwidth, height=3.5cm]{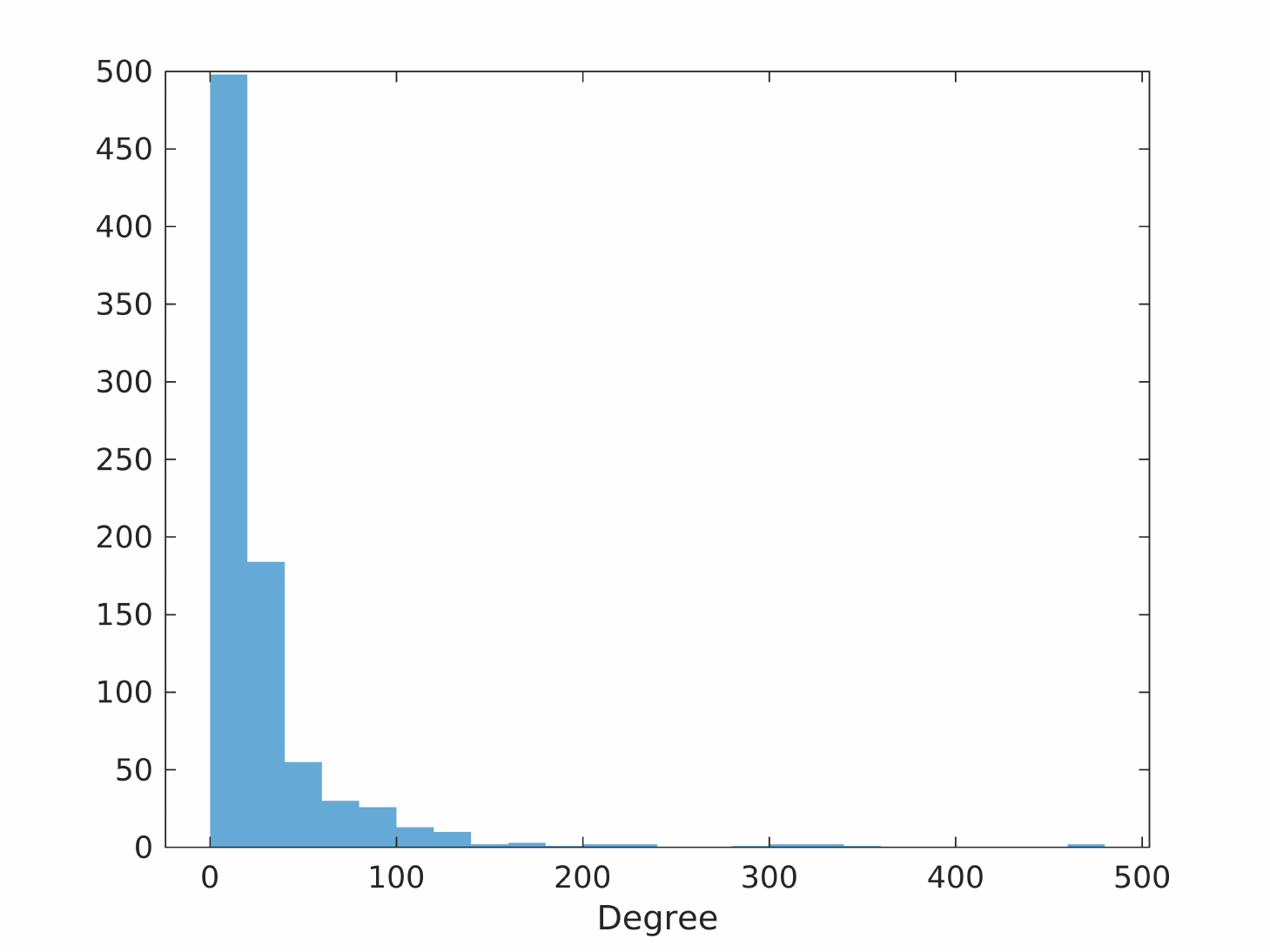}}
    \subfloat[ACM-DBLP \label{fig:deg_Acm}]{\includegraphics[width=0.3\textwidth, height=3.5cm]{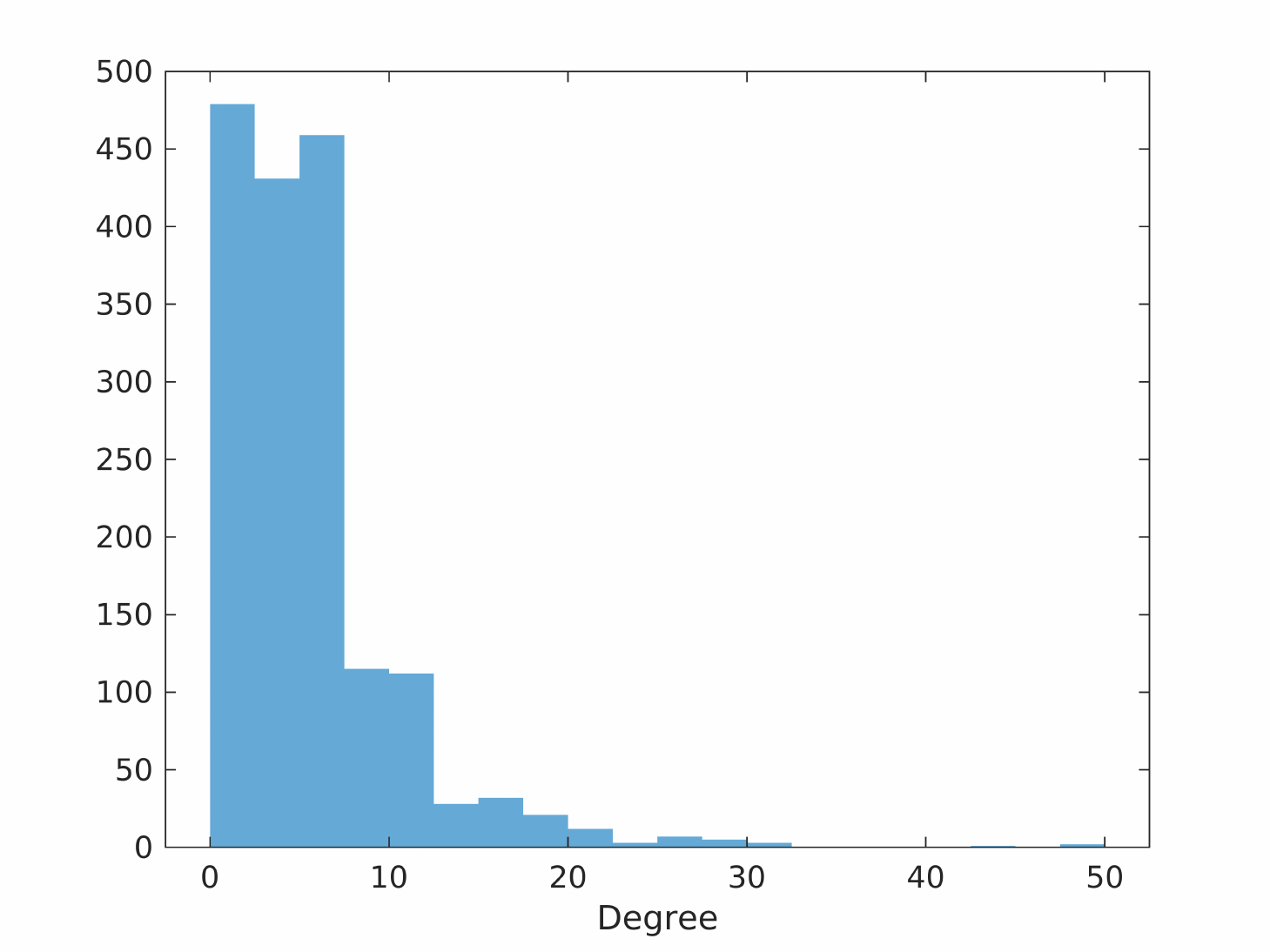}}
    \caption{Degree histogram for additional datasets}
    \label{fig:deg_add}
\end{figure}

We compare our algorithm with baselines on two additional real-world datasets summarized in Table \ref{tbl:add_dataset}. For the `Flickr' dataset \citep{10.1145/2806416.2806501}, the parent graph is given by the network of the image hosting and sharing website, where each user is treated as a vertex and the edges are connected between the vertices if they follow each other. This data consists of 7,575 vertices, 239,738 edges, and 9 communities. 
ACM-DBLP data \citep{8443159} consists of two correlated networks regarding the papers posted on ACM and DBLP in 2016, respectively. 
Authors are treated as vertices and an edge is added when two authors are co-authors. 
The community label of each author is given by their research area. 
This data consists of two graphs of different sizes, 9,872 / 9,916 vertices with $\sim$ 6,000 ground-truth pairs. Since the baseline algorithm GRAMPA runs on correlated graphs of the same size, we only choose vertices in the ground-truth pairs to compare the performance of our algorithm and the baselines. 
 Thus, we use the correlated graphs consisting of 6,299 vertices, 24,822 / 27,892 edges, and 4 communities. The correlations between the two graphs are 0.9511 / 0.8464 from the perspective of each graph.

Figure \ref{fig:deg_add} shows the histograms of vertex degrees within the comparison set $C_k$ for these two datasets. We use the logarithm of the degree when computing the signature vectors in our algorithm, since the degree distributions of these two additional datasets follow a power law like the BlogCatalog dataset considered in the main experiment.
In the case of the ACM-DBLP dataset, each community has a size of about 1,600 vertices, which is twice the size of the BlogCatalog data, and the average degree within the comparison set $C_k$ is lower, indicating a sparser network compared to the BlogCatalog. Therefore, we set the hyperparameter $\ell=3$ instead of $\ell=2$. In addition, since the maximum number of communities that can be used to construct signatures is $4-1$ (the number of communities $-1$), we set $k'=3$. On the other hand, in the case of the Flickr dataset, we set $\ell=2$, the same as before. When comparing the degree distributions (Figure \ref{fig:deg_fli} and Figure \ref{fig:deg_blog}), we observe that the Flickr data has a long-tail distribution compared to the BlogCatalog data. In order to evenly partition the neighbors of high-degree nodes belonging to the long tail for the Flickr dataset, we set $k'$ to be a larger value, 6 instead of 4.

\begin{figure}[!htb]
\centering
	\subfloat[Flickr \label{fig:flickr}]{\includegraphics[width=0.35\textwidth, height=4.5cm]{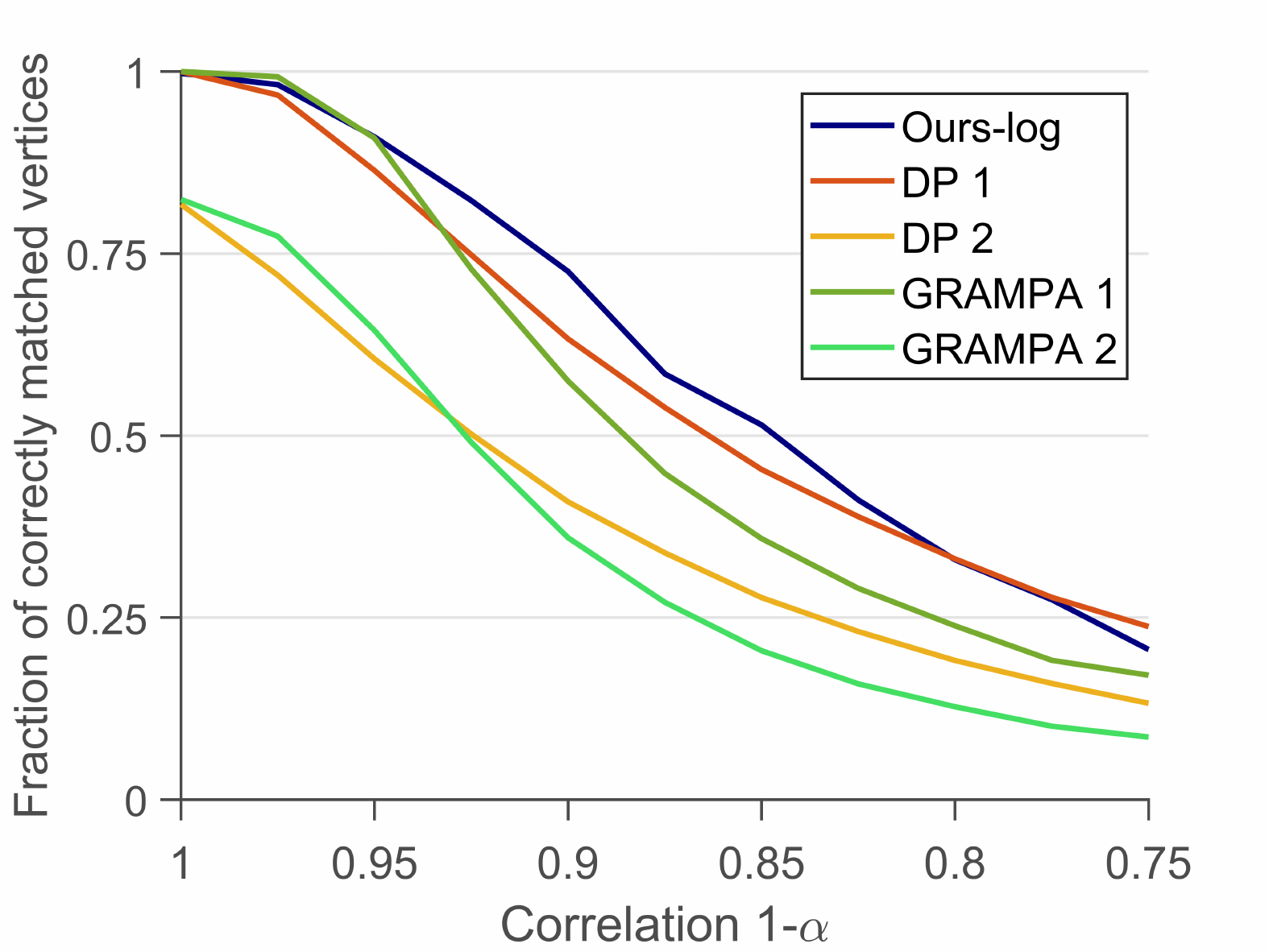}} 
    \subfloat[ACM-DBLP \label{fig:pb}]{\includegraphics[width=0.3\textwidth, height=4.5cm]{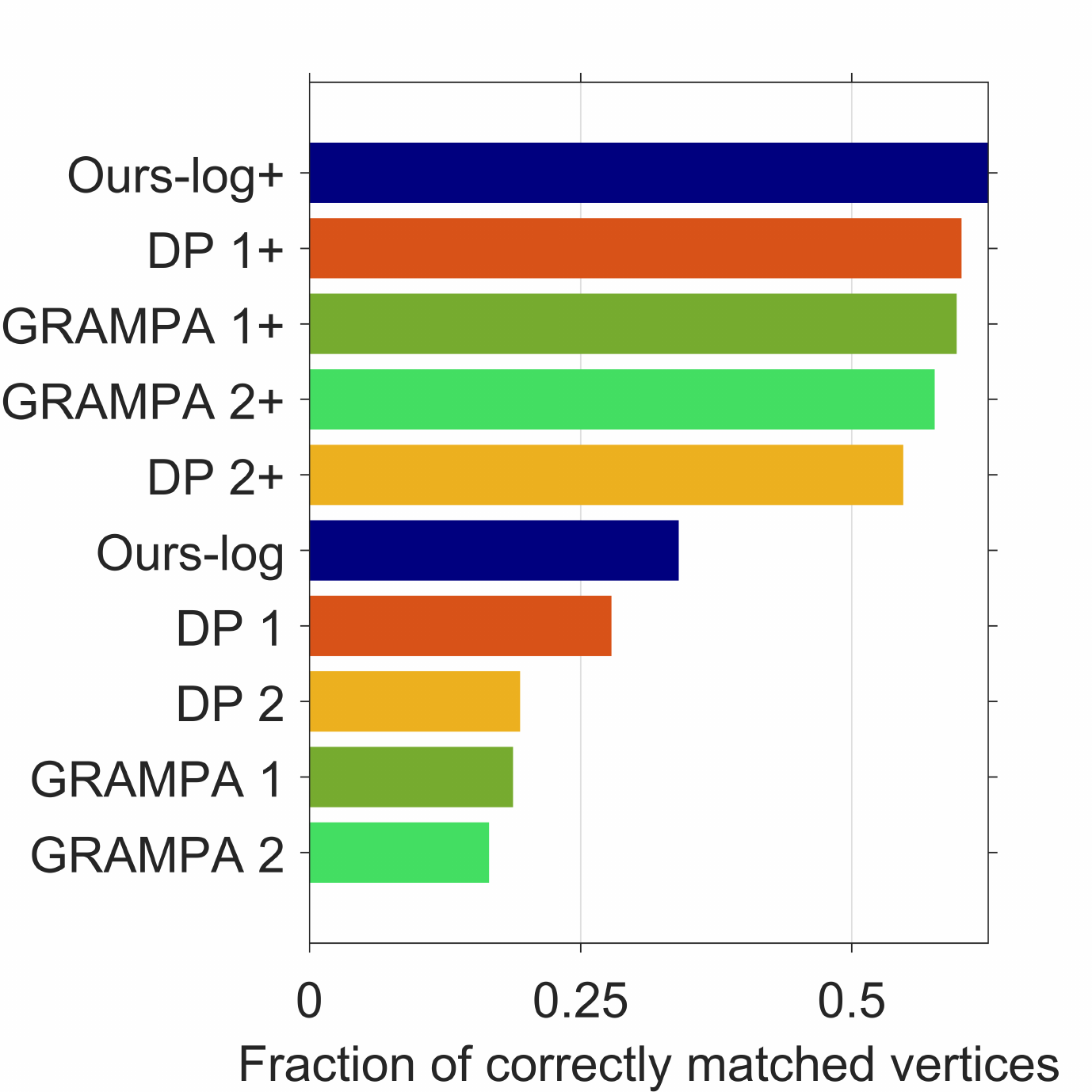}} 
    \caption{Comparison of our partition tree algorithm with other baselines, GRAMPA and Degree Profile (DP) for matching networks of (a) Flickr (b) ACM-DBLP}
    \label{fig:add_exp}
\end{figure}	

In Figure \ref{fig:flickr}, the result shows the fraction of correctly matched vertices within the comparison set $C_k$ on the Flickr dataset with different correlations, averaged over 20 independent runs. Although the performance gap between our algorithm and other algorithms gets smaller, our algorithm still outperforms other algorithms on this dataset with a long-tail distribution of vertex degrees.

For ACM-DBLP data, Figure \ref{fig:pb} shows the fraction of correctly matched vertices within the comparison set $C_k$ (without $+$ mark) and the accuracy after final matching, which is indicated with $+$ mark. For this dataset, since there are some vertices that do not have any intra-community edges, 
we added the refinement matching step using an adjacency matrix for all vertices in the final step. For this dataset, our algorithm also achieves the best performance.

\newpage	
\section{Proof of Corollary \ref{cor:main_all}}\label{app:sec:cor_main_all}
When the community structure is not known in both graphs $G^\pi$ and $G'$, we first perform community detection on each graph separately. Then, we apply our graph matching algorithm to the two correlated graphs, using the recovered community labels.
To ensure the independence between the two-step procedures, we use a random edge splitting approach in both graphs $G^\pi$ and $G'$. Specifically, we randomly partition the edges of each graph into two sub-graphs $(G^{\pi}_1,G^{\pi}_2)$ and $(G'_1,G'_2)$ with probabilities $\beta$ and $1-\beta$ respectively, where $\beta\in(0,1)$. We then apply the community detection algorithm from \citet{yan2018provable} to $G^{\pi}_1$ and $G'_1$, and our graph matching algorithm to $(G^{\pi}_2,G'_2)$. This ensures that the edges used for community detection are not reused in the graph matching process.

Note that both $G^{\pi}_1$ and $G_1'$ are graphs generated by the stochastic block model, with intra- and inter-community edge densities of $p\beta$ and $q\beta$, respectively. The result from \citep{yan2018provable} can be expressed in the following simple form.
\begin{theorem}[Theorem 1 in \citep{yan2018provable}]
    There exists an absolute constant $C_0>0$ with the properties below. Let $n_{\min}$ be the minimum community size. Let $B$ be the $k\times k$ symmetric matrix, whose $(i,j)$-th entry, $B_{ij}$, is equal to the probability that a node in community $i$ is connected to a node in community $j$.  If 
$$
\begin{aligned}
\min _k\left(B_{k k}-\max _{l \neq k} B_{k l}\right) \geq & 2 \sqrt{6 \log n} \max _k \sqrt{B_{k k} / n_k}+6 \max _{1 \leq k<l \leq r} \sqrt{B_{k l} \log n / n_{\min }}  +C_0 \sqrt{\left(n / n_{\min }^2\right)\left(\max B_{k k}\right)},
\end{aligned}
$$
then the exact community detection is possible with high probability by using the semidefinite programming (SDP) proposed in \citep{yan2018provable}.
\end{theorem}

By applying the result from \citet{yan2018provable}, under the condition \eqref{eq:main cor dense nmin}, we can complete the community detection on the graphs $G^{\pi}_{1}$ and $G'_{1}$ with high probability, where $M_{1}$ is a sufficiently large constant depending only $\beta$. Because of the assumption that the sizes of all the communities are different, once the communities are detected, we can complete the matching between the communities of the two graphs just by comparing their sizes.

Note that $G^{\pi}_{2} = G^{\pi} \backslash \caE(G^{\pi}_{1})$ and $G'_{2}= G' \backslash \caE(G'_{1})$ where $\caE(G)$ is the edge set of $G$. Then, $(1-\alpha)(1-\beta)$ is the correlation between the two graphs $G^{\pi}_{2} $ and $G'_{2}$, and the community structures are now revealed in both  $G^{\pi}_{2}$ and $G'_{2}$. If we choose $\beta=\alpha_1/2$ and $\alpha \in (0,\alpha_1/2)$, the correlation $(1-\alpha)(1-\beta)$ is greater than $1-\alpha_1$. The conditions \eqref{eqn:main_thm_cond1}, \eqref{eqn:main_thm_cond2}, and \eqref{eqn:main_thm_cond3} in Theorem \ref{thm:main_dense} are all satisfied on the graphs $G^{\pi}_2$ and $G'_2$ for any constant $\beta$ when $n$ is large enough.  Moreover, combining the maximum density condition in \eqref{eqn:main_thm_cond1} with the condition for community detection in \eqref{eq:main cor dense nmin} gives $n_{\min} =\Omega(n^{20/21})$, which is a stricter condition than $n_{\min}=\Omega(n^{10/19})$. Thus, all conditions in Theorem \ref{thm:main_dense} are satisfied in both $G^{\pi}_{2}$ and $G'_{2}$, showing that the exact matching of ($G^{\pi}_{2}, G'_{2}$) is possible by using our graph matching algorithm. This completes the proof.

\section{Proof of Theorem \ref{thm:seeded graph matching}}\label{app:sec:thm_seeded}
Let the set of seed node pairs given by $\pi_s:[n_s]\to[n_s]$ over $C_s$ be
	\begin{equation}
	    S:=\{(\pi_s(u_1),u_1),(\pi_s(u_2),u_2),\ldots,(\pi_s(u_{n_s}),u_{n_s}) \}.
	\end{equation}
Let $V_{t}$ and $V'_t$ be the vertex sets in the community $C_{t}$ of the graph $G^{\pi}$ and $G'$, respectively. Let $\mathcal{E}(G^{\pi})$ and $\mathcal{E}(G')$ be the sets of edges in the graphs $G^{\pi}$ and $G'$, respectively. For every vertex pair $(i,i') \in V_{t} \times V'_{t}$, we define the weight function 
	\begin{equation}
	    w(i,i'):=\left|\{j\in[n_s] \; | \; \exists z\in \caN_{G^{\pi}(C_{t})}(i) \text{ and }  \exists z'\in \caN_{G'(C_{t})}(i')  \text{ such that } (z,\pi_s(u_{j})) \in \mathcal{E}(G^{\pi}) \text{ and } (z',u_{j}) \in \mathcal{E}(G')   \}\right|.
	\end{equation}
	For a correct pair $(\pi(i),i) \in V_{t} \times V'_{t}$, the distribution of $\left|\mathcal{N}_{G^{\pi}(C_{t})}(\pi(i))\cap \mathcal{N}_{G'(C_{t})}(i)\right|$ is given by Binomial$\left(n_{t},p(1-\alpha)\right)$. Thus, applying the bounds on binomial tail (Theorem \ref{thm:binomial tail}), we have
 \begin{equation}
     \left|\mathcal{N}_{G^{\pi}(C_{t})}(\pi(i))\cap \mathcal{N}_{G'(C_{t})}(i)\right| \geq \frac{1}{2}n_tp(1-\alpha)
 \end{equation}
 with probability at least $1-{1}/{n_t^2}$, where $\alpha < 1/10$ and $n_tp\geq K\log n_t $ for a sufficiently large constant $K>0$.
 Thus, for any seed pair $\left(\pi_s(u_{j}),u_{j}\right)$, we have
 \begin{equation}
 \begin{aligned}
     &\P\left\{ \exists z\in \mathcal{N}_{G^{\pi}(C_{t})}(\pi(i))\cap \mathcal{N}_{G'(C_{t})}(i) \text{ s.t. } (z,\pi_s(u_{j})) \in \mathcal{E}(G^{\pi}) \text{ and } (z,u_{j}) \in \mathcal{E}(G') \right\}\\
     &=1-\left(1-q(1-\alpha)\right)^{\left|\mathcal{N}_{G^{\pi}(C_{2})}(\pi(i))\cap \mathcal{N}_{G'(C_{2})}(i)\right|}\\
     & \geq 1-\left(1-q(1-\alpha)\right)^{\frac{1}{2}n_{t}p(1-\alpha)}\\
     & =q(1-\alpha)\left(1+(1-q(1-\alpha))+(1-q(1-\alpha))^2+\cdots+\left(1-q(1-\alpha)\right)^{\frac{1}{2}n_{t}p(1-\alpha)-1}\right)\\
          &\stackrel{(a)}{\geq} \frac{1}{2}n_{t}pq(1-\alpha)^{2}\left(1-\frac{1}{2}n_{t}pq(1-\alpha)^{2}\right) \stackrel{(b)}{\geq} \frac{1}{3}n_{t}pq(1-\alpha)^{2} \geq \frac{1}{4}n_{t}pq.
 \end{aligned}
      \end{equation}
      where $(a)$ holds since $(1-x)^n \geq 1-nx$ for any $x<1$, $(b)$ holds since $n_{t}pq \leq 2/3$, and the last inequality holds since $\alpha<1/10$.
 Since the weight function $w(\pi(i),i)$ counts the number of seed pairs connected to at least one neighboring node pair of $(\pi(i),i)$, we have
	\begin{equation}
	    \begin{aligned}
	        \P\left\{w(\pi(i),i) \leq \frac{1}{8}n_{s}n_{t}pq \right\}  & \leq \P\left\{\operatorname{Bin}\left(n_{s},\frac{1}{4}n_{t}pq \right) \leq \frac{1}{8}n_{s}n_{t}pq \right\}  \stackrel{(a)}{\leq} \frac{1}{n_{t}^{2}},
	    \end{aligned}
	\end{equation}
	where $(a)$ holds by Theorem \ref{thm:binomial tail} with $n_sn_tpq\geq K\log n_t$ for a sufficiently large constant $K>0$. By applying a union bound over all $i\in V_t$, we have 
\begin{equation}\label{eq:seed correct}
	    w(\pi(i),i)>\frac{1}{8}n_{s}n_{t}pq
	\end{equation}
	for all correct pairs with probability at least $1-\frac{2}{n_{t}}$.
	
	For a wrong pair $(\pi(i),i') \in V_{t} \times V'_{t}$ with $i\neq i'$, we have $\left|\mathcal{N}_{G^{\pi}(C_{t})}(\pi(i))\cap \mathcal{N}_{G'(C_{t})}(i')\right|\sim \operatorname{Bin}\left(n_{t},p^{2}(1-\alpha)^{2} \right)$. Thus,
 \begin{equation}
    \begin{aligned}
         \P\left\{\left|\mathcal{N}_{G^{\pi}(C_{t})}(\pi(i))\cap \mathcal{N}_{G'(C_{t})}(i')\right| \geq 4n_{t}p^{2} \vee 8\log n_{t} \right\} &= \P\left\{\operatorname{Bin}\left(n_{t},p^{2}(1-\alpha)^{2} \right) \geq 4n_{t}p^{2} \vee 8\log n_{t} \right\} \\
         &\leq \P\left\{\operatorname{Bin}\left(n_{t},p^{2} \right) \geq 4n_{t}p^{2} \vee 8\log n_{t} \right\} \stackrel{(a)}{\leq} \frac{1}{n_{t}^{3}},
    \end{aligned}
 \end{equation}
where $(a)$ holds from the binomial tail bounds (Theorem \ref{thm:binomial tail}). Moreover, from $\left|\mathcal{N}_{G_{0}(C_{t})}(i)\right| \sim \operatorname{Bin}\left(n_t,p/(1-\alpha) \right)$, we can show that
  \begin{equation}\label{eq:max neigh}
 \begin{aligned}
      \P\left\{ \left|  \mathcal{N}_{G_{0}(C_{t})}(i) \right| \geq 2n_tp \right\} &= \P\left\{\operatorname{Bin}\left(n_t,p/(1-\alpha) \right) \geq 2n_tp \right\} \leq \frac{1}{n_t^3},
 \end{aligned}
 \end{equation}
 where the last inequality holds by Theorem \ref{thm:binomial tail} since $\alpha<1/10$ and $n_tp\geq K\log n_t $ for a sufficiently large constant $K>0$.
\eqref{eq:max neigh} implies that 
$$\left|\mathcal{N}_{G^{\pi}(C_{t})}(\pi(i))|\vee |\mathcal{N}_{G'(C_{t})}(i')\right| \leq 2n_{t}p $$
holds with probability at least $1-\frac{1}{n^{3}_{t}}$. Thus, for any seed pair $\left(\pi_s(u_{j}),u_{j})\right)$ and $i\neq i'$, we have
\begin{equation}
    \begin{aligned}
        &\P\left\{\exists z\in \mathcal{N}_{G^{\pi}(C_{t})}(\pi(i)) \text{ and }\exists  z'\in \mathcal{N}_{G'(C_{t})}(i') \text{ s.t. } (z,\pi_s(u_{j})) \in \mathcal{E}(G^{\pi}) \text{ and } (z',u_{j}) \in \mathcal{E}(G') \right\} \\
        & \stackrel{(a)}{\leq} \left\{1-(1-q(1-\alpha))^{4n_{t}p^{2} \vee 8 \log n_{t}} \right\}+ \left\{1-(1-q)^{2n_{t}p} \right\}^{2} \stackrel{(b)}{\leq} 4n_{t}p^{2}q\vee 8q \log n_{t} +4n_t^{2}p^{2}q^{2} \\
        & \leq 16\left( q \log n_t \vee n_t^{2}p^{2}q^{2} \vee n_tp^{2}q \right) \leq \frac{1}{16}n_{t}pq.
    \end{aligned}
\end{equation}
The inequality $(a)$ holds since the first term is the probability that there exists a common neighbor pair that is connected to the seed pair and the second term is the probability that two different neighboring nodes $z\in\mathcal{N}_{G^{\pi}(C_{t})}(\pi(i))$ and $z'\in \mathcal{N}_{G'(C_{t})}(i')$ are connected to the seed pair. The inequality $(b)$ holds since $(1-x)^{n} \geq 1-nx$ for $x<1$, and the last inequality holds since $n_tp\geq K\log n_t $ for a sufficiently large constant $K$, $n_tpq\leq 1/256$ and $p\leq 1/256$. Therefore, applying Theorem \ref{thm:binomial tail}, we get
\begin{equation}
    \begin{aligned}
        \P\left\{w(\pi(i),i') \geq \frac{1}{8}n_{s}n_{t}pq \right\} & \leq  \P\left\{\operatorname{Bin}\left(n_{s},\frac{1}{16}n_{t}pq\right) \geq \frac{1}{8}n_{s}n_{t}pq \right\} \leq \frac{1}{n_{t}^{3}},
    \end{aligned}
\end{equation}
 where $n_sn_tpq\geq K\log n_t $ for a large enough constant $K$. Applying the union bound over $(\pi(i),i') \in V_{t} \times V'_{t}$ with $i\neq i'$, we have
	\begin{equation}\label{eq:seed wrong}
	    w(\pi(i),i')<\frac{1}{8}n_{s}n_{t}pq
	\end{equation}
for all wrong pairs with probability at least $1-\frac{2}{n_{t}}$. By combining \eqref{eq:seed correct} and \eqref{eq:seed wrong}, the proof is complete.
	
\section{Proof of Corollary \ref{cor:combined}}\label{app:sec:cor_combined}
    Applying Corollary \ref{cor:exact_smallest}, we obtain the permutation $\hat{\pi}_k$ over $C_k$ such that $\hat{\pi}_k=\pi|_{C_k}$ with probability at least $1-2n_{\min}^{-D}$. 
    From the assumptions \eqref{eqn:main_thm_cond1}, \eqref{eqn:main_thm_cond2} and \eqref{eqn:main_thm_cond3}  in Theorem \ref{thm:almost_smallest}, all conditions of Theorem \ref{thm:seeded graph matching} except $n_{t}pq\leq 1/256$ are satisfied. The condition $n_{t}pq\leq 1/256$ can be satisfied when $n_{\min} \geq M_1 n^{10/19}$ for a sufficiently large constant $M_1>0$, since for any $t \in [k]$, we have
    \begin{equation}\label{eq:condition}
        n_{t}pq \leq npq \leq np^{2} \stackrel{(a)}{\leq } M_1^{-19/10}n_{\min}^{19/10}p^{2}= M_1^{-19/10}\frac{(n_{\min}p)^{2}}{n^{1/10}_{min}} \leq M_1^{-19/10},
    \end{equation}
    where the inequality $(a)$ holds when $n_{\min}\geq M_1n^{10/19}$ and the last inequality holds since $n_{\min}p \leq n_{\min}^{1/20}$ from \eqref{eqn:main_thm_cond1}. 
    
    Let us choose a community $C_{r}$ that has not been used to generate the signature vector for vertices in $C_k$. Then, applying Theorem \ref{thm:seeded graph matching} with $C_{k}$, we can obtain the permutation $\hat{\pi}_r$ over $C_r$ such that $\hat{\pi}_r=\pi|_{C_r}$ with probability at least $1-4/n_{r}$. Similarly, we can complete the exact matching over the remaining communities with probability at least $1-\sum^{k}_{i=1}\frac{4}{n_i}$  by applying Theorem \ref{thm:seeded graph matching} with $C_{r}$ as seeds. Since $n_{\min}\geq M_1 n^{10/19}$, the number of communities $k$ should be less than  $\frac{1}{M_1}n^{9/19}$. Thus, we have
    \begin{equation}
        1-\sum^{k}_{i=1}\frac{4}{n_{i}} \geq 1-\frac{4k}{n_{\min}} \geq 1-\frac{1}{n^{1/19}},
    \end{equation}
    which completes the proof.

\begin{remark}[Relation between $q$ and $n_{\min}$]
    The reason we proposed a new seeded matching algorithm (Algorithm \ref{alg:algorithm5} beyond \citep{KL13}) and proved Theorem \ref{thm:seeded graph matching} was to do the seeded matching at a smaller $q$ such as $n_{\min}q= M' (\log \log n_{\min})^2$ for a sufficiently large constant $M'>0$. But if $q$ is large enough such as  $n_{\min}q  \geq C\log n$ for a sufficiently large constant $C>0$, then the seeded matching can be performed by simply counting the number of common seed pairs in the 1-hop neighbors as proposed by \citet{KL13}. In this case, $n_{\min}\geq M_1 n^{10/19}$ is not required.
\end{remark}

\subsection{Balanced Community Size}\label{sec:balanced_community_size}

  In Remark \ref{rmk:general density}, we discussed the case where the community sizes are approximately balanced, i.e., $n_{\min}=\Theta(n/k)$, and derived the conditions to guarantee the exact recovery of the underlying permutation for the correlated SBMs with constant correlation within a low-order polynomial-time complexity. We prove that when the community label is given stochastically with a probability distribution $[r_1,r_2,\dots, r_k]$ with each $r_i=\Theta(1/k)$ the community sizes are approximately balanced and the sizes of all communities are different  with high probability as mentioned in Remark \ref{rmk:balance}
 \begin{lemma}\label{lem:proof rmk}
    Suppose that there are $n$ vertices and $k$ communities. Let $r_{i}$ be the probability that each vertex belongs to community $C_{i}$, where 
    \begin{equation}
        \sum^{k}_{i=1}r_{i}=1 \text{ and } r_{i}=\Theta\left( \frac{1}{k} \right).
    \end{equation}
    If the number of communities
    \begin{equation}\label{eq:rmk k}
       k =o(n^{1/5})
    \end{equation}
    then the communities are approximately balanced, i.e., $n_{i}=\Theta(n/k)$ for all $i\in[k]$, and the sizes of all communities are different with high probability as $n\to \infty$.
 \end{lemma}
 \begin{proof} Without loss of generality, assume that $r_{1}\geq \ldots \geq r_{k}$. Let $n_{i}$ be the community size of $C_{i}$ and $n_{\min}$ be the smallest community size. Define an event 
 \begin{equation*}
      \mathcal{F}:=\left\{n_{\min} \geq \frac{1}{2}nr_{k}\right\}.
 \end{equation*}

 Then, we can get
 \begin{equation}\label{eq:nmin bound}
 \begin{aligned}
     \P\{ \mathcal{F}^{c}\} &\leq \sum^{k}_{i=1} \P\{ n_{i} \leq \frac{1}{2}nr_{k} \} \leq k\P\left\{n_{k} \leq \frac{1}{2}nr_{k} \right\} \\
     & \leq k \P\left\{\operatorname{Bin}(n,r_{k})  \leq \frac{1}{2}nr_{k}\right\} \stackrel{(a)}{\leq} k \exp(-3\log n) \leq \frac{k}{n^{3}}.
 \end{aligned}
      \end{equation}
      The inequality $(a)$ holds by Theorem \ref{thm:binomial tail} with $\frac{n}{k}=\omega(n^{4/5})$. Moreover, we can show that $n_{\max} \leq 2nr_{1}$ with probability at least $1-\frac{k}{n^3}$ using a similar approach as in \eqref{eq:nmin bound}. Therefore, $n_i=\Theta(n/k)$ holds with high probability as $n\to \infty$.
      
 For any $i,j\in[k]$, $i\neq j$, we have
 \begin{equation}\label{eq:event nmin}
 \begin{aligned}
     \P\{ n_{i}=n_{j} | \mathcal{F}\} &\leq \sum_{t\geq \frac{1}{2}nr_{k}} \P\{ n_{i}=n_{j} | n_{i}+n_{j}=2t \}\P\{n_{i} + n_{j}=2t\}\\
     & \leq \max_{t\geq \frac{1}{2}nr_{k}} \P\{ n_{i}=n_{j} | n_{i}+n_{j}=2t \} \\
     & \leq \max_{t\geq \frac{1}{2}nr_{k}} {2t \choose t} \left(\frac{r_{i}}{r_{i}+r_{j}} \right)^{t} \left(\frac{r_{j}}{r_{i}+r_{j}} \right)^{t} \\
     & \leq \max_{t\geq \frac{1}{2}nr_{k}} {2t \choose t} \frac{1}{2^{2t}} \stackrel{(a)}{=}O\left(\frac{1}{\sqrt{nr_{k}}}\right).
 \end{aligned}
\end{equation}
 The equality (a) holds by Stirling's approximation.  Thus, we have
 \begin{equation}
     \P\{ n_{i}=n_{j} \} \leq \P\{ n_{i}=n_{j} | \mathcal{F}\} + \P\{\mathcal{F}^{c}\} =O\left(1/\sqrt{nr_{k}}\right)=O\left(\sqrt{k/n}\right).
 \end{equation}
 Applying a union bound over $i,j$, we can show that the community sizes are all different with probability $1-O\left(k^{2}\sqrt{\frac{k}{n}}\right)$. We have $k^{2}\sqrt{\frac{k}{n}} \to 0 $ as $n\to \infty$ since \eqref{eq:rmk k}. 
\end{proof}
In Remark \ref{rmk:general density},  we assumed that $k\leq C'\sqrt{\frac{n}{p}}(p-q)$. On the condition $np\leq n^{1/20}$, we have
\begin{equation*}
    C'\sqrt{\frac{n}{p}}(p-q)\leq C'\sqrt{np} \leq n^{1/40} =o(n^{1/5})
\end{equation*}
as $n\to \infty$. Therefore, the condition \eqref{eq:rmk k} in Lemma \ref{lem:proof rmk} can be satisfied.

\section{Proof of Theorem \ref{thm:almost_smallest}}\label{sec:proof of almost}
From now on, assume that the permutation $\pi:[n]\to[n]$ is the identity, so that $G^{\pi}=G$.
Let $C_k$ be the smallest community and $m$ denote the size of the community $C_k$, i.e, $m=n_{\min}$. 

 In proving Theorem \ref{thm:almost_smallest} (almost exact matching of $G^\pi(C_k)$ and $G'(C_k)$), we follow the proof structure mostly similar to that of \citet{MRT21a}, where the almost exact matching of the correlated ER model was proved based on the correlation between the signature vectors of the correct pair of vertices, where the signature vector is defined in terms of the binary partition tree rooted from each vertex. The main difference of our analysis from that of \citet{MRT21a} is that we use both inter- and intra-community edges when constructing the $2^{k'}$-ary partition tree, while that of \citet{MRT21a} uses only inter-community edges (since there is no community structure in the ER model) both in generating the binary tree and in partitioning the vertices within the tree. 
In our analysis,  the intra-community edges are used to generate the tree structure, while the inter-community edges are used to partition the vertices at each level of the tree into different nodes based on their inter-community edge degrees. The proof becomes simpler for our analysis of the correlated SBMs due to this difference.

In this section, we present two main lemmas, Lemma \ref{lem:correct final} and Lemma \ref{lem:wrong final}, to prove Theorem \ref{thm:almost_smallest}. Lemma \ref{lem:correct final} shows that for most correct vertex pairs, the value of  the normalized distance $\sum_{\bss^{\ell} \in J} \frac{\left(f_{\bss^{\ell}}(i)-f'_{\bss^{\ell}}(i)\right)^{2}}{\mathrm{v}_{\bss^{\ell}}(i)+\mathrm{v}'_{\bss^{\ell}}(i)}$ between the signature vectors in Algorithm \ref{alg:algorithm2} is less than the threshold $2w\left(1-\frac{1}{\sqrt{\log m}}\right)$ with high probability. 

\begin{lemma}\label{lem:correct final}

For any constants $C, D > 0$, there exist constants $R, Q_1, Q_2, m_0, \alpha_1, c > 0$ with the properties below. Consider the graphs $G^\pi$ and $G'$, which are defined as the correlated Stochastic Block Models (SBMs) described in Section \ref{sec:model} with correlation $1 - \alpha$. 
Consider a random subset $J$ uniformly drawn from $\{-1,1\}^{k'\ell}$ with cardinality $2w$, where $w$ is an integer.
Suppose that $m \geq m_0$, $\alpha \in (0, \alpha_1)$, $w \geq (\log m)^4$, and the following conditions hold:
    \beq
    \begin{split}\nonumber
    (\log m)^{1.1} &\leq mp \leq m^{1/20},\\
    mq&\geq Q_1 k'^2 \ell^2,\\
  2\log \left(w^{4}(\log m)\right) \leq k'\ell \leq C \log \log m,& \quad k'\leq Q_2 \log mp,  \quad \ell \leq \frac{\log m}{R \log mp} \wedge C \log \log m.
    \end{split}
    \eeq
  Then, there are at least $m-m^{1-c}$ vertices $i \in C_k$ satisfying
  \begin{equation}\label{eq:correct result}
      \sum_{\bss^\ell \in J} \frac{\left(f_{\bss^\ell}(i)-f'_{\bss^\ell}(i)\right)^{2}}{\mathrm{v}_{\bss^\ell}(i)+\mathrm{v}'_{\bss^\ell}(i)} \leq 2 w\left(1-\frac{1}{(\log m)^{0.1}}\right) ,
  \end{equation}
  with probability at least $1-m^{-D}$.		
\end{lemma}
The proof of Lemma \ref{lem:correct final} can be found in Section \ref{sec:proof of correct final}.


Lemma \ref{lem:wrong final}, on the other hand, demonstrates that for most of the incorrect pairs $(i,i')$ in $C_k$, the normalized distance value $\sum_{\bss^\ell \in J} \frac{\left(f_{\bss^\ell}(i)-f'_{\bss^\ell}(i')\right)^{2}}{\mathrm{v}_{\bss^\ell}(i)+\mathrm{v}'_{\bss^\ell}(i')}$ between the signature vectors exceeds the threshold $2w\left(1-\frac{1}{\sqrt{\log m}}\right)$ with high probability.

\begin{lemma}\label{lem:wrong final}
    For any constants $C, D>0$, there exists constant $R,Q_1,Q_2, m_{0},\alpha_{1},c>0$ with the properties below.  Consider the graphs $G^\pi$ and $G'$, which are defined as the correlated Stochastic Block Models (SBMs) described in Section \ref{sec:model} with correlation $1 - \alpha$. Consider a random subset $J$ uniformly drawn from $\{-1,1\}^{k'\ell}$ with cardinality $2w$, where $w$ is an integer.  Suppose that $m\geq m_{0}, \;\alpha \in(0, \alpha_{1})$, $w \geq \lfloor (\log m)^{5} \rfloor$, and the following conditions hold:
        \beq
    \begin{split}\nonumber
    (\log m)^{1.1} &\leq mp \leq m^{1/20},\\
    mq&\geq Q_1 k'^2 \ell^2,\\
   k'\ell \leq C \log \log m, \quad k'\leq Q_2& \log mp,  \quad \ell \leq \frac{\log m}{R \log mp} \wedge C \log \log m.
    \end{split}
    \eeq
  Then, there exists a subset $\mathcal{I} \subset C_k$ with a size $|\mathcal{I}| \geq m-m^{1-c}$ such that for any $i,i' \in \mathcal{I}$, $i\neq i'$, we have 
  \begin{equation}\label{eq:wrong result}
      \sum_{\bss^\ell \in J} \frac{\left(f_{\bss^\ell}(i)-f'_{\bss^\ell}(i')\right)^{2}}{\mathrm{v}_{\bss^\ell}(i)+\mathrm{v}'_{\bss^\ell}(i')} \geq 2 w\left(1-\frac{1}{(\log m)^{0.9}}\right) ,
  \end{equation}
  with probability at least $1-m^{-D}$.		
\end{lemma}
The proof of Lemma \ref{lem:wrong final} can be found in Section \ref{sec:proof of wrong final}.

\begin{proof}[Proof of Theorem \ref{thm:almost_smallest}]
We will set the order of $\ell$ as large as possible to minimize the order of $k'$ (or $k$).  Let $\ell= \left\lceil \frac{\log m}{40 \log mp}\right\rceil \wedge \lceil 42\log \log m \rceil$ and $k'=\left\lceil 1680 \log \log m \frac{\log mp}{\log m} \right\rceil $. We consider a random subset $J\in \{-1,1\}^{k'\ell}$ of a size $|J|=2w$ where $w =\left\lfloor(\log m)^{5}\right\rfloor$. 
By Lemma \ref{lem:correct final} and \ref{lem:wrong final}, we can find a subset $\mathcal{I} \subset C_k$ with its size $|\mathcal{I}| \geq m-m^{1-c}$ for a positive constant $c$ such that 	$\sum_{\bss^\ell \in J} \frac{\left(f_{\bss^\ell}(i)-f'_{\bss^\ell}(i)\right)^{2}}{\mathrm{v}_{\bss^\ell}(i)+\mathrm{v}'_{\bss^\ell}(i)} <  2 w\left(1-\frac{1}{\sqrt{\log m}}\right)$ and $	\sum_{\bss^\ell \in J} \frac{\left(f_{\bss^\ell}(i)-f'_{\bss^\ell}(i')\right)^{2}}{\mathrm{v}_{\bss^\ell}(i)+\mathrm{v}'_{\bss^\ell}(i')} >  2 w\left(1-\frac{1}{\sqrt{\log m}}\right) $ for any distinct $i,i' \in \mathcal{I}$. Hence, by running Algorithm \ref{alg:algorithm2}, we can get a matrix $B$ satisfying that 
$$
B_{\pi(i'), i}=\begin{cases}
    1 &\text{ if } \; i=i' \in \mathcal{I},\\
    0 &\text{ if }  \; i,i'(i\neq i') \in \mathcal{I} .
\end{cases} 
$$
By Proposition 2.2 in \citep{MRT21a} with matrix $B$ as an input, we can obtain a permutation $\hat{\pi}$ on community $C_{k}$ that satisfies
    \begin{equation}
    |i\in [m] : \hat{\pi}(i) \neq \pi(i)| \leq 4m^{1-c}.
    \end{equation}
    This completes the proof.
\end{proof}

\section{Structural Properties of Partition Tree}\label{sec:vertex classes} 
We analyze the statistical behavior of large neighborhoods of vertices in  Erd\H{o}s-R\'enyi (ER) graph and the partition trees rooted from each vertex, which will be used to prove the main theoretical results. 
Section \ref{lem:degree} is dedicated to the statement of lemmas concerning the statistics of the large neighborhoods in the ER graph $\caG(n,p)$. In Section \ref{sec:size of partition tree}, we show that  the sizes of the nodes $T^d_{\bss^d}(i)$ at each level $d\in[\ell]$ of the partition tree are well balanced for every $\bss^d\in\{-1,1\}^{k'd}$. Furthermore, in Section \ref{app:sec:tree} we establish that the $\ell$-neighborhoods of a majority of vertices in $C_k$ form a tree with high probability.

\subsection{Large Neighborhoods in ER Graph}\label{sec:large neighbor}
In this subsection, we state lemmas that analyze the statistics of large neighborhoods in Erd\H{o}s-R\'enyi (ER) graph $\caG(n,p)$.

\begin{lemma}[Sizes of neighbors]\label{lem:degree}
    For any constant $D,R,\delta>0$, there exists a constant $n_{0}$ which depends only on $D,R$ and $\delta$ with the properties below. Consider an Erd\H{o}s-R\'enyi (ER)  graph $G$ defined as $\mathcal{G}(n, p)$. Assume that
    \begin{equation}\label{eq:density}
         n\geq n_0, \quad np \geq  (\log n)^{1+\delta} .
    \end{equation}
    Then, we have that
    \begin{equation}\label{eq:degree 1}
        \left|\mathcal{S}_{G}(i, 1) -np \right| \leq \frac{np}{R \log \log n}	\quad \forall i \in [n] ,
    \end{equation}
    with probability at least $1-n^{-D}$.
\end{lemma}

\begin{proof} Applying Bernstein's inequality (stated in Lemma \ref{lem:bernstein's inequality}), for any vertex $i$ we have
\begin{equation}\label{eqn:SG1}
    \begin{aligned}
        &\P\left\{  \left|\mathcal{S}_{G}(i, 1) -np \right| \geq \frac{np}{R \log \log n} \right\}\leq 2 \exp\left(-R'\frac{np}{(\log \log n)^2} \right) \leq n^{-D-1},
    \end{aligned}
\end{equation}
where $R'$ depends only on $R$. The last inequality holds since $(\log n)^{1+\delta} \leq np $ and $n\geq n_0(D,R,\delta)$. By applying a union bound over all vertices $i\in[n]$, we can complete the proof.
\end{proof}

\begin{lemma}[Sizes of intersections of neighbors (Lemma 4.1 in \citep{MRT21a})] \label{lem:Sizes of neighborhoods and their intersections}
 For any constant $D>1$, there exist constants $K>0$ and $n_{0} $ which depend only on $D$ with the properties below. Consider an Erd\H{o}s-R\'enyi (ER) graph  $G\sim \caG(n, p)$ with $n \geq n_{0}$ and $ n p\geq \log n$. With probability at least $1-n^{-D}$, the following inequality holds:
    \beq\label{eqn:lemma41ref}
    \left|\mathcal{B}_{G}(i, l)\right| \leq K(n p)^{l} \quad \text { for any } i, l \in[n].
    \eeq
    Under the event defined by \eqref{eqn:lemma41ref}, consider any positive integer $m$ and $i, j \in [n]$ satisfying that $i \neq j$ and $G\left(\mathcal{B}_{G}(i, 3m)\right)$ forms a tree. If the distance $d = \operatorname{dist}_{G}(i, j) \leq 2m$, then we have:
    $$
    \left|\mathcal{B}_{G}(i, m) \cap \mathcal{B}_{G}(j, m)\right| \leq K(n p)^{m-\left\lceil d / 2\right\rceil}.
    $$
\end{lemma}


\subsection{Sizes of Nodes in the Partition Tree}\label{sec:size of partition tree}
We next provide the bounds on the sizes of the nodes $\{T_{\bss^{d}}^{d}(i)\}_{\bss^d\in\{-1,1\}^{k'd}, d\in[\ell]}$ in the $2^{k'}$-ary partition tree.

\begin{lemma}[Sizes of nodes in the partition tree]\label{lem:sizes of vertex class} For any constants $D,R,\delta>0$ , there exist a constant $m_0$ depending on $D,R,\delta$, a constant $k_0$ depending on $\delta$ and an absolute constant $Q>0$  with the properties below. Let $G\sim \text{SBM}(n, p, q,k)$ with communities $(C_1,\dots, C_k)$. Assume that the community labels $(C_1,C_2,\dots, C_k)$ are given. Assume that
    \begin{equation}\label{eq:lem sizes of vertex class}
        m\geq m_0,\quad (\log m)^{1+\delta} \leq mp ,  \quad k'\leq k_0 \log mp,\quad mq \geq Qk'^{2}\ell^2,
    \end{equation}
for a fixed positive integer $\ell \leq \left\lfloor \frac{\log m}   {\log mp} \right\rfloor \wedge R\log \log m$. For any $d \in[\ell], \bss^{d} \in\{-1,1\}^{k'd}$, and $i \in C_{k}$, recall $T_{\bss^d}^{d}(i)$ defined in  Sec. \ref{sec:partition_tree}. Then, for any $i \in C_{k}$ satisfying that $G\left(\mathcal{B}_{G(C_{k})}(i, \ell)\right)$ forms a tree, we have
    \beq\label{eqn:bound_node_parttree}
    \left|T_{\bss^{d}}^{d}(i)\right| \leq 6\left(\frac{m p}{2^{k'}}\right)^{d}
    \eeq
    with probability at least $1-m^{-D}$.
\end{lemma}

\begin{proof} For notational simplicity, we will use $G_{k}$ instead of $G(C_{k})$. We will prove \eqref{eqn:bound_node_parttree} using mathematical induction. We start with the base case $d=1$, and by applying Lemma \ref{lem:degree}, we obtain the following result:
\begin{equation}\label{eq:S1}
    \P\left\{ |\mathcal{S}_{G_{k}}(i,1)| \leq mp\left(1+\frac{1}{5\ell}\right) \right\} \geq 1-m^{-D-2}
\end{equation}
since $mp\geq (\log m)^{1+\delta}$ and $m\geq m_{0}(D,R,\delta)$.
Let $\bss_{1} \in \{-1,1\}^{k'}$. 
By applying the bounds from Lemma \ref{lem:binomial 1} and \ref{lem:binomial 2} for a binomial random variable, we can derive the following inequality for $j\in \mathcal{S}_{G_{k}}(i,1)$ and $a\in[k']$:
\begin{equation}\label{eq:signature}
    \left|\P\{\sign(\operatorname{deg}^{a}_{G}(j)-n_{a}q)=1\}-\frac{1}{2}\right| \leq \frac{C}{\sqrt{n_{a}q}} \leq \frac{C}{\sqrt{mq}}
\end{equation}
where $C$ is a universal constant.
Therefore, we can show that
\begin{equation}\label{eq:s}
    \left| \P\{\sign(\operatorname{deg}^{a}_{G}(j)-n_{a}q)= \bss_{1}(a) \text{ for all } a\in[k']\} -\frac{1}{2^{k'}}\right| \leq  \frac{4Ck'}{2^{k'}\sqrt{mq}}
\end{equation}
by applying Lemma \ref{lem:1+n2x} with the condition $mq\geq 4C^2k'^{2}$.
If $mq \geq Q k'^{2}\ell^2$ for a sufficiently large constant $Q$,  we obtain
\begin{equation}\label{eq:mq condition}
    \P\{\sign(\operatorname{deg}^{a}_{G}(j)-n_{a}q)= \bss_{1}(a) \text{ for all } a\in[k']\} \leq  \frac{1}{2^{k'}}\left(1+\frac{1}{5\ell}\right).
\end{equation}
Thus, by applying Hoeffding's inequality (Lemma \ref{lem:Hoeffding}), we have
\begin{equation}\label{eq:T_1}
\begin{aligned}
 \P\left\{ \left|T^{1}_{\bss_1}(i)\right| \leq  2\frac{mp}{2^{k'}}\left(1+\frac{1}{\ell}\right)  \right\} \stackrel{(a)}{\geq} 1-\exp\left(-C'\frac{mp}{2^{2k'} }\right) \geq 1-m^{-D-1},
\end{aligned}
    \end{equation}
where the inequality $(a)$ holds for a constant $C'>0$ by \eqref{eq:S1} and \eqref{eq:mq condition}. The last inequality holds since $k'\leq k_0\log mp$ for a sufficiently small constant $k_0$ depending on $\delta$ such that $2^{2k_0 \log mp} \leq (mp)^{0.5\delta/(1+\delta)}$, $mp\geq (\log m)^{1+\delta}$ and $m\geq m_{0}(D,\delta)$.

 Define an event
\begin{equation}
    \mathcal{F}_{d}:=\left\{\operatorname{deg}_{G_{k}}(j) \leq mp\left(1+\frac{1}{5\ell}\right), \; \forall j\in \mathcal{B}_{G_{k}}(i,d)  \right\}.
\end{equation}
Let $\P_{d}$ denote the conditional probability given the subgraph $G_{k}(\mathcal{B}(i,d))$ under the condition that $\mathcal{F}_{d-1}$ holds.
For $\bss_{u} \in \{-1,1\}^{k'}$, let $\bss^{d}=(\bss_{1},\ldots,\bss_{d})$. Define 
	\begin{equation}
	    Q_{d}:=2\left(1+\frac{1}{\ell} \right)^{d}.
	\end{equation}
	By Lemma \ref{lem:degree}, we have
	\begin{equation}
		\P_{d}\left\{ \mathcal{F}_{d} \text{ occurs} \; \left| \; \left|T^{d}_{\bss^{d}}(i)\right| \leq Q_{d}\left(\frac{mp}{2^{k'}}\right)^{d}  \right\}\right. \geq 1-m^{-D-2}
	\end{equation}
	and
	\begin{equation}
		\left.\P_{d}\left\{ \left|\mathcal{N}_{G_{k}}(T^{d}_{\bss^{d}}(i)) \cap \mathcal{S}_{G_{k}}(i,d+1) \right| \leq |T^{d}_{\bss^{d}}(i)|mp\left(1+\frac{1}{5\ell}\right) \right\rvert \left|T^{d}_{\bss^{d}}(i)\right| \leq Q_{d}\left(\frac{mp}{2^{k'}}\right)^{d} \text{ and } \mathcal{F}_{d} \text{ occurs} \right\} \geq 1-m^{-D-2}.
	\end{equation}
	Thus, we can get
	\begin{equation}\label{eq:cond1}
	\begin{split}
		&\left.\P_{d}\left\{  \left|\mathcal{N}_{G_{k}}(T^{d}_{\bss^{d}}(i)) \cap \mathcal{S}_{G_{k}}(i,d+1) \right| \leq 2^{k'}Q_{d}\left(\frac{mp}{2^{k'}}\right)^{d+1} \left(1+\frac{1}{5\ell}\right) \text{ and }\mathcal{F}_{d} \text{ occurs }  \right\rvert |T^{d}_{\bss^{d}}(i)| \leq Q_{d}\left(\frac{mp}{2^{k'}}\right)^{d} \right\}\\
		& \geq 1-2m^{-D-2}.
	\end{split}
	\end{equation}
	Let $H:=\mathcal{N}_{G_{k}}(T^{d}_{\bss^{d}}(i)) \cap \mathcal{S}_{G_{k}}(i,d+1)$. For $\bss_{d+1} \in \{-1,1\}^{k'}$ and any $j \in H$, 	\begin{equation}
		\P_{d}\{\sign(\operatorname{deg}^{a}_{G}(j)-n_{a}q)= \bss_{d+1}(a) \text{ for all } a\in[k']\} \leq \frac{1}{2^{k'}}\left(1+ \frac{1}{5\ell}\right).
	\end{equation}
	The above result follows from \eqref{eq:mq condition}. Similar to \eqref{eq:T_1}, by applying Hoeffding's inequality we have
	\begin{equation}\label{eq:cond2}
	    \begin{aligned}
	    &\left.\P_{d}\left\{ |T^{d+1}_{\bss^{d+1}}(i)| \leq Q_{d+1}\left(\frac{mp}{2^{k'}}\right)^{d+1}   \right\rvert  |H| \leq 2^{k'}Q_{d} \left(\frac{mp}{2^{k'}} \right)^{d+1}\left(1+ \frac{1}{5\ell}\right), \mathcal{F}_{d} \text{ occurs and  }   |T^{d}_{\bss^{d}}(i)| \leq Q_{d}\left(\frac{mp}{2^{k'}}\right)^{d} \right\}\\
	    &=\left.\P_{d}\left\{ |T^{d+1}_{\bss^{d+1}}(i)| \leq Q_{d+1}\left(\frac{mp}{2^{k'}}\right)^{d+1}   \right\rvert  |H| \leq 2^{k'}Q_{d} \left(\frac{mp}{2^{k'}} \right)^{d+1}\left(1+ \frac{1}{5\ell}\right) \right\}  \geq 1-m^{-D-2}.
	    \end{aligned}
	\end{equation}
	  By \eqref{eq:cond1} and \eqref{eq:cond2}, we finally get
	\begin{equation}\label{eq:T_s}
	\begin{aligned}
		&\left.\P_{d}\left\{ |T^{d+1}_{\bss^{d+1}}(i)| \leq Q_{d+1}\left(\frac{mp}{2^{k'}}\right)^{d+1} \text{ and }   \mathcal{F}_{d} \text{ occurs} \right\rvert |T^{d}_{\bss^{d}}(i)| \leq  Q_{d}\left(\frac{mp}{2^{k'}}\right)^{d} \right\}\\
		&=\left.\P_{d}\left\{ |T^{d+1}_{\bss^{d+1}}(i)| \leq Q_{d+1}\left(\frac{mp}{2^{k'}}\right)^{d+1}   \right\rvert  |H| \leq 2^{k'}Q_{d} \left(\frac{mp}{2^{k'}} \right)^{d+1}\left(1+ \frac{1}{5\ell}\right), \mathcal{F}_{d} \text{ occurs and  }   |T^{d}_{\bss^{d}}(i)| \leq Q_{d}\left(\frac{mp}{2^{k'}}\right)^{d} \right\}\\
		&\quad \times 	\left.\P_{d}\left\{  \left|H \right| \leq 2^{k'}Q_{d}\left(\frac{mp}{2^{k'}}\right)^{d+1} \left(1+\frac{1}{5\ell}\right) \text{ and } \mathcal{F}_{d} \text{ occurs }  \right\rvert |T^{d}_{\bss^{d}}(i)| \leq Q_{d}\left(\frac{mp}{2^{k'}}\right)^{d} \right\} \\
		&\geq 1-m^{-D-1}.
	\end{aligned}
	\end{equation}
	Applying a union bound over \eqref{eq:T_1} and \eqref{eq:T_s} for $d\in[\ell -1]$ and $\bss^{d}\in\{-1,1\}^{k'd}$, we can get
	\begin{equation}
		\P\left\{ |T^{d}_{\bss^{d}}(i)| \leq Q_d\left(\frac{mp}{2^{k'}}\right)^{d} \; \forall d\in[\ell], \; \bss^{d}\in \{-1,1\}^{k'd}  \right\} \geq 1-\Sigma^{\ell}_{d=1} 2^{k'd}m^{-D-1} \geq 1-m^{-D},
	\end{equation}
	where the last inequality holds since $k'\leq k_0 \log mp$ and $\ell \leq \left\lfloor \frac{\log m}   {\log mp} \right\rfloor$. Moreover, we have $Q_\ell =2\left(1+\frac{1}{\ell}\right)^{\ell} \leq 2e \leq    6$. Thus, the proof is complete.
	\end{proof}

	\subsection{Tree Structure and Typical Vertices}\label{app:sec:tree}
Consider the parent graph $G_{0}$, which follows an SBM$(n,p/(1-\alpha),q/(1-\alpha),k)$ with community labels $(C_1,C_2,\ldots,C_k)$. 
Recall that $C_k$ is the smallest community with size $m$.
 
	Following Definition 4.10 in \citep{MRT21a}, we introduce the concept of a ``typical" vertex in the parent graph, which includes certain conditions that will be useful for analyzing the signature vectors.  Compared to Definition 4.10 in \citep{MRT21a}, our definition has fewer conditions, due to the main difference of our algorithm that the edges used to construct the tree structure are not reused to partition the vertices at each level of the tree.

	\begin{definition}\label{def:type}
	We define a vertex $i \in C_{k}$ of the parent graph $G_{0}$ as typical with parameters $\ell \in \mathbb{N}$ and $\epsilon>0$, denoted as $i \in \operatorname{Typ}_{G_{0}(C_{k})}(\ell,\epsilon)$, if it satisfies the following conditions:
		\begin{flalign*}
        &(A1) \; G_{0}\left(\mathcal{B}_{G_{0}(C_{k})}(i, \ell)\right) \text{ forms a tree}.\\
		&(A2) \; \operatorname{deg}_{G_{0}(C_{k})}(j) \leq 2m\frac{p}{1-\alpha} \text{ for any } j\in \mathcal{B}_{G_{0}(C_{k})}(i,\ell-1).\\
		&(A3) \; \operatorname{deg}_{G_{0}(C_{k})}(j) > (1-\epsilon)m\frac{p}{1-\alpha} \text{ for any } j\in \mathcal{B}_{G_{0}(C_{k})}(i,\ell-1). &&
        \end{flalign*} 
	\end{definition}
	
\begin{lemma}[Lemma 4.9 in \citep{MRT21a}]\label{lem:tree}
Consider an  Erd\H{o}s-R\'enyi graph $\Gamma\sim \mathcal{G}(n, r)$ for a positive integer $n$ and $r \in (0,1)$. For any positive integer $x$, the probability that there exist at least $(5 x)^{3}(\log n)^{6}(n r)^{3 x}$ vertices $i \in \Gamma$ satisfying $\Gamma\left(\mathcal{B}_\Gamma(i,x)\right)$ is not a tree, is upper bounded by $\exp \left(-\log ^{2} n\right)$.
	\end{lemma}

The following result shows that the majority of vertices in $G_0(C_k)$ are typical, satisfying Definition \ref{def:type}.
	\begin{proposition}\label{prop:type}
	 For any constants $D,\delta,\epsilon>0$, there exist $R>1$, $c\in(0,1/2)$ and $m_0$ depending on $D,\delta,\epsilon$ with the properties below. If
		\begin{equation}
		m\geq m_{0}, \quad (\log m)^{1+\delta} \leq m\frac{p}{1-\alpha} \leq m^{1/20} , \quad \ell \leq  \frac{\log m}{R \log mp} ,
		\end{equation}
        then
		$$
		\P\left\{\left|\operatorname{Typ}_{G_{0}(C_{k})}(\ell,\epsilon)\right| \geq m-m^{1-c}\right\} \geq 1-m^{-D}.
		$$
	\end{proposition}
	\begin{proof}
	Recall that $G_0(C_k)$ is a  $\mathcal{G}(m, p/(1-\alpha))$ Erd\H{o}s-R\'enyi graph. Thus, by applying Lemma \ref{lem:tree} for $x=\ell \leq  \frac{\log m}{R \log mp}$ with a sufficiently large $R$, we can obtain that there are at most 
	\begin{equation}\label{eq:tree number}
		(5\ell)^{3}(\log m)^6\left(\frac{mp}{1-\alpha}\right)^{3\ell} \leq \sqrt{m}
	\end{equation}
	vertices whose $\ell $-neighborhoods are not trees in $G_{0}(C_{k})$ with probability at least $1-\exp\left(-(\log m)^2 \right) \geq 1-m^{-D-1}$ for $m\geq m_0(D)$. Thus, the condition $(A1)$ holds. By Lemma \ref{lem:degree}, it can be easily checked that the two conditions $(A2)$ and $(A3)$ are satisfied for all $i\in C_k$ with probability $1-m^{-D-1}$ for $m\geq m_0(D,\delta,\epsilon)$.
\end{proof}

\section{Proof of Lemma \ref{lem:correct final}}\label{sec:proof of correct final}

From now on, for simpleness in the notation we will use $T^{d}_{\bss^{d}}(i)$ and $T'^{d}_{\bss^{d}}(i)$ instead of $T^{d}_{\bss^{d}}(i,G)$ and $T^{d}_{\bss^{d}}(i,G')$, respectively.
Consider a random subset $J$ uniformly drawn  from $\{-1,1\}^{k'\ell}$ with cardinality $2w$, where $w$ is a positive integer satisfying $w > 2(\log m)^2$.
We say that the event $\mathcal{G}$ holds if and only if there exists a subset $\tilde{J}(i) \subset J$ such that
		\begin{equation}\label{eq:jhat}
			|J \backslash \tilde{J}(i)|<(\log m)^{2}
		\end{equation}
		and
		\begin{align}
			&\label{eq:R_s in Jhat}\left|R_{\bss^\ell}(i)\right| \vee\left|R'_{\bss^\ell}(i)\right| \leq 96e w^{4} \frac{(m p)^{\ell}}{4^{k'\ell}} \quad \forall \bss^\ell \in \tilde{J}(i), \\
			&\label{eq:R_s in J}\left|R_{\bss^\ell}(i)\right| \vee\left|R'_{\bss^\ell}(i)\right| \leq 6\frac{(mp)^{\ell}}{2^{k'\ell}} \quad \forall \bss^\ell \in J,
		\end{align}
		where
		\begin{align}
			R_{\bss^\ell}(i) &\label{eq:R_s}:=T^{\ell}_{\bss^\ell}(i) \cap T'^{\ell}_{J \backslash\{\bss^\ell\}}(i) , \\
			R'_{\bss^\ell}(i) &\label{eq:R'_s}:=T'^{\ell}_{\bss^\ell}(i) \cap T^{\ell}_{J \backslash\{\bss^\ell\}}(i) .
		\end{align}
Let us define the five conditions $(B1)-(B5)$ as follows. 
\begin{flalign*}
    &(B1) \;  G_{0}\left(\mathcal{B}_{G_{0}(C_{k})}(i, \ell)\right) \text{ forms a tree}.\\
    &(B2) \; \left|\mathcal{B}_{G_{0}(C_k)}(i, \ell)\right| \leq m^{0.1}.\\
    &(B3) \; \left|T^{\ell}_{\bss^\ell}(i)\right| \vee\left|T'^{\ell}_{\bss^\ell}(i)\right| \leq 6(m p / 2^{k'})^{\ell} \text{ for all } \bss^\ell \in\{-1,1\}^{k'\ell}.\\
    &(B4) \; \text{Event } \mathcal{G} \text{ holds.}        \\
    &(B5) \;   \left|T^{\ell}_{\bss^\ell}(i) \cap T'^{\ell}_{\bss^\ell}(i)\right| \geq(m p / 2^{k'})^{\ell}(1-6\epsilon)^{k'\ell} \text{ for a constant } \epsilon \text{  and for all }\bss^\ell \in\{-1,1\}^{k'\ell}  .&&
\end{flalign*} 
We will show that there are enough number of vertices $i\in C_k$ satisfying the conditions $(B1)-(B5)$, and for all such vertices \eqref{eq:correct result} holds with high probability. 
The conditions $(B1)-(B3)$ will be shown to be held with high probability by using the analysis from Sec. \ref{sec:vertex classes}, and the conditions $(B4)-(B5)$ will be proved in Sec. \ref{sec:proof of B4} and Sec. \ref{sec:proof of B5}, respectively.

\subsection{Proof of Lemma \ref{lem:correct final}}
We consider a fixed vertex $i \in C_k$, and subsets $J$ and $\tilde{J}(i) \subset J$.  Additionally, we condition on the neighborhoods $G_0(\mathcal{B}_{G_0(C_k)}(i,\ell))$, $G(\mathcal{B}_{G_0(C_k)}(i,\ell))$, and $G'(\mathcal{B}_{G_0(C_k)}(i,\ell))$ such that conditions $(B1)-(B5)$ hold.
Let us define 
\begin{equation}\label{eq:tilde m}
    \tilde{m}:=m-\left|\mathcal{B}_{G_{0}(C_{k})}(i, \ell)\right|.
\end{equation}
Note that on the conditions $(B1)$ and $(B2)$, we can see that for any $j\in \mathcal{S}_{G_0(C_{k})}(i,\ell)$, $\operatorname{deg}_{G(C_{k})}(j)-1$  follows a binomial distribution with parameters $\tilde{m}$ defined in \eqref{eq:tilde m} with $\tilde{m} \geq m-m^{0.1}$ and $p$. Furthermore, it is independent for each $j\in \mathcal{S}_{G_0(C_{k})}(i,\ell)$. The same statement holds for $\operatorname{deg}_{G'(C_{k})}(j)-1$.

	

\begin{lemma}[Similar to Lemma 5.7 in \citep{MRT21a}]\label{lem:small overlap correct}
    For any constant $D > 0$, there exists a constant $K > 0$ which depends on $D$ with the properties below. For a fixed vertex $i\in C_k$, a subset $J$ and $\tilde{J}(i) \subset J$, define $R_{\bss^\ell}(i)$ and $R'_{\bss^\ell}(i)$ as in \eqref{eq:R_s} and \eqref{eq:R'_s}. Suppose that conditions $(B1)-(B5)$ hold. Then, we have
    \begin{align}
        &\label{eq:eta in J hat}\left|\eta_{\bss^\ell}(i)\right| \vee\left|\eta'_{\bss^\ell}(i)\right| \leq K \frac{(mp)^{(\ell +1) / 2}}{2^{k'\ell}} w^{2} \sqrt{\log m} \quad \forall \bss^\ell \in \tilde{J}(i) \\
        &\label{eq:eta in J}\left|\eta_{\bss^\ell}(i)\right| \vee\left|\eta'_{\bss^\ell}(i)\right| \leq K \frac{(m p)^{(\ell +1 ) / 2}}{2^{k'\ell / 2}} \sqrt{\log m} \quad \forall \bss^\ell \in J
    \end{align}
    where
    \begin{equation}\label{eq:eta}
        \eta_{\bss^\ell}(i):=\sum_{j \in R_{\bss^\ell}(i)}\left(\operatorname{deg}_{G(C_{k})}(j)-1-m p\right) \quad \text { and } \quad \eta'_{\bss^\ell}(i):=\sum_{j \in R'_{\bss^\ell}(i)}\left(\operatorname{deg}_{G'(C_{k})}(j)-1-m p\right)
    \end{equation}
    with probability at least $1-m^{-D}$.
\end{lemma}

\begin{proof}
Since $R_{\bss^\ell} (i) \subset \mathcal{S}_{G_0 (C_k)}(i,\ell)$, for any $j \in R_{\bss^\ell} (i)$, $(\operatorname{deg}_{G(C_k)}(j)-1)$ is distributed as  $\operatorname{Bin}(\tilde{m},p)$ with $\tilde{m}:=m-\left|\mathcal{B}_{G_{0}(C_{k})}(i, \ell)\right| \geq m-m^{0.1}$. In addition, it is independent for each $j \in R_{\bss^\ell} (i)$.  Therefore we have $\sum_{j \in R_{\bss^\ell}(i)}\left(\operatorname{deg}_{G(C_{k})}(j)-1\right) \sim \operatorname{Bin}\left(\tilde{m}\left|R_{\bss^\ell}(i)\right|, p\right)$. From this result, we can obtain
    \begin{equation}\label{eq:eta correct}
         \begin{aligned}
    \left|\eta_{\bss^\ell}(i)\right| &=\left|\sum_{j \in R_{\bss^\ell}(i)}\left(\operatorname{deg}_{G(C_{k})}(j)-1-m p\right)\right| \\
    &= \left|\sum_{j \in R_{\bss^\ell}(i)}\left(\operatorname{deg}_{G(C_{k})}(j)-1\right) -\tilde{m}p|R_{\bss^\ell} (i)| -(m-\tilde{m})p|R_{\bss^\ell} (i)|\right|\\
    & \leq K_{1}\left(\sqrt{m p\left|R_{\bss^\ell}(i)\right| \log m}+\log m\right)+(m-\tilde{m}) p\left|R_{\bss^\ell}(i)\right| \\
    & \leq K_{2} \sqrt{m p\left(\left|R_{\bss^\ell}(i)\right|+1\right) \log m}
    \end{aligned}
    \end{equation}
    with probability at least $1-m^{-D-1}$, where $K_1, K_2$ are sufficiently large constants depending on $D$. 
    In the same way, a bound for $\eta'_{\bss^\ell} (i)$ can be obtained. By combining \eqref{eq:eta correct} with $(B4)$, we obtain \eqref{eq:eta in J} and \eqref{eq:eta in J hat}. By applying a union bound over ${\bss^\ell} \in \{-1,1\}^{k'\ell}$, we can conclude the proof.
\end{proof}

For a fixed vertex $i$, we omit $(i)$ from $T_{\bss^\ell}(i),$ $R_{\bss^\ell}(i),$ $f_{\bss^\ell}(i),$ and $\mathrm{v}_{\bss^\ell}(i)$ for the sake of brevity in the notation.

	\begin{lemma}[Similar to Lemma 5.9 in \citep{MRT21a}]\label{lem:entrywise difference between signatures}
		Assume a realization of edges between $(R_{\bss^\ell} \cup R'_{\bss^\ell})$ and $\mathcal{S}_{G_{0}(C_{k})}(i, \ell+1)$ in the graphs $G_{0}, G$, and $G'$. Consider the corresponding conditional probability and expectation, denoted by $\hat{\mathbb{P}}$ and $\hat{\mathbb{E}}$, respectively. Then, for any $\bss^\ell \in J$, we have 
		$$
		f_{\bss^\ell}-f'_{\bss^\ell}=Z_{\bss^\ell}+\Delta_{\bss^\ell}.
		$$
		Here, $Z_{\bss^\ell}$ is a random variable and $\Delta_{\bss^\ell}$ is a deterministic value that satisfy the following:
		\begin{itemize}
			\item $\hat{\mathbb{E}}\left[Z_{\bss^\ell}\right]=0$;
			\item $\hat{\mathbb{E}}\left[Z_{\bss^\ell}^{2}\right] \leq \mathrm{v}_{\bss^\ell}+\mathrm{v}'_{\bss^\ell}-2 \tilde{m} p(1-p-\alpha)\left|T^{\ell}_{\bss^\ell}(i)\cap T'^{\ell}_{\bss^\ell}(i) \right|$;
			\item $\hat{\mathbb{P}}\left\{\left|Z_{\bss^\ell}\right| \geq t\right\} \leq 2 \exp \left(\frac{-t^{2} / 2}{\hat{E}\left[Z_{\bss^\ell}^{2}\right]+t / 3}\right)$;
			\item $\left|\Delta_{\bss^\ell}\right| \leq\left|\eta_{\bss^\ell}\right|+\left|\eta'_{\bss^\ell}\right|+2 m^{0.2} p$.
		\end{itemize}
		Furthermore, the random variables $Z_{\bss^\ell}$ are conditionally independent for distinct $\bss^\ell \in J$.
	\end{lemma}

\begin{proof}
        Lemma \ref{lem:entrywise difference between signatures} can be proved in a similar way as that of Lemma 5.9 in \citep{MRT21a}. 
        The main observation behind the proof is that $f_{\bss^\ell}-f'_{\bss^\ell}=Z_{\bss^\ell}+\Delta_{\bss^\ell}$, where
$$
\begin{aligned}
Z_{\bss^\ell}:= & \sum_{j \in T^{\ell}_{\bss^\ell} \cap T'^\ell_{\bss^\ell}}\left(\operatorname{deg}_{G(C_k)}(j)-\operatorname{deg}_{G'(C_k)}(j)\right) \\
& +\sum_{j \in T^\ell_{\bss^\ell} \backslash T'^\ell_{J}}\left(\operatorname{deg}_{G(C_k)}(j)-1-\tilde{m} p\right) \\
& -\sum_{j \in T'^\ell_{\bss^\ell} \backslash T^\ell_{J}}\left(\operatorname{deg}_{G'(C_k)}(j)-1-\tilde{m} p\right)
\end{aligned}
$$
and
$$
\Delta_{\bss^\ell}:=\eta_{\bss^\ell}-\eta'_{\bss^\ell}+(\tilde{m}-m) p\left(\left|T^\ell_{\bss^\ell}\right|-\left|T'^\ell_{\bss^\ell}\right|\right).
$$
\end{proof}        

	\begin{lemma}[Upper bound on the normalized distance of sparsified signature vectors for a correct pair]\label{lem:correct}
        For any constants $C, D,K,\delta>0$, there exist constant $\epsilon$ and $m_0$ with the properties below. Suppose that $m\geq m_{0}, \; \alpha \in(0, \epsilon),\; w \geq(\log m)^{4},$ and
    \beq
    \begin{split}\nonumber
    (\log m)^{1+\delta} \leq m \frac{p}{1-\alpha} \leq m^{1/20}\\
  2\log \left(w^{4}\log m\right) \leq k'\ell \leq C \log \log m.
    \end{split}
    \eeq
		Moreover, suppose that a vertex $i\in C_k$ satisfies conditions $(B1) - (B5)$ with a constant $\epsilon>0$. Consider a subset $J \subset \{-1.1\}^{k'\ell}$ satisfying $|J|=2 w$. 
		Under the same conditions as stated in Lemma \ref{lem:entrywise difference between signatures}, where $\eta_{\bss^\ell}$ and $\eta'_{\bss^\ell}$ satisfy \eqref{eq:eta in J hat} and \eqref{eq:eta in J} with a constant $K>0$, we have	
		$$
		\sum_{\bss^\ell \in J} \frac{\left(f_{\bss^\ell}(i)-f'_{\bss^\ell}(i)\right)^{2}}{\mathrm{v}_{\bss^\ell}(i)+\mathrm{v}'_{\bss^\ell}(i)} \leq 2 w\left(1-\frac{1}{(\log m)^{0.1}}\right),
		$$
  with a conditional probability of at least $1 - m^{-D}$.
	\end{lemma}

 Lemma \ref{lem:correct} will be proved in Section \ref{sec:proof of lemma correct}.
\begin{proof}[Proof of Lemma \ref{lem:correct final}]
   First, we will show that conditions $(B1)-(B5)$ hold for at least $m-m^{1-c}$ vertices $i \in C_{k}$ with probability at least $1-m^{-D-1}$.
    \begin{itemize}
    \item By Proposition \ref{prop:type}, there exist at least $m-m^{1-c_1}$ vertices $i\in C_k$ for a constant $c_1 \in (0,0.5) $ that satisfy condition $(B1)$ with probability at least $1-m^{-D-2}$. 
    
    \item Based on Lemma \ref{lem:Sizes of neighborhoods and their intersections}, we have $|\mathcal{B}_{G_{0}(C_{k})}(i,\ell)| =O\left( (mp/(1-\alpha))^{\ell}\right) \leq m^{0.1}$,  $\forall i\in C_k$ and  $\ell \leq \frac{\log m}{11 \log mp}$, with probability at least $1-m^{-D-2}$. Thus, the condition $(B2)$ holds for all $i\in C_k$ with probability at least $1-m^{-D-2}$.
    
    \item On the conditions $(B1)$ and $(B2)$, the condition $(B3)$ holds for all $i\in C_k$ with probability at least $1-m^{-D-2}$ by Lemma \ref{lem:sizes of vertex class}. 
    
    \item Condition $(B4)$ holds for all $i\in C_k$ satisfying conditions $(B1)$ and $(B3)$  with probability at least $1-m^{-D-2}$ by Lemma \ref{lem:sparsification of correct}.
    
    \item By Proposition \ref{prop:type}, $|\operatorname{Typ}_{G_0(C_k)}(\ell,\epsilon)|\geq m-m^{1-c_2}$ holds for a constant $c_2 \in (0,0.5)$ with probability $1-m^{-D-2}$. Thus, condition $(B5)$ holds for at least $m-m^{1-c_2}$ vertices $i \in C_k$ with probability at least $1-2m^{-D-2}$ by Lemma \ref{prop:common Ts}.
    \end{itemize}
Moreover, by applying Lemma \ref{lem:small overlap correct} to the vertex $i$ satisfying conditions $(B1)-(B5)$, we obtain $\eta_{\bss^\ell}$ and $\eta'_{\bss^\ell}$ that satisfy \eqref{eq:eta in J hat} and \eqref{eq:eta in J} with a probability of at least $1 - m^{-D-2}$. Therefore, by applying Lemma \ref{lem:correct}, the proof is complete.

\end{proof}

    \subsection{Proof of Lemma \ref{lem:correct}}\label{sec:proof of lemma correct}
  	By condition $(B5)$, we have that
	\begin{equation}
		|T^{\ell}_{\bss^\ell}(i) \cap T'^{\ell}_{\bss^\ell}(i)| \geq (mp/2^{k'})^{\ell}(1-6\epsilon)^{k'\ell}
	\end{equation}
        for any $\bss^\ell \in\{-1,1\}^{k'\ell}$. 	Thus, we obtain
	\begin{equation}\label{eq:vs1}
	    \begin{aligned}
		\mathrm{v}_{\bss^\ell} &=m p(1-p)\left|T^{\ell}_{\bss^\ell}(i) \right| \geq m p(1-p)|T^{\ell}_{\bss^\ell}(i) \cap T'^{\ell}_{\bss^\ell}(i)|  \\
		& \geq (1-p)\left(\frac{1-6\epsilon}{2}\right)^{k'\ell}(mp)^{\ell+1}.
	\end{aligned}
	\end{equation}
	   Moreover, by condition $(B3)$,
	\begin{equation}\label{eq:vs2}
	    \mathrm{v}_{\bss^\ell}=m p(1-p)\left|T^{\ell}_{\bss^\ell}(i)\right| \leq 6(1-p)\frac{(mp)^{\ell+1}}{2^{k'\ell}} .
	\end{equation}
	 In the same way, the bounds for $\mathrm{v}'_{\bss^\ell}$ can be obtained as follows:
        \begin{equation}
          (1-p)\left(\frac{1-6\epsilon}{2}\right)^{k'\ell}(mp)^{\ell+1}  \leq \mathrm{v}'_{\bss^\ell} \leq 6(1-p)\frac{(mp)^{\ell+1}}{2^{k'\ell}}.
        \end{equation}
  One can see from Lemma \ref{lem:entrywise difference between signatures} that for a correct pair of vertices the difference  between the entries of signatures $f_{\bss^\ell}$ and $f'_{\bss^\ell}$ at some $\bss^\ell\in\{-1,1\}^{k'\ell}$ can be decomposed into the  random variable part $Z_{\bss^\ell}$ and the deterministic part $\Delta_{\bss^\ell}$.  Define
	$$
	X_{\bss^\ell}:=\frac{f_{\bss^\ell}-f'_{\bss^\ell}}{\sqrt{\mathrm{v}_{\bss^\ell}+\mathrm{v}'_{\bss^\ell}}}=\frac{Z_{\bss^\ell}+\Delta_{\bss^\ell}}{\sqrt{\mathrm{v}_{\bss^\ell}+\mathrm{v}'_{\bss^\ell}}},
	$$
	where $Z_{\bss^\ell}$ and $\Delta_{\bss^\ell}$ are defined in Lemma \ref{lem:entrywise difference between signatures}.	
 
We first estimate the expectation and variance of $X_{\bss^\ell}$. Recall $\hat{\mathbb{E}}$ defined in Lemma \ref{lem:entrywise difference between signatures}. By Lemma \ref{lem:entrywise difference between signatures}, \eqref{eq:eta in J hat},\eqref{eq:eta in J},\eqref{eq:vs1} and \eqref{eq:vs2}, we can have
	\begin{equation}\label{eq:Xs correct Jhat}
	    \begin{aligned}
		\left|\hat{\mathbb{E}}\left[X_{\bss^\ell}\right]\right|=\frac{\left|\Delta_{\bss^\ell}\right|}{\sqrt{\mathrm{v}_{\bss^\ell}+\mathrm{v}'_{\bss^\ell}}} & \leq \frac{\left|\eta_{\bss^\ell}\right|+\left|\eta'_{\bss^\ell}\right|+2 m^{0.2} p}{\sqrt{2(1-p)(mp)^{\ell+1}((1-6\epsilon)/2)^{k'\ell}}} \\
		& \leq \frac{2 K \frac{(m p)^{(\ell+1) / 2}}{2^{k'\ell}} w^{2} \sqrt{\log m}+2 m^{0.2} p}{\sqrt{2(1-p)(mp)^{\ell+1}((1-6\epsilon)/2)^{k'\ell}}} \leq \frac{5 K w^{2} \sqrt{\log m}}{2^{k'\ell / 2}(1-6\epsilon)^{k'\ell / 2}} \quad \text { for } \bss^\ell \in \tilde{J}(i), 
		\end{aligned}
  \end{equation}

	\begin{equation}\label{eq:Xs correct J}
		\left|\hat{\mathbb{E}}\left[X_{\bss^\ell}\right]\right|=\frac{\left|\Delta_{\bss^\ell}\right|}{\sqrt{\mathrm{v}_{\bss^\ell}+\mathrm{v}'_{\bss^\ell}}} \leq \frac{2 K \frac{(m p)^{(\ell+1) / 2}}{2^{k'\ell/2}} \sqrt{\log m}+2 m^{0.2} p}{\sqrt{2(1-p)(mp)^{\ell+1}((1-6\epsilon)/2)^{k'\ell}}} \leq \frac{5 K  \sqrt{\log m}}{(1-6\epsilon)^{k'\ell / 2}} \quad \text { for } \bss^\ell \in J,
	\end{equation}
	
	and
	\begin{equation}\label{eq:Var Xs}
	    \begin{aligned}
		\operatorname{Var}\left(X_{\bss^\ell}\right)=\frac{\hat{\mathbb{E}}\left[Z_{\bss^\ell}^{2}\right]}{\mathrm{v}_{\bss^\ell}+\mathrm{v}'_{\bss^\ell}} & \leq \frac{\mathrm{v}_{\bss^\ell}+\mathrm{v}'_{\bss^\ell}-2 \tilde{m} p(1-p-\alpha)\left|T^{\ell}_{\bss^\ell} \cap T'^{\ell}_{\bss^\ell} \right|}{\mathrm{v}_{\bss^\ell}+\mathrm{v}'_{\bss^\ell}} \\
		& \leq 1-\frac{2\tilde{m} p(1-p-\alpha)(m p / 2^{k'})^{\ell}(1-6\epsilon)^{k'\ell}}{12(1-p)(m p)^{\ell+1} / 2^{k'\ell}} \leq 1-\frac{(1- 6\epsilon)^{k'\ell}}{13},
	\end{aligned}
	\end{equation}
	if $\alpha \leq \frac{1}{10}$.
	By \eqref{eq:Xs correct Jhat} and \eqref{eq:Var Xs}, we can get
 \begin{equation}\label{eq:Xs2 correct Jhat}
     \hat{\mathbb{E}}\left[X_{\bss^\ell}^{2}\right]=\frac{\hat{\mathbb{E}}\left[Z_{\bss^\ell}^{2}\right]+\Delta_{\bss^\ell}^{2}}{\mathrm{v}_{\bss^\ell}+\mathrm{v}'_{\bss^\ell}} \leq 1-\frac{(1- 6\epsilon)^{k'\ell}}{13}+\frac{25 K^{2} w^{4} \log m}{2^{k'\ell}(1-6\epsilon)^{k'\ell }} \leq 1-\frac{(1-6\epsilon)^{k'\ell}}{14} \quad \text { for } \bss^\ell \in \tilde{J}(i)
 \end{equation}
	 if $2\log \left(w^{4}\log m\right) \leq k'\ell$, $\epsilon$ is small enough and $m\geq m_0(K,\epsilon)$. Similarly, we have
	\begin{equation}\label{eq:Xs2 correct J}
	    \hat{\mathbb{E}}\left[X_{\bss^\ell}^{2}\right]=\frac{\hat{\mathbb{E}}\left[Z_{\bss^\ell}^{2}\right]+\Delta_{\bss^\ell}^{2}}{\mathrm{v}_{\bss^\ell}+\mathrm{v}'_{\bss^\ell}} \leq 1-\frac{(1- 6\epsilon)^{k'\ell}}{14}+\frac{25 K^{2}\log m}{(1-6\epsilon)^{k'\ell }} \leq \frac{26 K^{2}\log m}{(1-6\epsilon)^{k'\ell }} \quad \text { for } \bss^\ell \in J ,
	\end{equation}
        if $k'\ell \leq C \log \log m$, $\epsilon$ is a small enough value depending on $C$ and $m\geq m_0(K,\epsilon)$.  Moreover, we can obtain   $\hat{\mathbb{P}}\left\{\left|X_{\bss^\ell}-\hat{\mathbb{E}}\left[X_{\bss^\ell}\right]\right| \geq t\right\} \leq 2 \exp \left(\frac{-t^{2} / 2}{1+t / 3}\right)$ since $\hat{\mathbb{P}}\left\{\left|Z_{\bss^\ell}\right| \geq t\right\} \leq 2 \exp \left(\frac{-t^{2} / 2}{\hat{E}\left[Z_{\bss^\ell}^{2}\right]+t / 3}\right)$ holds by Lemma \ref{lem:entrywise difference between signatures} and  $\operatorname{Var}\left(X_{\bss^\ell}\right) \leq 1$ from \eqref{eq:Var Xs}. Therefore, applying Lemma \ref{lem:hoeffding's with truncation}, we obtain 
	\begin{equation}\label{eq:Xs bound}
	    \sum_{\bss^\ell \in J} X_{\bss^\ell}^{2} \leq \sum_{\bss^\ell \in J} \hat{\mathbb{E}}\left[X_{\bss^\ell}^{2}\right]+C_{1} \sqrt{w}(\log m)^{3 / 2}+C_{1}\left(\max _{\bss^\ell \in J}\left|\hat{\mathbb{E}}\left[X_{\bss^\ell}\right]\right|\right)(\sqrt{w \log m}+\log m),
	\end{equation}
     with conditional probability at least $1-m^{-D}$, for $J \subset \{-1,1\}^{k'\ell}$ such that $|J|=2w$ and a sufficiently large constant $C_{1}$ depending on $D$. By \eqref{eq:Xs2 correct J}, \eqref{eq:Xs2 correct Jhat} and \eqref{eq:jhat} in condition $(B4)$, we can get 
	\begin{equation}\label{eq:Xs bound-1}
	    \begin{aligned}
		\sum_{\bss^\ell \in J} \hat{\mathbb{E}}\left[X_{\bss^\ell}^{2}\right] &=\sum_{\bss^\ell \in \tilde{J}(i)} \hat{\mathbb{E}}\left[X_{\bss^\ell}^{2}\right]+\sum_{\bss^\ell \in J \backslash \tilde{J}(i)} \hat{\mathbb{E}}\left[X_{\bss^\ell}^{2}\right] \\
		& \leq 2 w\left(1-\frac{(1-6\epsilon)^{k'\ell}}{14}\right)+(\log m)^{2} \frac{26 K^{2}\log m}{(1-6\epsilon)^{k'\ell }} \leq 2 w\left(1-\frac{(1-6\epsilon)^{k'\ell}}{15}\right)
	\end{aligned}
	\end{equation}
	if $k'\ell \leq C \log \log m$, $\epsilon>0$ is small enough which depends on $C$, $w \geq(\log m)^{4}$ and $m \geq m_0(K,\epsilon)$. Finally, by \eqref{eq:Xs correct J}, \eqref{eq:Xs bound} and \eqref{eq:Xs bound-1}, we obtain
	$$
	\begin{aligned}
		\sum_{\bss^\ell \in J} X_{\bss^\ell}^{2} & \leq2 w\left(1-\frac{(1-6\epsilon)^{k'\ell}}{15}\right)+C_{1} \sqrt{w}(\log m)^{3 / 2}+C_{1} \frac{5 K  \sqrt{\log m}}{(1-6\epsilon)^{k'\ell / 2}}(\sqrt{w \log m}+\log m) \\
		& \leq 2 w\left(1-\frac{(1-6\epsilon)^{k'\ell}}{16}\right) \leq 2 w\left(1-\frac{1}{(\log m)^{0.1}}\right)
	\end{aligned}
	$$
	if $k'\ell \leq C \log \log m$, $\epsilon$ is a sufficiently small constant depending on $C$, $w \geq(\log m)^{4}$ and  $m\geq m_0(K,\epsilon)$. 

	\subsection{Proof of Condition $(B4)$}\label{sec:proof of B4}
    In this subsection, we will show that event $\mathcal{G}$ (or condition $(B4)$) holds with high probability. The idea of comparing sparsified signature vectors was addressed and used in \cite{MRT21a, MRT21b} to mitigate the dependency issue between the signature vectors. We will follow the same trick here.

	\begin{lemma}[Similar to Lemma 5.6 in \citep{MRT21a}]\label{lem:sparsification of correct}
		Consider a fixed $i\in C_k$ and a random subset $J$ uniformly drawn from $\{-1,1\}^{k'\ell}$ with cardinality $2w$, for a positive integer  $w$ satisfying $w > 2(\log m)^2$. If a vertex $i$ satisfies that
\begin{flalign*} 
&(C1) \; G_{0}\left(\mathcal{B}_{G_{0}(C_{k})}(i, \ell)\right) \text{ forms a tree};\\
&(C2) \;\left|T^{\ell}_{\bss^\ell}(i)\right| \vee \left|T'^{\ell}_{\bss^\ell}(i)\right| \leq 6\left(\frac{m p}{2^{k'}}\right)^{\ell} \text{ for all } \bss^\ell \in\{-1,1\}^{k'\ell},&&
\end{flalign*} 
		then event $\mathcal{G}$ (or condition $(B4)$) holds with probability at least $1-\exp \left(-(\log m)^{1.5}\right)$.
	\end{lemma}

\begin{proof}   
By the condition $(C2)$, we have
	\begin{equation}
		\left|T^{\ell}_{\bss^\ell}(i) \right| \vee \left|T'^{\ell}_{\bss^\ell}(i)  \right| \leq 6\frac{(mp)^{\ell}}{2^{k'\ell}}
	\end{equation}
    for all $\bss^\ell \in \{-1,1\}^{k'\ell}$.
    Thus, it is easy to see that \eqref{eq:R_s in J} is established since $R_{\bss^\ell}(i) \subset T^\ell _{\bss^\ell} (i)$ and $R'_{\bss^\ell}(i) \subset T'^\ell _{\bss^\ell} (i)$.

	Next, we apply Lemma \ref{lem:sparsification lem} with $\Omega=\mathcal{S}_{G(C_{k})}(i, \ell), \Omega'=\mathcal{S}_{G'(C_{k})}(i, \ell), k=2^{k'\ell}$, $S=6(m p)^{\ell}, L=8 e w^{3}$, and $\rho=\frac{1}{4 w}(\log m)^{2}$. Moreover, we use $T^{\ell}_{\bss^\ell}(i)$ and $T'^{\ell}_{\bss^\ell}(i)$ for the partition set of $\Omega=\mathcal{S}_{G(C_{k})}(i, \ell)$ and $ \Omega'=\mathcal{S}_{G'(C_{k})}(i, \ell)$. Then, we can get 
	\begin{equation}\label{eq:sparsification}
		\left| \left\{\bss^\ell \in J: \exists \bst^\ell \in J \backslash\{\bss^\ell\} \text { s.t. 	}\left|T^{\ell}_{\bss^\ell}(i) \cap T'^{\ell}_{\bst^\ell}(i) \right| \geq 48 e w^{3}  \frac{(m p)^{\ell}}{4^{k'\ell}}\right\} \right| 	<\frac{1}{2}(\log m)^{2},
	\end{equation}
	with probability at least $1-\exp \left(-(\log m)^{2} / 4\right)$. 
 In a similar way, the following results can be obtained.
 \begin{equation}\label{eq:sparsification swap}
		\left| \left\{\bss^\ell \in J: \exists \bst^\ell \in J \backslash\{\bss^\ell\} \text { s.t. 	}\left|T'^{\ell}_{\bss^\ell}(i) \cap T^{\ell}_{\bst^\ell}(i) \right| \geq 48 e w^{3}  \frac{(m p)^{\ell}}{4^{k'\ell}}\right\} \right| 	<\frac{1}{2}(\log m)^{2},
	\end{equation}
	with probability at least $1-\exp \left(-(\log m)^{2} / 4\right)$. Define
	$$
	\tilde{J}(i):=\left\{\bss^\ell \in J:\left|R_{\bss^\ell}(i)\right| \vee\left|R'_{\bss^\ell}(i)\right| \leq 96 e w^{4}  \frac{(m p)^{\ell}}{4^{k'\ell}}\right\}.
	$$
We can see that $\tilde{J}(i)$ is a superset of
	$$
	\begin{aligned}
		&\left\{\bss^\ell \in J: \forall \bst^\ell  \in J \backslash\{\bss^\ell\},\;\left|T^{\ell}_{\bss^\ell}(i) \cap T'^{\ell}_{\bst^\ell}(i) \right| \leq 48 e w^{3}  \frac{(m p)^{\ell}}{4^{k'\ell}}\right\} \\
		&\bigcap\left\{\bss^\ell \in J: \forall \bst^\ell \in J \backslash\{\bss^\ell\},\;\left|T'^{\ell}_{\bss^\ell}(i) \cap T^{\ell}_{\bst^\ell}(i) \right| \leq 48 e w^{3}  \frac{(m p)^{\ell}}{4^{k'\ell}}\right\} .
	\end{aligned}
	$$
By combining the above result with \eqref{eq:sparsification} and \eqref{eq:sparsification swap}, we see that \eqref{eq:jhat} and \eqref{eq:R_s in Jhat} hold. By taking a union bound over all $i \in C_{k}$, we can complete the proof.
\end{proof}


	\subsection{Proof of Condition $(B5)$}\label{sec:proof of B5}
	
	In this subsection, we will show that  `typical' vertices, defined Definition \ref{def:type}, satisfy condition $(B5)$ with high probability, i.e., $\left|T^{\ell}_{\bss^\ell}(i) \cap T'^{\ell}_{\bss^\ell}(i)\right| \geq(m p / 2^{k'})^{\ell}(1-6\epsilon)^{k'\ell}$ for a constant $\epsilon$  and for all $\bss^\ell \in\{-1,1\}^{k'\ell} $ with high probability.
	\begin{lemma}[Degree correlation between correct pairs of vertices]\label{lem:correct signature}
         For any constant $\epsilon>0$, there exist constants $\alpha_{0},L>0$ depending only on $\epsilon$ with the properties below. Consider the two graphs $G^\pi$ and $G'$, which are generated from the correlated SBMs defined in Sec. \ref{sec:model}, with correlation $1-\alpha$. Suppose that community labels $(C_1,C_2,\dots, C_k)$ are given in both graphs and $n_{a}q \geq L$. If $\alpha \in (0,\alpha_{0})$, then for all $i\in C_{k}$ and any $a\in[k-1]$,
		\begin{equation}
		  \P\left\{ \sign(\operatorname{deg}^{a}_{G}(i)-n_{a}q) =  \sign(\operatorname{deg}^{a}_{G'}(i)-n_{a}q)  \right\} \geq 1-\epsilon.
		\end{equation}
	\end{lemma}
	\begin{proof} From Lemma \ref{lem:binomial 1}, we can show that there exists a universal constant $C>0$ such that for any $r>0$,
	\begin{equation}\label{eq:binomial}
		\P\left\{ |\operatorname{deg}^{a}_{G}(i)-n_{a}q| > r\sqrt{n_{a}q(1-q)}  \right\} \geq 1-Cr.
	\end{equation}
	By applying Lemma \ref{lem:correlation of neighbor} with $J=J'=C_{a}$,  for any $t>0$, we have
	\begin{equation}\label{eq:correlation of same node}
		\P\left\{\left| \operatorname{deg}^{a}_{G}(i)-\operatorname{deg}^{a}_{G'}(i) \right| \geq 4\left(t+ \sqrt{t\alpha q n_{a}}\right) \right\} \leq 6 \exp(-t) + 2 \exp \left(-\frac{qn_{a}}{3(1-\alpha)}\right).
	\end{equation}
	Let $t= \frac{r\sqrt{n_{a}q(1-q)}}{8} \wedge \frac{r^{2}(1-q)}{64\alpha} $. Then, we can get
	\begin{equation}
		4\left(t+ \sqrt{t\alpha q n_{a}}\right) \leq r\sqrt{n_{a}q(1-q)}.
	\end{equation}
	From \eqref{eq:binomial} and \eqref{eq:correlation of same node},  we have
	\begin{equation}
		\begin{aligned}
			&\P\left\{\sign(\operatorname{deg}^{a}_{G}(i)-n_{a}q) = \sign(\operatorname{deg}^{a}_{G'}(i)-n_{a}q)  \right\} \\
			& \geq \P\left\{ |\operatorname{deg}^{a}_{G}(i)-n_{a}q| > r\sqrt{n_{a}q(1-q)} \text{ and }   \left| \operatorname{deg}^{a}_{G}(i)-\operatorname{deg}^{a}_{G'}(i) \right| < r\sqrt{n_{a}q(1-q)}  \right\} \\
			&\geq 1-Cr-6\exp(-t)- 2 \exp \left(-\frac{qn_{a}}{3(1-\alpha)}\right).
		\end{aligned}
	\end{equation}
        We can take a small enough $r$ to make $Cr<\epsilon/3$. We can also make $ 2 \exp \left(-\frac{qn_{a}}{3(1-\alpha)}\right)<\epsilon/3$ and $6\exp(-t) < \epsilon/3$ by setting $t= \frac{r\sqrt{n_{a}q(1-q)}}{8} \wedge \frac{r^{2}(1-q)}{64\alpha} $, since $n_{a}q\geq L$ for a sufficiently large constant $L$ depending on $\epsilon$ and we can take a small enough $\alpha_0$. Thus, the proof is complete.
	\end{proof}

 We will prove condition $(B5)$ using the above results. Similar to Proposition 5.1 in \citep{MRT21a}, we prove that there exists a significant overlap between the partitioning nodes of a correct pair.
	
\begin{proposition}[Overlap between the partitioning nodes of a correct pair]\label{prop:common Ts}
    For any constants $D,\delta,\epsilon >0$, there exist constants $m_0,Q,k_1$ and $\alpha_1$, which depend only on $D,\delta$ and $\epsilon$, with the properties below. 
    Consider the two graphs $G$ and $G'$, which are generated from the correlated SBMs defined in Sec. \ref{sec:model}, 
with correlation $1-\alpha$. 
Suppose that community labels $(C_1,C_2,\dots, C_k)$ are given. Assume a given instance of $G_{0}$.  For a fixed positive integer $\ell$, assume that 
    $$
    m\geq m_0,\quad  (\log m)^{1+\delta} \leq m p/(1-\alpha) \leq m^{1/20} , \quad  k' \leq k_1 \log mp, \quad mq\geq Q k'^{2},\quad \alpha \in (0,\alpha_1).
    $$
    Then, for every vertex $i \in \operatorname{Typ}_{G_{0}(C_{k})}(\ell,\epsilon)$, satisfying Definition \ref{def:type}, and for any $\bss^{\ell} \in\{-1,1\}^{k'\ell}$,
    $$
    \left|T_{\bss^{\ell}}^{\ell}(i, G) \cap T_{\bss^{\ell}}^{\ell}\left(i, G'\right)\right| \geq \left(\frac{mp}{2^{k'}}\right)^{\ell}(1-6\epsilon)^{k'\ell},
    $$
    with probability at least $1-m^{-D}$.
\end{proposition}	
\begin{proof} 
Let $\P_{0}$ represent the conditional probability given an instance of $G_{0}$. Fix a vertex $i \in \operatorname{Typ}_{G_{0}(C_{k})}(\ell,\epsilon)$, and let $d\in \{0,1,\ldots,\ell-1\}$.
For any $\epsilon >0$ and $j \in \mathcal{S}_{G_{0}(C_{k})}(i,d)$ ,  we have
\begin{equation}
    \begin{aligned}
        &\P_{0}\{\operatorname{deg}_{G\cap G'(C_{k})}(j) \leq (1-3\epsilon)mp\} \\
        &\leq \P_{0}\{(1-\alpha)^{2}\operatorname{deg}_{G_{0}(C_{k})}(j)- \operatorname{deg}_{G\cap G'(C_{k})}(j) > (1-\alpha)^{2}\operatorname{deg}_{G_{0}(C_{k})}(j) - (1-3\epsilon)mp\}\\
        &\stackrel{(a)}{\leq} \P_{0}\{(1-\alpha)^{2}\operatorname{deg}_{G_{0}(C_{k})}(j)- \operatorname{deg}_{G\cap G'(C_{k})}(j) > \epsilon mp\} \stackrel{(b)}{\leq} \exp\left(-K mp \right) \leq m^{-D-2},
    \end{aligned}
\end{equation}
where $(a)$ holds from condition $(A3)$ and by choosing $\alpha<\epsilon$. We have that $\operatorname{deg}_{G_{0}(C_{k})}(j) \leq 2m\frac{p}{1-\alpha}$ by condition $(A2)$ and $\operatorname{deg}_{G\cap G'(C_{k})}(j) $ is $\operatorname{Binomial}\left(\operatorname{deg}_{G_{0}(C_{k})}(j) ,(1-\alpha)^{2}\right)$. Therefore, using Bernstein's inequality (Lemma \ref{lem:bernstein's inequality}), it can be confirmed that the inequality $(b)$ holds for a constant $K$ depending only on $\epsilon$. The last inequality holds since $mp/(1-\alpha) \geq (\log m)^{1+\delta}$, $\alpha<\epsilon$ and $m\geq m_0(D,\delta,\epsilon)$.

Define an event
\begin{equation}
    \mathcal{F}(i,d):=\left\{\operatorname{deg}_{G\cap G'(C_{k})}(j) \geq (1-3\epsilon)mp, \; \; \forall j\in \mathcal{S}_{G_{0}(C_{k})}(i,d)\right\}.
\end{equation}
By applying a union bound over $j\in \mathcal{S}_{G_{0}(C_{k})}(i,d)$ and $d\in \{0,\ldots,\ell-1\}$, we can obtain
\begin{equation}
    \P_{0}\left\{\cap^{\ell-1}_{d=0} \mathcal{F}(i,d)\right\} \geq 1-m^{-D-1}.
\end{equation}

    On the event $\cap^{\ell-1}_{d=0} \mathcal{F}(i,d)$, we will prove that 
    \begin{equation}\label{eq:Ts induction}
             \left|T_{\bss^{d}}^{d}(i, G) \cap T_{\bss^{d}}^{d}\left(i, G'\right)\right| \geq \left(\frac{mp}{2^{k'}}\right)^{d}(1-6\epsilon)^{k'd}
    \end{equation}
for all $d\in \{0,1,\ldots,\ell \}$ and $\bss^d \in \{-1,1\}^{k'd}$. The case where $d=0$ holds trivially. Assume that \eqref{eq:Ts induction} holds for $d \in \{0,1,\ldots,\ell -1\}$ and $\bss^d  \in \{-1,1\}^{k'd} $.    For any $\bss^{d} \in \{-1,1\}^{k'd} \text{ and } \bss_{d+1} \in \{-1,1\}^{k'}$, let 	$\bss^{d+1}=(\bss^{d},\bss_{d+1}) \in \{-1,1\}^{k'(d+1)}$.
On the event $\cap^{\ell-1}_{d=0} \mathcal{F}(i,d)$, we also have that
\begin{equation}
\begin{aligned}
      \left|\mathcal{N}_{G\cap G'(C_{k})}\left(T_{\bss^{d}}^{d}(i, G) \cap T_{\bss^{d}}^{d}\left(i, G'\right)\right)  \cap \mathcal{S}_{G_0(C_k)}(i,d+1) \right| &\geq (1-4\epsilon)mp \left| \left(T_{\bss^{d}}^{d}(i, G) \cap T_{\bss^{d}}^{d}\left(i, G'\right)\right) \right|
\end{aligned}
\end{equation}
because of the condition $(A1)$. For any $j \in \mathcal{N}_{G\cap G'(C_{k})}\left(T_{\bss^{d}}^{d}(i, G)  \cap T_{\bss^{d}}^{d}\left(i, G'\right)\right) \cap \mathcal{S}_{G_0(C_k)}(i,d+1)$, we have
\begin{equation}
    \P_{0}\left\{j \in T_{\bss^{d+1}}^{d+1}(i, G) \cap T_{\bss^{d+1}}^{d+1}\left(i, G'\right) \right\} \geq \frac{(1-\epsilon)^{k'}}{2^{k'}}\left(1-\frac{2Ck'}{\sqrt{mq}}\right) \geq  \frac{(1-\epsilon)^{k'}}{2^{k'}}\cdot \left(1-\epsilon\right)
\end{equation}
for a universal constant $C$ and $\alpha \in (0,\alpha_{0})$ because of Lemma \ref{lem:correct signature} and \eqref{eq:signature} with $(1-x)^{n}\geq 1-nx$. The last inequality holds from $mq\geq Qk'^{2}$ for a sufficiently large constant $Q$ depending on $\epsilon$. Therefore, by Hoeffding's inequality (Lemma \ref{lem:Hoeffding}), we obtain
\begin{equation}
\begin{aligned}
   &\P_{0}\left\{ \left| T_{\bss^{d+1}}^{d+1}(i, G) \cap T_{\bss^{d+1}}^{d+1 } (i,G')\right | \geq \left(\frac{1-6\epsilon}{2}\right)^{k'}mp \left| T_{\bss^{d}}^{d}(i, G) \cap T_{\bss^{d}}^{d}\left(i, G'\right) \right|  \right\} \\
   & \geq 1-\exp \left(-K' \left(\frac{1-6\epsilon}{2}\right)^{2k'} mp \left| T_{\bss^{d}}^{d}(i, G) \cap T_{\bss^{d}}^{d}\left(i, G'\right) \right| \right)    \\
   & \geq 1-m^{-D-1}
\end{aligned}
  \end{equation}
for a positive constant $K'$ which depend on $\epsilon$. The last inequality holds since $k' \leq k_1 \log mp$ for a constant $k_1$ depending on $\epsilon,\delta$ such that $\left(\frac{1-6\epsilon}{2}\right)^{2k_1 \log mp} \geq (\log m)^{-0.5\delta/(1+\delta)} $ and $m\geq m_0(D,\delta,\epsilon)$.
Let $\alpha_{1}= \alpha_{0} \wedge \epsilon$. Taking a union bound over $d\in \{0,1,\ldots,\ell-1 \}$ and $\bss^d \in \{-1,1\}^{k'd}$, we have 
\begin{equation}
    \left|T_{\bss^{\ell}}^{\ell}(i, G) \cap T_{\bss^{\ell}}^{\ell}\left(i, G'\right)\right| \geq \left(\frac{mp}{2^{k'}}\right)^{\ell}(1-6\epsilon)^{k'\ell}
\end{equation}
with probability at least $1-m^{-D}$.
\end{proof}

\section{Proof of Lemma \ref{lem:wrong final}}\label{sec:proof of wrong final}

Consider a random subset $J$ uniformly drawn  from $\{-1,1\}^{k'\ell}$ with cardinality $2w$, where $w$ is an integer satisfying $2w > \log m$.   We define the event $\mathcal{H}$ to be true if and only if there exist positive constants $K$ and $K'$ such that the following conditions hold:
	\begin{align}
			&\label{eq:L max}\left(\max _{\bss^\ell \in J}\left|L_{\bss^\ell}\left(i, i'\right)\right|\right) \vee\left(\max _{\bss^\ell \in J}\left|L'_{\bss^\ell}\left(i, i'\right)\right|\right) \leq K \frac{(m p)^{\ell}}{2^{k'\ell}}, \\
			&\label{eq: L sum}\left(\sum_{\bss^\ell \in J}\left|L_{\bss^\ell}\left(i, i'\right)\right|\right) \vee\left(\sum_{\bss^\ell \in J}\left|L'_{\bss^\ell}\left(i, i'\right)\right|\right) \leq K \frac{(m p)^{\ell}}{2^{k'\ell}}\left(\frac{w}{m p(1-\alpha)^{\ell-1}}+\sqrt{w \log m}\right),
		\end{align}
		where
		\begin{align}
			&\label{eq:Ls}L_{\bss^\ell}\left(i, i'\right):=T^{\ell}_{\bss^\ell}(i) \cap \mathcal{B}_{G_{0}(C_{k})}\left(i', \ell\right) , \\
			&\label{eq:Ls'}L'_{\bss^\ell}\left(i, i'\right):=T'^{\ell}_{\bss^\ell}\left(i'\right) \cap \mathcal{B}_{G_{0}(C_{k})}(i, \ell) 
		\end{align}
  and 
  	\begin{equation}\label{eq:zeta max}
			\left(\max _{\bss^\ell \in J}\left|\zeta_{\bss^\ell}\left(i, i'\right)\right|\right) \vee\left(\max _{\bss^\ell \in J}\left|\zeta'_{\bss^\ell}\left(i, i'\right)\right|\right) \leq K' \frac{(m p)^{(\ell+1) / 2}}{2^{k'\ell / 2}} \sqrt{\log m},
		\end{equation}
		where
		\begin{equation}\label{eq:zeta def}
			\zeta_{\bss^\ell}\left(i, i'\right):=\sum_{j \in L_{\bss^\ell}\left(i, i'\right)}\left(\operatorname{deg}_{G(C_{k})}(j)-1-m p\right) \quad \text { and } \quad \zeta'_{\bss^\ell}\left(i, i'\right):=\sum_{j \in L_{\bss^\ell}'\left(i, i'\right)}\left(\operatorname{deg}_{G'(C_{k})}(j)-1-m p\right) .
		\end{equation}
Let us define the four conditions $(D1)-(D4)$ as follows.
\begin{flalign*}
    &(D1) \; G_{0}\left(\mathcal{B}_{G_{0}(C_{k})}(i, \ell) \cup \mathcal{B}_{G_{0}(C_{k})}\left(i', \ell\right)\right) \text{ forms a tree or a forest of two trees.}\\
    &(D2) \; \left|\mathcal{B}_{G_{0}(C_{k})}(i, \ell)\right|+\left|\mathcal{B}_{G_{0}(C_{k})}\left(i', \ell\right)\right| \leq m^{0.1}.\\
    &(D3) \; \text{Event } \mathcal{H} \text{ holds.}\\
    &(D4) \; \left|T^{\ell}_{\bss^{\ell}}(i)\right| \wedge\left|T'^{\ell}_{\bss^{\ell}}\left(i'\right)\right| \geq(m p / 2^{k'})^{\ell}(1-6\epsilon)^{k'\ell} ,\text{ for a sufficiently small positive } \epsilon \text{ and for all } \bss^\ell \in\{-1,1\}^{k'\ell}.  &&
\end{flalign*} 
We will show that there exists a sufficiently large vertex set $\mathcal{I} \subset C_k$ such that conditions $(D1)-(D4)$ hold for every pair $i,i'\in \mathcal{I}$, $i\neq i'$, and for all such wrong pairs \eqref{eq:wrong result} holds with high probability. The conditions $(D1),(D2)$ and $(D4)$ will be shown to be held with high probability by using the analysis from Sec. \ref{sec:vertex classes} and Sec.\ref{sec:proof of B5}, and the condition $(D3)$ will be proven in Sec. \ref{sec:proof of D3}.


\subsection{Proof of Lemma \ref{lem:wrong final}}\label{sec:proof of wrong final-1}

We consider a fixed distinct pair $i,i'\in C_k$ and a subset $J$ uniformly drawn  from $\{-1,1\}^{k'\ell}$ with $|J|=2w$. In addition, we condition on the subgraphs 
$$	G_{0}\left(\mathcal{B}_{G_{0}(C_{k})}(i, \ell) \cup \mathcal{B}_{G_{0}(C_{k})}\left(i', \ell\right)\right),  G\left(\mathcal{B}_{G_{0}(C_{k})}(i, \ell) \cup \mathcal{B}_{G_{0}}(C_{k})\left(i', \ell\right)\right),  G'\left(\mathcal{B}_{G_{0}(C_{k})}(i, \ell) \cup \mathcal{B}_{G_{0}(C_{k})}\left(i', \ell\right)\right)
	$$ 
as well as every edge between
$$
\mathcal{S}_{G_0(C_k)}(i',\ell) \cap \mathcal{B}_{G_0(C_k)}(i,\ell) \text{ and } \mathcal{S}_{G_0(C_k)}(i',\ell+1) 
$$
 in $G_0(C_k)$, $G(C_k)$ and $G'(C_k)$, and every edge between
  $$
\mathcal{S}_{G_0(C_k)}(i,\ell) \cap \mathcal{B}_{G_0(C_k)}(i',\ell) \text{ and } \mathcal{S}_{G_0(C_k)}(i,\ell+1) 
$$
 in graphs $G_0(C_k)$, $G(C_k)$ and $G'(C_k)$
such that conditions $(D1)-(D4)$ hold. Under these conditions, $\zeta_{\bss^\ell}(i,i')$ and $\zeta_{\bss^\ell}(i',i)$ are deterministic quantities defined in \eqref{eq:zeta def}. Furthermore, we assume that the remaining edges are randomly generated.
Denote the conditional probability and expectation as $\overline{\mathbb{P}}$ and $\overline{\mathbb{E}}$, respectively. By conditions $(D1)$ and $(D2)$,  we can see that $\operatorname{deg}_{G(C_{k})}(j)-1$  follows a binomial distribution with parameters $\bar{m}$ and $p$, where 
$$
\bar{m}:=m-\left|\mathcal{B}_{G_{0}(C_{k})}(i, \ell) \cup \mathcal{B}_{G_{0}(C_{k})}(i', \ell)\right| \geq m-m^{0.1},
$$
for any $j\in \mathcal{S}_{G_0(C_{k})}(i,\ell) \backslash \mathcal{B}_{G_0(C_{k})}(i',\ell)$. The same statement holds for $\operatorname{deg}_{G'(C_{k})}(j)-1$. Furthermore, it is independent for each $j\in \mathcal{S}_{G_0(C_{k})}(i,\ell) \backslash \mathcal{B}_{G_0(C_{k})}(i',\ell)$. Since vertices $i$ and $i'$ are fixed, we omit $(i,i')$ from $\zeta_{\bss^\ell}(i,i')$ and $ \zeta'_{\bss^\ell}(i,i')$ defined in \eqref{eq:zeta def} for notational simplicity. We also use $L_{\bss^\ell}$ and $L'_{\bss^\ell}$ for the quantities defined in \eqref{eq:Ls} and \eqref{eq:Ls'}.
	
\begin{lemma}[Similar to Lemma $6.3$ in \citep{MRT21a}]\label{lem:entrywise wrong}
     For each $\bss^\ell \in J$, we can express the difference between $f_{\bss^\ell}(i)$ and $f'_{\bss^\ell}\left(i'\right)$ as
    $$
    f_{\bss^\ell}(i)-f'_{\bss^\ell}\left(i'\right)=Z_{\bss^\ell}+\Delta_{\bss^\ell}
    $$
    where $Z_{\bss^\ell}$ is a random variable and $\Delta_{\bss^\ell}$  is a deterministic value that satisfy the following:
    \begin{itemize}
        \item $\overline{\mathbb{E}}\left[Z_{\bss^\ell}\right]=0$
        \item $\mathrm{v}_{\bss^\ell}(i)+\mathrm{v}'_{\bss^\ell}\left(i'\right)-2 m^{0.2} p-p(1-p) \bar{m}\left(\left|L_{\bss^\ell}\right|+\left|L'_{\bss^\ell}\right|\right) \leq \overline{\mathbb{E}}\left[Z_{\bss^\ell}^{2}\right] \leq \mathrm{v}_{\bss^\ell}(i)+\mathrm{v}'_{\bss^\ell}\left(i'\right)$;
        \item $\overline{\mathbb{P}}\left\{\left|Z_{\bss^\ell}\right| \geq t\right\} \leq 2 \exp \left(\frac{-t^{2} / 2}{E\left[Z_{\bss^\ell}^{2}\right]+t / 3}\right)$
        \item $\left|\Delta_{\bss^\ell}\right| \leq\left|\zeta_{\bss^\ell}\right|+\left|\zeta'_{\bss^\ell}\right|+2 m^{0.2} p$
    \end{itemize}		
  Furthermore, the random variables $Z_{\bss^\ell}$ are conditionally independent for different $\bss^\ell \in J$.
\end{lemma}
\begin{proof}
         The proof of Lemma \ref{lem:entrywise wrong} follows a similar approach to that of Lemma 6.3 in \citep{MRT21a}, so we omit the detailed proof here. The key observation is that we can express $f_{\bss^\ell}(i)-f'_{\bss^\ell}(i')$ as the sum of $Z_{\bss^\ell}$ and $\Delta_{\bss^\ell}$, where
$$
\begin{aligned}
Z_{\bss^\ell}:= & \sum_{j \in T^{\ell}_{\bss^\ell}(i) \backslash \mathcal{B}_{G_{0}(C_{k})}(i',\ell)}\left(\operatorname{deg}_{G(C_{k})}(j)-1-\bar{m} p\right) \\
& -\sum_{j \in\ T'^{\ell}_{\bss^\ell}(i')\backslash \mathcal{B}_{G_{0}(C_{k})}(i,\ell)  }\left(\operatorname{deg}_{G'(C_{k})}(j)-1-\bar{m} p\right)
\end{aligned}
$$
and
$$
\Delta_{\bss^\ell}:=\zeta_{\bss^\ell}-\zeta'_{\bss^\ell}+(\bar{m}-m) p\left(\left|T^{\ell}_{\bss^\ell}(i)\right|-\left|T'^{\ell}_{\bss^\ell}(i')\right|\right) .
$$
\end{proof}

	\begin{lemma}[Lower bound on the normalized distance of sparsified signature vectors for a wrong pair]\label{lem:wrong}
		For any constants $C, D, K, K',\delta>0$, there exist constants $m_{0},\epsilon>0$ with the properties below. Let $m\geq m_0$, $\alpha \in(0, \epsilon)$, $w \geq\left\lfloor(\log m)^{5}\right\rfloor $, and
    \beq
    \begin{split}\nonumber
    (\log m)^{1+\delta} \leq m \frac{p}{1-\alpha} \leq m^{1/20}, \; k'\ell \leq C \log \log m.
    \end{split}
    \eeq
		Furthermore, let conditions $(D1)-(D4)$ hold with constants $K, K', \epsilon>0$, and consider a subset $J$ which has cardinality $2w$. Then, with a conditional probability larger than $1-m^{-D}$, the following statement holds:
		$$
		\sum_{\bss^\ell \in J} \frac{\left(f_{\bss^\ell}(i)-f'_{\bss^\ell}\left(i'\right)\right)^{2}}{\mathrm{v}_{\bss^\ell}(i)+\mathrm{v}'_{\bss^\ell}\left(i'\right)} \geq 2 w\left(1-\frac{1}{(\log m)^{0.9}}\right) .
		$$
	\end{lemma}
 Lemma \ref{lem:wrong} will be proven in Section \ref{sec:proof of lemma wrong}.
\begin{proof}[Proof of Lemma \ref{lem:wrong final}]
    First, we will establish the existence of a subset $\mathcal{I} \subset C_k$ of cardinality $|\mathcal{I}|>m-m^{1-c}$ such that for any $i, i' \in \mathcal{I}$, $i\neq i'$, conditions $(D1)-(D4)$ are satisfied with a probability of at least $1 - m^{-D-1}$.
    \begin{itemize}
	\item Using Proposition \ref{prop:type}, we can deduce that with a probability of at least $1 - m^{-D-2}$ (by choosing a sufficiently large constant $R$), there exists a subset $\operatorname{Typ}_{G_0(C_k)}(3\ell,\epsilon)$ such that $|\operatorname{Typ}_{G_0(C_k)}(3\ell,\epsilon)| \geq m - m^{1-c_1}$, where $c_1 \in (0, 0.5)$ is a constant. For any distinct vertices $i$ and $i'$ in $\operatorname{Typ}_{G_0(C_k)}(3\ell,\epsilon)$, condition $(D1)$ holds because if $\operatorname{dist}_{G_0(C_k)}(i,i')\leq 2\ell $, $G_{0}\left(\mathcal{B}_{G_{0}(C_{k})}(i, \ell) \cup \mathcal{B}_{G_{0}(C_{k})}\left(i', \ell\right)\right)$ becomes a forest, and if $\operatorname{dist}_{G_0(C_k)}(i,i') > 2\ell $, $G_{0}\left(\mathcal{B}_{G_{0}(C_{k})}(i, \ell) \cup \mathcal{B}_{G_{0}(C_{k})}\left(i', \ell\right)\right)$ becomes two trees.
	
        \item By Lemma \ref{lem:Sizes of neighborhoods and their intersections}, we have $|\mathcal{B}_{G_{0}(C_{k})}(i,\ell)| =O\left( (mp)^{\ell}\right) \leq \frac{1}{2}m^{0.1}$,  $\forall i\in C_k$ and  $\ell \leq \frac{\log m}{11 \log mp}$, with probability at least $1-m^{-D-2}$. Thus, the condition $(D2)$ holds for all $i\in C_k$ with probability at least $1-m^{-D-2}$.

        \item We can assert that the condition $(D3)$ holds with $\mathcal{I}_1 \subset C_k$ of cardinality at least $m-m^{1-c_2}$ for a constant $c_2\in(0,0.5)$ with probability at least $1-m^{-D-2}$ by Lemma \ref{lem:condition E3}.

        \item By Proposition \ref{prop:type}, $|\operatorname{Typ}_{G_0(C_k)}(\ell,\epsilon)|\geq m-m^{1-c_3}$ for a constant $c_3 \in (0,0.5)$ with probability at least $1-m^{-D-2}$. Thus, the condition $(D4)$ holds for at least $m-m^{1-c_3}$ vertices $i \in C_k$ with probability at least $1-2m^{-D-2}$ by Lemma \ref{prop:common Ts} since  $\left|T^{\ell}_{\bss^{\ell}}(i)\right   | \geq  \left|T_{\bss^{\ell}}^{\ell}(i, G) \cap T_{\bss^{\ell}}^{\ell}\left(i, G'\right)\right|$.
                
    \end{itemize}
The proof is complete by applying Lemma \ref{lem:wrong} .
\end{proof}

\subsection{Proof of Lemma \ref{lem:wrong}}\label{sec:proof of lemma wrong}

\begin{proof}[Proof of Lemma \ref{lem:wrong}]
    For fixed distinct vertices $i,i' \in C_k$, by condition $(D4)$, we have 
\begin{equation}
    |T^{\ell}_{\bss^\ell}(i)| \wedge |T'^{\ell}_{\bss^\ell}(i')| \geq (mp/2^{k'})^{\ell}(1-6\epsilon)^{k'\ell}
\end{equation}
 for any $\bss^\ell \in\{-1,1\}^{k'\ell}$.
	Thus, we obtain
	\begin{equation}\label{eq:vs wrong}
 \begin{aligned}
	    \mathrm{v}_{\bss^\ell}(i) &=m p(1-p)\left|T^{\ell}_{\bss^\ell}(i)\right| \geq m p(1-p)\left(|T^{\ell}_{\bss^\ell}(i)|\wedge |T'^{\ell}_{\bss^\ell}(i') | \right) \\
		& \geq (1-p)\left(\frac{1-6\epsilon}{2}\right)^{k'\ell}(mp)^{\ell+1}.
	\end{aligned}
	\end{equation}
Similarly, we obtain 
 \begin{equation}\label{eq:v's wrong}
 \begin{aligned}
	    \mathrm{v}'_{\bss^\ell}(i') \geq (1-p)\left(\frac{1-6\epsilon}{2}\right)^{k'\ell}(mp)^{\ell+1}.
	\end{aligned}
	\end{equation}
 One can see from Lemma \ref{lem:entrywise wrong} that for a wrong pair of vertices the difference  between the entries of signatures $f_{\bss^\ell}(i)$ and $f'_{\bss^\ell}(i')$ at some $\bss^\ell\in\{-1,1\}^{k'\ell}$ can be decomposed into the  random variable part $Z_{\bss^\ell}$ and the deterministic part $\Delta_{\bss^\ell}$. Define
 $$
	X_{\bss^\ell}:=\frac{f_{\bss^\ell}(i)-f'_{\bss^\ell}\left(i'\right)}{\sqrt{\mathrm{v}_{\bss^\ell}(i)+\mathrm{v}'_{\bss^\ell}\left(i'\right)}}=\frac{Z_{\bss^\ell}+\Delta_{\bss^\ell}}{\sqrt{\mathrm{v}_{\bss^\ell}(i)+\mathrm{v}'_{\bss^\ell}\left(i'\right)}},
	$$
 where $Z_{\bss^\ell}$ and $\Delta_{\bss^\ell}$ are defined in Lemma \ref{lem:entrywise wrong}.	
 
 We will first estimate expectation and variance of $X_{\bss^\ell}$. By Lemma \ref{lem:entrywise wrong}, \eqref{eq:zeta max}, and \eqref{eq:vs wrong}, we can have
	\begin{equation}\label{eq:E_var}
		\left|\overline{\mathbb{E}}\left[X_{\bss^\ell}\right]\right|=\frac{\left|\Delta_{\bss^\ell}\right|}{\sqrt{\mathrm{v}_{\bss^\ell}(i)+\mathrm{v}'_{\bss^\ell}\left(i'\right)}} \leq \frac{2 K' \frac{(m p)^{(\ell+1) / 2}}{2^{k'\ell / 2}} \sqrt{\log m}+2 m^{0.2} p}{\sqrt{2(1-p)(m p)^{\ell+1}((1-6\epsilon)/2)^{k'\ell}}} \leq \frac{6 K' \sqrt{\log m}}{(1-6\epsilon)^{k'\ell / 2}}
	\end{equation}
	
	and
	\begin{equation}\label{eq:Var Xs wrong}
	    \begin{aligned}
		\operatorname{Var}\left(X_{\bss^\ell}\right)=\frac{\overline{\mathbb{E}}\left[Z_{\bss^\ell}^{2}\right]}{\mathrm{v}_{\bss^\ell}(i)+\mathrm{v}'_{\bss^\ell}\left(i'\right)} & \geq \frac{\mathrm{v}_{\bss^\ell}(i)+\mathrm{v}'_{\bss^\ell}\left(i'\right)-2 m^{0.2} p-p(1-p) \bar{m}\left(\left|L_{\bss^\ell}\right|+\left|L'_{\bss^\ell}\right|\right)}{\mathrm{v}_{\bss^\ell}(i)+\mathrm{v}'_{\bss^\ell}\left(i'\right)} \\
		& \geq 1-\frac{2 m^{0.2} p+p(1-p) \bar{m}\left(\left|L_{\bss^\ell}\right|+\left|L'_{\bss^\ell}\right|\right)}{2(1-p)(mp)^{\ell+1}((1-6\epsilon) /2)^{k'\ell}},
	\end{aligned}
	\end{equation}
	for any $\bss^\ell \in J$. Moreover, we have that  $\overline{\mathbb{P}}\left\{\left|X_{\bss^\ell}-\overline{\mathbb{E}}\left[X_{\bss^\ell}\right]\right| \geq t\right\} \leq 2 \exp \left(\frac{-t^{2} / 2}{1+t / 3}\right)$ since $\overline{\mathbb{P}}\left\{\left|Z_{\bss^\ell}\right| \geq t\right\} \leq 2 \exp \left(\frac{-t^{2} / 2}{E\left[Z_{\bss^\ell}^{2}\right]+t / 3}\right)$  and $\operatorname{Var}\left(X_{\bss^\ell}\right)=\frac{\overline{\mathbb{E}}\left[Z_{\bss^\ell}^{2}\right]}{\mathrm{v}_{\bss^\ell}(i)+\mathrm{v}'_{\bss^\ell}\left(i'\right)} \leq 1$ hold from Lemma \ref{lem:entrywise wrong}. Therefore, applying Lemma \ref{lem:hoeffding's with truncation}, we can obtain
	\begin{equation}\label{eq:Xs bound wrong}
	    \sum_{\bss^\ell \in J} X_{\bss^\ell}^{2} \geq \sum_{\bss^\ell \in J} \overline{\mathbb{E}}\left[X_{\bss^\ell}^{2}\right]-C_{1} \sqrt{w}(\log m)^{3 / 2}-C_{1}\left(\max _{\bss^\ell \in J}\left|\overline{\mathbb{E}}\left[X_{\bss^\ell}\right]\right|\right)(\sqrt{w \log m}+\log m)
	\end{equation}
	 with conditional probability at least $1-m^{-D}$ for $J\subset \{-1,1\}^{k'\ell}$ such that $|J|=2w$ and for a sufficiently large constant $C_{1}$ depending on $D$. By \eqref{eq: L sum} in condition $(D3)$, we can get
	\begin{equation}\label{eq:Xs bound wrong-1}
	    \begin{aligned}
		\sum_{\bss^\ell \in J} \overline{\mathbb{E}}\left[X_{\bss^\ell}^{2}\right] & \geq  \sum_{\bss^\ell \in J} \operatorname{Var}(X_{\bss^\ell}) \geq 2 w-\sum_{\bss^\ell \in J} \frac{2 m^{0.2} p+p(1-p) \bar{m}\left(\left|L_{\bss^\ell}\right|+\left|L'_{\bss^\ell}\right|\right)}{2(1-p)(m p)^{\ell+1}((1- 6\epsilon)/2)^{k'\ell}} \\
		& \geq 2 w-m^{-0.7}-\frac{p(1-p) \bar{m}}{2(1-p)(mp)^{\ell+1}((1-6\epsilon)/2)^{k'\ell}} \cdot 2K \frac{(m p)^{\ell}}{2^{k'\ell}}\left(\frac{w}{m p(1-\alpha)^{\ell-1}}+\sqrt{w \log m}\right) \\
		& \geq 2 w\left(1-\frac{K_{1}}{(1-7\epsilon)^{k'\ell} \log m}\right)
	\end{aligned}
	\end{equation}
	for a sufficiently large constant $K_{1}$ depending on $K$, if $w \geq(\log m)^3, \; \alpha<\epsilon$, $mp/(1-\alpha) \geq (\log m)^{1+\delta}$ and $m\geq m_0(\epsilon)$. Finally, by \eqref{eq:E_var}, \eqref{eq:Xs bound wrong} and \eqref{eq:Xs bound wrong-1}, we can obtain
	$$
	\begin{aligned}
		\sum_{\bss^\ell \in J} X_{\bss^\ell}^{2} & \geq 2 w\left(1-\frac{K_{1}}{(1-7\epsilon)^{k'\ell} \log m}\right)-C_{1} \sqrt{w}(\log m)^{3 / 2}-C_{1} \frac{6 K' \sqrt{\log m}}{(1-6\epsilon)^{k'\ell / 2}}(\sqrt{w \log m}+\log m) \\
		& \stackrel{(a)}{\geq} 2 w\left(1-\frac{K_{2}}{(1-7\epsilon)^{k'\ell} \log m}\right) \geq 2w\left(1-\frac{1}{(\log m)^{0.9}}\right).
	\end{aligned}
 $$
	The inequality $(a)$ holds by $w \geq\left\lfloor(\log m)^{5}\right\rfloor$ and the last inequality holds since $k'\ell \leq C \log \log m$ and $\epsilon>0$ is small enough which depends on $C$. This completes the proof.
	\end{proof}

	\subsection{Proof of Condition $(D3)$}\label{sec:proof of D3}
	As in the analysis for correct pairs, we consider the sparsification for a distinct pair of vertices $(i,i')$ in the lemma below (similar to Lemma $6.1$ in \citep{MRT21a}). 
	\begin{lemma}\label{lem:sparsification of wrong}
		 For constants $K_{1}, D>0$, there exists $K>0$ with the properties below. 
		 Consider a random subset $J$ uniformly drawn  from $\{-1,1\}^{k'\ell}$ with cardinality $2w$, where $w$ is an integer satisfying $2w \geq \log m$.   
		 Then, for any distinct vertices $i, i' \in C_{k}$ if $(i, i')$ satisfies that
  \begin{flalign*} 
&(E1) \;  G_{0}\left(\mathcal{B}_{G_{0}(C_{k})}(i, \ell)\right) \text{ and }G_{0}\left(\mathcal{B}_{G_{0}(C_{k})}\left(i', \ell\right)\right) \text{ form trees;}\\
&(E2) \; \left|T^{\ell}_{\bss^\ell}(i)\right| \vee\left|T'^{\ell}_{\bss^\ell}\left(i'\right)\right| \leq 6\left(\frac{m p}{2^{k'}}\right)^{\ell} \text{ for all }\bss^\ell \in\{-1,1\}^{k'\ell}; \\
&(E3) \; \left|\mathcal{B}_{G_{0}(C_{k})}(i, \ell) \cap \mathcal{B}_{G_{0}(C_{k})}\left(i', \ell\right)\right| \leq K_{1}\left(\frac{mp}{1-\alpha}\right)^{\ell-1}, &&
\end{flalign*} 
		then $L_{\bss^\ell}\left(i, i'\right)$ and $L'_{\bss^\ell}\left(i, i'\right)$ defined in \eqref{eq:Ls} and \eqref{eq:Ls'} satisfy \eqref{eq:L max} and \eqref{eq: L sum} with probability at least $1-m^{-D}$.
	\end{lemma}

\begin{proof}
For fixed vertices $i,i' \in C_k$, $i\neq i' $, let
	$$
	a_{\bss^\ell}:=\left|L_{\bss^\ell}\left(i, i'\right)\right| ,\;a'_{\bss^\ell}:=\left|L'_{\bss^\ell}\left(i, i'\right)\right|.
	$$
	Since  $L_{\bss^\ell} (i,i') \subset T^\ell _{\bss^\ell} (i)$, $L'_{\bss^\ell} (i,i') \subset T'^\ell _{\bss^\ell} (i)$ and from the condition $(E2)$, we get
	\begin{equation}
		a_{\bss^\ell} \vee a'_{\bss^\ell} \leq 6\left(\frac{m p}{2^{k'}}\right)^{\ell}.
	\end{equation}
        Therefore, \eqref{eq:L max} holds.  By conditions $(E1)$ and $(E3)$, we can obtain that 
	$$
	\sum_{\bss^\ell \in\{-1,1\}^{k'\ell}} a_{\bss^\ell} \leq\left|\mathcal{B}_{G_{0}(C_{k})}(i, \ell) \cap \mathcal{B}_{G_{0}(C_{k})}\left(i', \ell\right)\right| \leq K_{1}\left(\frac{mp}{1-\alpha}\right)^{\ell-1} .
	$$
	For a random subset $J$ uniformly drawn  from $\{-1,1\}^{k'\ell}$ with cardinality $2w\geq 2\log m$, by applying Lemma \ref{lem:bernstein's inequality for sampling without replacement} we have
	\begin{equation}\label{eq:as}
	    \begin{aligned}
		\sum_{\bss^\ell \in J} a_{\bss^\ell} & \leq \frac{|J|}{2^{k'\ell}} \sum_{\bss^\ell \in\{-1,1\}^{k'\ell}} a_{\bss^\ell}+K_{2} \sqrt{\frac{|J|}{2^{k'\ell}} \sum_{\bss^\ell \in\{-1,1\}^{k'\ell}} a_{\bss^\ell}^{2} \cdot \log m}+K_{2} \max _{\bss^\ell \in\{-1,1\}^{k'\ell}}\left|a_{\bss^\ell}\right| \cdot \log m \\
		& \leq K_{3} \frac{|J|}{2^{k'\ell}}\left(\frac{mp}{1-\alpha}\right)^{\ell-1}+K_{3} \frac{(m p)^{\ell}}{2^{k'\ell}} \sqrt{|J| \cdot \log m}+K_{3} \frac{(m p)^{\ell}}{2^{k'\ell}} \log m \\
		& \leq 3 K_{3} \frac{(m p)^{\ell}}{2^{k'\ell}}\left(\frac{w}{m p(1-\alpha)^{\ell-1}}+\sqrt{w \log m}\right),
	\end{aligned}
	\end{equation}
     with probability at least $1-m^{-D-3}$ for sufficiently large constants $K_{2}, K_{3}$ depending on $K_{1}$, and $D$. Similarly, we obtain
     \begin{equation}\label{eq:a's}
	    \begin{aligned}
		\sum_{\bss^\ell \in J} a'_{\bss^\ell}  \leq 3 K_{3} \frac{(m p)^{\ell}}{2^{k'\ell}}\left(\frac{w}{m p(1-\alpha)^{\ell-1}}+\sqrt{w \log m}\right).
	\end{aligned}
	\end{equation}
    Hence, \eqref{eq: L sum} holds from \eqref{eq:as} and \eqref{eq:a's}. By applying a union bound over all distinct vertices pairs $i,i' \in C_k$, the proof is completed.
\end{proof}

	\begin{lemma}[Similar to Lemma 6.2 in \citep{MRT21a}]\label{lem:condition E3}
		 For any constants $D,\delta>0$, there exists $c \in(0,1 / 2)$ and $R,Q_1,Q_2,K, K',m_0$, which depend on $D,\delta$, with the properties below. 
		 Consider a random subset $J$ uniformly drawn  from $\{-1,1\}^{k'\ell}$ with cardinality $2w$, where $w$ is an integer satisfying $2w > \log m$.   
		Suppose that
		\begin{equation}
			m\geq m_0, \quad (\log m)^{1+\delta} \leq m p/(1-\alpha) \leq m^{1/20} , \quad mq\geq Q_1 k'^2\ell^2,\quad k'\leq Q_2 \log mp, \quad \ell \leq \frac{\log m}{R \log mp}.
		\end{equation}
        Then,  there exists a subset $\mathcal{I} \subset C_{k}$ with $|\mathcal{I}| \geq m-m^{1-c}$ such that the event $\mathcal{H}$ (or condition $(D3)$) holds for every pair $i,i'\in \mathcal{I}$, where $i\neq i'$, with probability at least $1-m^{-D}$.
	\end{lemma}
	
\begin{proof}
To apply Lemma \ref{lem:sparsification of wrong}, we first show that existence of a subset $\mathcal{I} \subset C_k$ with $|\mathcal{I}|\geq m-m^{1-c}$ such that conditions $(E1),(E2)$ and $(E3)$ are satisfied for all pairs $i,i'\in \mathcal{I}$, where $i\neq i'$.
\begin{itemize}
    \item  By Proposition \ref{prop:type}, there exists a subset $\operatorname{Typ}_{G_0(C_k)}(3\ell,\epsilon)$ such that $\left|\operatorname{Typ}_{G_0(C_k)}(3\ell,\epsilon)\right|\geq m-m^{1-c}$, where $c\in (0,0.5)$ with probability at least $1-m^{-D-2}$. Since $G_0 \left(\mathcal{B}_{G_0 (C_k)}(i,3\ell)\right)$ and $G_0 \left(\mathcal{B}_{G_0 (C_k)}(i',3\ell)\right)$ are trees for any $i,i'\in \operatorname{Typ}_{G_0(C_k)}(3\ell,\epsilon)$, the condition $(E3)$ holds with probability at least $1-m^{-D-2}$ for $K_1$ depending on $D$ by Lemma \ref{lem:Sizes of neighborhoods and their intersections}. (Since $3\ell$ is greater than or equal to 3, in order to apply Proposition \ref{prop:type}, we have set $mp/(1-\alpha) \leq m^{1/20}$.) 

    \item The condition $(E1)$ is directly satisfied since $G_0 \left(\mathcal{B}_{G_0 (C_k)}(i,3\ell)\right)$ and $G_0 \left(\mathcal{B}_{G_0 (C_k)}(i',3\ell)\right)$ form trees for any $ i,i'\in \mathcal{I}$.

    \item By Lemma \ref{lem:sizes of vertex class}, the condition $(E2)$ holds with probability at least $1-m^{-D-2}$ for any $i,i' \in \mathcal{I}$.
\end{itemize}
    Thus, by applying Lemma \ref{lem:sparsification of wrong}, we can have \eqref{eq:L max} and \eqref{eq: L sum} for all pairs $i,i'\in \mathcal{I}$, where $i\neq i'$, with probability at least $1-2m^{-D-1}$. 
    The verification of \eqref{eq:zeta max} follows a similar approach as presented in Lemma 6.2 of \citep{MRT21a}; thus, we omit the detailed explanation.
\end{proof}

	\section{Technical Tools} \label{sec:appx tool}
	\begin{theorem}[Tail bounds on binomial distribution \citep{Oka59}]\label{thm:binomial tail}
		Let $X \sim \operatorname{Bin}(n, p)$. It holds that
		$$
		\begin{aligned}
			&\mathbb{P}\{X \leq n t\} \leq \exp \left(-n(\sqrt{p}-\sqrt{t})^{2}\right), \quad \forall 0 \leq t \leq p \\
			&\mathbb{P}\{X \geq n t\} \leq \exp \left(-2 n(\sqrt{t}-\sqrt{p})^{2}\right), \quad \forall p \leq t \leq 1 \\
		\end{aligned}
		$$
	\end{theorem}
		
	\begin{lemma}[Lemma 22 in \citep{MRT21b}]\label{lem:binomial 1}
		Let $X\sim \operatorname{Bin}(n,p)$ with real number $p \in(0,1)$. For any $t, r>0$, it holds that
		$$
		\mathbb{P}\{X \in[t, t+r]\} \leq \frac{C r}{\sqrt{n p(1-p)}}
		$$
  with a universal constant $C>0$ .
	\end{lemma}

	\begin{lemma}[Lemma 23 in \citep{MRT21b}]\label{lem:binomial 2}
		 Let $X\sim \operatorname{Bin}(n, p)$ with real number $p \in(0,1)$. If $n p(1-p) \geq C$ then for any $r \geq 2$, it holds that
		$$
		\mathbb{P}\{X>n p-r\} \geq \frac{1}{2}+c\left(\frac{r}{\sqrt{n p(1-p)}} \wedge 1\right), \quad \mathbb{P}\{X<n p+r\} \geq \frac{1}{2}+c\left(\frac{r}{\sqrt{n p(1-p)}} \wedge 1\right) 
		$$
  with universal constants $C,c>0$.
	\end{lemma}
	   
         \begin{lemma}[Hoeffding's inequality]\label{lem:Hoeffding}
             Let $X_1, \ldots, X_n$ be independent random variables such that $a_{i}\leq X_{i} \leq b_{i}$. Define $S_{n}:=X_{1}+\ldots+X_{n}$. Then, for all $t>0$, we have
             \begin{equation}
                 \begin{aligned}
                     &\P\left\{S_{n}-E[S_{n}] \geq t  \right\} \leq  \exp \left(-\frac{2t^{2}}{\sum^{n}_{i=1}(b_{i}-a_{i})^{2}}\right),\\
                     &\P\left\{ \left|S_{n}-E[S_{n}] \right| \geq t  \right\} \leq 2 \exp \left(-\frac{2t^{2}}{\sum^{n}_{i=1}(b_{i}-a_{i})^{2}}\right).
                 \end{aligned}
             \end{equation}
         \end{lemma}

  \begin{lemma}[Bernstein's inequality]\label{lem:bernstein's inequality}
            Let $X_1, \ldots, X_n$ be independent random variables such that $a_{i}\leq X_{i} \leq b_{i}$. Define $\sigma^2:=\frac{1}{n} \sum_{i=1}^n \mathbb{E}\left[\left(X_i-\mathbb{E}\left[X_i\right]\right)^2\right]$. Then, for any $t>0$, we have
        $$
        \mathbb{P}\left\{\left|\sum_{i=1}^n\left(X_i-\mathbb{E}\left[X_i\right]\right)\right| \geq t\right\} \leq 2 \exp \left(\frac{-t^2 / 2}{\sigma^2 n+(b-a) t / 3}\right)
        $$
        \end{lemma}

	\begin{lemma}[Bernstein's inequality for sampling without replacement]\label{lem:bernstein's inequality for sampling without replacement}
  Let $\left\{x_{1}, \ldots, x_{N}\right\}$ be given for a finite positive integer $N$. For a positive integer $n<N$, let $X_{1}, \ldots, X_{n}$ be drawn from $\left\{x_{1}, \ldots, x_{N}\right\}$ without replacement. Define
		$$
		a:=\min _{i \in[N]} x_{i}, \quad b:=\max _{i \in[N]} x_{i}, \quad \bar{x}:=\frac{1}{N} \sum_{i=1}^{N} x_{i}, \quad \sigma^{2}:=\frac{1}{N} \sum_{i=1}^{N}\left(x_{i}-\bar{x}\right)^{2} .
		$$
		Then, for any $t>0$, we have
		$$
		\mathbb{P}\left\{\left|\sum_{i=1}^{n} X_{i}-n \bar{x}\right| \geq t\right\} \leq 2 \exp \left(\frac{-t^{2} / 2}{\sigma^{2} n+(b-a) t / 3}\right) .
		$$
	\end{lemma}

	\begin{lemma}[Lemma 13 in \citep{MRT21b}]\label{lem:correlation of neighbor}
		Fix $k \in \mathbb{N}_{+}, p \in(0,1)$ and $\alpha \in[0,1-p]$. Let $\Gamma_{0}\sim \mathcal{G} \left(k, \frac{p}{1-\alpha}\right)$. We can get two subgraphs $\Gamma$ and $\Gamma'$ by independently subsampling every edge of $\Gamma_0$ with probability $1-\alpha$. Fix subsets $J, J' \subset[k]$ and a vertex $i \in$ $[k] \backslash\left(J \cup J'\right)$. For any $t>0$, we have that with probability at least $1-6 \exp (-t)-2 \exp \left(\frac{-p\left|J \cap J'\right|}{3(1-\alpha)}\right)$,
		$$
		\left|| \mathcal{N}_{\Gamma}(i ; J)|-p| J|-| \mathcal{N}_{\Gamma'}\left(i ; J'\right)|+p| J'|\right| \leq 4\left(t+\sqrt{t \alpha p\left|J \cap J'\right|}+\sqrt{t p\left|J \triangle J'\right|}\right).
		$$
	\end{lemma}

\begin{lemma}[Lemma 5.5 in \citep{MRT21a}]\label{lem:sparsification lem}
     Given a constant $S>0$ and an even integer $k \in \mathbb{N}$, let $\Omega$ and $\Omega'$ be two finite sets that can be expressed as $\Omega=\bigcup_{i=1}^{k} \Omega_{i}$ and $\Omega'=\bigcup_{i=1}^{k} \Omega_{i}'$ where $\left|\Omega'_{i}\right| \leq S / k$ for all $i \in [k]$. Additionally, let $w \in \{2,3,\ldots,k/2\}$, and $I$ be a random subset uniformly drawn  from $[k]$ with a cardinality of $2w$. Then, for any given values of $L \geq 1$ and $\rho \in (0,1/4)$, where $\rho w$ is an integer, we have that
    $$  
		\mathbb{P}\left\{\mid\left\{i \in I: \exists j \in I \backslash\{i\} \text { s.t. }\left|\Omega_{i} \cap \Omega'_{j}\right| \geq L S / k^{2}\right\} \mid \geq 2 \rho w\right\} \leq\left(\frac{8 w^{3}}{L}\right)^{\rho w} .
		$$
	\end{lemma}

	\begin{lemma}[Hoeffding's inequality with truncation (Lemma 5.10 in \citep{MRT21a})]\label{lem:hoeffding's with truncation}
        Let $X_{1}, \ldots, X_{N}$ be independent random variables such that $\left|\mathbb{E}\left[X_{i}\right]\right| \leq \tau$ for $\tau>0$, and that
        $$
        \mathbb{P}\left\{\left|X_{i}-\mathbb{E}\left[X_{i}\right]\right| \geq t\right\} \leq 2 \exp \left(\frac{-c t^{2}}{1+t}\right), \quad \forall t>0,
        $$
        for a constant $c>0$ for each $i \in[N]$. Then there exists a constant $C>0$ depending only on $c$ such that, for any $\epsilon \in(0,0.1)$,
        $$
        \mathbb{P}\left\{\left|\sum_{i=1}^{N}\left(X_{i}^{2}-\mathbb{E}\left[X_{i}^{2}\right]\right)\right| \geq C \log (N / \epsilon) \sqrt{N \log (1 / \epsilon)}+C \tau(\sqrt{N \log (1 / \epsilon)}+\log (1 / \epsilon))\right\} \leq \epsilon .
        $$
	\end{lemma}

        \begin{lemma}\label{lem:1+n2x}
            Let $n$  be a positive integer and $x$ be a positive real number. If $nx\leq 1$, then $(1+x)^{n} \leq 1+2nx$.
        \end{lemma}
	\begin{proof}
  For $k=1,\ldots,n-1$, we have
 \begin{equation}
     {n \choose k+1}x^{k+1}={n \choose k}x^{k} \frac{(n-k)x}{k+1} \leq \frac{1}{2}{n \choose k}x^{k}
 \end{equation}
	since $nx\leq 1$.
 Thus, we get 
 \begin{equation}
     (1+x)^{n}=1+\sum^{n}_{k=1} {n \choose k}x^{k} \leq 1+ 2nx.
 \end{equation}
\end{proof}

	\section{Refinement Matching Algorithm }\label{app:sec:exact}

In this section, we summarize the refinement matching algorithm from \citet{MRT21a}.

\begin{algorithm}[thb]
    \caption{RefinementMatching \citep{MRT21a}}\label{alg:algorithm4}
    \begin{algorithmic}
        \STATE {\bfseries Input:} \text{two graphs $\Gamma$ and $\Gamma'$ on $[n]$, a permutation $\tilde{\pi}:[n] \rightarrow [n]$, and a parameter $\epsilon>0$} 
        \STATE {\bfseries Output:} \text{a permutation $\hat{\pi} : [n] \rightarrow [n]$}
        \STATE \text{$\pi_{0}\leftarrow \tilde{\pi}$}
        \FOR{$t=1,\ldots,\left\lceil \log_{2}n \right\rceil$}
        \FOR{$i=1,\ldots,n$}
        \IF{there is a vertex $i'\in[n]$ such that}
        \STATE $\bullet$ $|\pi^{-1}_{t-1}\left(\caN_{\Gamma}(i) \right)\cap \caN_{\Gamma'}(j')| \geq \epsilon^{2}pn/512$ for all $j'\in[n] \backslash \{i'\}$
        \STATE $\bullet$  $|\pi^{-1}_{t-1}\left(\caN_{\Gamma}(i) \right)\cap \caN_{\Gamma'}(j')| \geq \epsilon^{2}pn/512$ for all $j'\in[n] \backslash \{i'\}$
        \STATE $\bullet$  $|\pi^{-1}_{t-1}\left(\caN_{\Gamma}(j) \right) \cap \caN_{\Gamma'}(i')| \geq \epsilon^{2}pn/512$ for all $j\in [n]\backslash \{i\}$
        \ENDIF
        \ENDFOR
        \STATE extend $\pi_{t}$ to a permutation on $[n]$ in an arbitrary manner
        \ENDFOR
        \STATE $\hat{\pi}$ $\leftarrow \pi_{\left\lceil \log_{2} n \right\rceil}$
        \STATE \textbf{return} $\hat{\pi}$
    \end{algorithmic}
\end{algorithm}